\documentclass[12pt]{article}
\usepackage{a4wide}

\usepackage{soul}
\usepackage{url}
\usepackage[utf8]{inputenc}
\usepackage{graphicx}
\usepackage{amsmath}
\usepackage{amsthm}
\usepackage{amssymb}
\usepackage{thmtools,thm-restate}

\usepackage{graphics}
\usepackage{tikz}
\usepackage{enumerate}

\usepackage{wrapfig}

\newcommand{\eq}{\leftrightarrow}

\newcommand{\imp}{\rightarrow}

\newcommand{\et}{\wedge}
\newcommand{\vel}{\vee}
\newcommand{\Et}{\bigwedge}
\newcommand{\Vel}{\bigvee}

\newcommand{\all}{\forall}
\newcommand{\is}{\exists}

\renewcommand{\phi}{\varphi}
\newcommand{\union}{\cup}
\newcommand{\Union}{\bigcup}
\newcommand{\inter}{\cap}

\newtheorem{theorem}{Theorem}
\newtheorem{corollary}[theorem]{Corollary}
\newtheorem{proposition}[theorem]{Proposition}

\newtheorem{example}[theorem]{Example}
\newtheorem{lemma}[theorem]{Lemma}

\newcommand{\cs}{\mathbf{s}}
\newcommand{\cstwo}{\mathbf{t}}
\newcommand{\CS}{\mathbf{S}}
\newcommand{\CStwo}{\mathbf{R}}

\newcommand{\Kw}{\mathit{Kw}}

\newcommand{\M}{\hat{K}}

\usepackage[ruled,vlined,linesnumbered,lined]{algorithm2e}
\usepackage{multirow}

\newcommand{\powerset}{\mathcal P}

\newcommand{\uu}{\mathbf{u}}
\newcommand{\weg}[1]{}

\newcommand{\nnf}{\mathsf{nnf}}

\newcommand{\obar}[1]{\overline{#1}}
\newcommand{\boolgame}{B}
\newcommand{\obsgame}{G}

\newcommand{\Ne}{\mathit{NE}}
\def\qslash{\delimiter"502F30E\mathopen{}}

\newcommand{\util}{u}
\newcommand{\boolutil}{u^B}

\newcommand{\bigO}{\mathcal{O}}

\newcommand{\cons}{\mathit{cons}}
\newcommand{\fcons}{\mathcal{C}}
\newcommand\restr[2]{#1 \cap #2}

\newcommand{\pre}{\mathsf{pre}}
\newcommand{\opt}{\mathsf{opt}}
\newcommand{\pess}{\mathsf{pess}}
\newcommand{\real}{\mathsf{real}}
\newcommand{\mmax}{\mathsf{max}}

\title{Boolean Observation Games}

\author{Hans van Ditmarsch \thanks{CNRS, University of Toulouse, IRIT, France}
  \and Sunil Simon \thanks{Department of CSE, IIT Kanpur, India}}
\date{}

\begin{document}
\maketitle

\begin{abstract}
We introduce {\em Boolean Observation Games}, a subclass of
multi-player finite strategic games with \textit{incomplete
information} and \textit{qualitative objectives}. In Boolean observation
games, each player is associated with a finite set of propositional \emph{variables} of which only it can observe the value, and it controls whether and to whom it can {\em reveal} that value. It does not control the given, fixed, value of variables. Boolean observation games are a generalization of Boolean games, a well-studied subclass of strategic games but with complete information, and wherein each player controls the value of its variables. 

In Boolean observation games, player goals describe multi-agent knowledge of variables. As in classical strategic games, players choose their strategies
simultaneously and therefore observation games capture aspects of both
imperfect and incomplete information. They require reasoning about
sets of outcomes given sets of indistinguishable valuations of
variables. An outcome relation
between such sets determines what the Nash equilibria are.  We present various outcome relations,
including a qualitative variant of ex-post
equilibrium. 
We identify conditions under which, given an outcome relation,
Nash equilibria are guaranteed to exist. We also study the complexity
of checking for the existence of Nash equilibria and of verifying if a
strategy profile is a Nash equilibrium.  We further study the subclass
of Boolean observation games with `knowing whether' goal formulas, for which
the satisfaction does not depend on the value of variables. We
show that each such Boolean observation game corresponds to
a Boolean game and vice versa, by a different correspondence, and that both correspondences are precise in terms of existence of Nash equilibria.
\end{abstract}

\section{Introduction}\label{sec.introduction}
Reasoning about strategic agents is an important problem in the theory
of multi-agent systems and game-theoretic models and techniques are
often used as a tool in such analysis. Strategic games
\cite{osborneetal:1994} is a classic and well-studied framework that
models one-shot multi-player games where agents make their choice
simultaneously.
It forms a simple and intuitive formalism to analyse and
reason about the strategic behaviour of agents. From the perspective
of computer science and artificial intelligence, one of the main
drawbacks of strategic games is that the explicit representation of
the payoff (or utility) function is exponential in the number of
players and the strategies available for each player. In many
applications, compact representation of the underlying game model is
highly desirable.

Various approaches have been suggested to achieve compact
representation of games and these mainly involve imposing restrictions
on the payoff functions. For instance, constraining the payoff
functions to be pairwise separable \cite{Jan68,CD11} results in the
well-studied class of games with a compact representation, called
\textit{polymatrix games}. \textit{Additively separable hedonic games}
\cite{AS16}
form another subclass of strategic games with pairwise separable
payoff functions which can be used to analyse coalition formation in
multi-agent systems. It is also possible to achieve compact
representation by explicitly restricting the dependency of payoff
functions to a ``small'' number of other agents (or neighbourhood) as
done in \textit{graphical games} \cite{KLS01}.

An alternative approach to imposing quantitative constraints on
payoffs is to restrict the payoffs to qualitative outcomes which are
presented as logical formulas. For example, ``extensive'' games played on graphs where the goal
formulas can specify the evolution of play with a combination of
temporal and epistemic specifications.
Although originally
defined as two-player perfect information games motivated by questions
in automata theory and logic, these models are now sophisticated to
reason about multi-player games and imperfect
information  \cite{CDH07,AG11,GMP17}. Boolean
games \cite{harrensteinetal:2001}, a subclass of strategic
games with complete information where objectives are expressed as
Boolean formulas, is also a well-studied framework with such qualitative outcomes. 

In Boolean games, each player controls a disjoint subset of
propositional variables where their strategies correspond to choosing
values for these variables and each player's goal is specified by a
Boolean formula over the set of all variables. 
While the model was
originally defined to analyse two-player games, the framework has
been extended in many directions.

Multi-player, non-zero-sum Boolean games are studied in
\cite{harrenstein:2004,BonzonLLZ06}. In
\cite{harrensteinetal:2001,harrenstein:2004} Boolean games are
modelled as imperfect information games by taking the uncertainty over
the other player's actions as an information set, as in
\cite{jfak.bulletin:2001}. In \cite{DH04,BonzonLLZ06,DM12} the
computational properties of Boolean games are adressed, in
\cite{bonzonetal:2009} graphical dependency structures for Boolean
games and their implications for various structural and computational
properties, and in \cite{EL14} mixed strategy Nash equilibria and
related computational questions. The issue of equilibrium selection is
considered in \cite{AHH13}. Iterated Boolean games
\cite{GutierrezHW15,GutierrezHPW16} model repeated interaction between
players with temporal goals specified in linear time temporal logic
(LTL). Partial ordering of the run-time events in terms of a
dependency graph on propositions is studied in \cite{BGW16}.

Epistemic Boolean games, wherein goal formulas may be epistemic, 
were proposed in \cite{AgotnesHHW13,HerzigLMS16}. Both works combine the control of 
variables with the observation of variables (or
formulas), where some of this is strategic and some is given with the game. This hybrid setting allows the authors to
continue to analyse these epistemic Boolean games as complete information 
strategic form games. Realizing epistemic objectives depends on the valuation of variables {\em resulting} from strategic action.

In this paper, we
introduce Boolean observation games as a qualitative model to
analyse and reason about a subclass of strategic games
with \textit{incomplete information}. In Boolean observation games,
{\em players control whether and to whom they reveal (announce) the value of propositional
variables that can only be observed by them}. This constitutes a
multi-player game model with concise representation where players have
(qualitative) epistemic objectives. It is incomplete information because realizing the objectives depends on a given fixed valuation that the players cannot control. Players do not know what that valuation is and therefore do not know what game they play. 
Realizing epistemic objectives depends on the unknown valuation of variables that is independent from strategic action. (We should note that such incomplete games of imperfect information can also be modelled as complete games of imperfect information by assuming an initial random move of a player `nature' determining the valuation.)

Since Boolean observation games define a subclass of strategic games,
they form an ideal framework to analyse interactive situations that
incorporate aspects of both imperfect as well as incomplete
information games. Please consider the following examples.

\begin{example}[A West Side Story] \label{example.tonymaria}
Tony and Maria (or was Romeo and Juliet? or Shanbo and Jingtai?) are in love with each other. But they have not declared their love to each other yet. This is risky business, as they are both uncertain about the feelings of the other one. Surely, given that they both love each other, their objective is to get to know that. But they consider it possible that the other person does not love them, in which case they might prefer not to declare their love. Their personalities are different in that respect. What Tony wants to know, depends on how his feelings (being in love / not being in love) relate to the other person's: if they match, he wants the other person to know, otherwise, he doesn't. Whereas what Maria wants to know only depends on the other person's feelings: if the other one is in love, she wants the other one to know her true feelings and otherwise not.

Given their state of mind and their personalities, should they declare their love to each other?

Let Tony be player 1 and Maria be player 2, and let $p_1$ represent `Tony is in love' and $p_2$ represent `Maria is in love'. Propositions $p_1$ and $p_2$ are both true and remain so forever after. They cannot be controlled. The objectives  (goals) denoted $\gamma_i$ for player $i$, and where $K_i p_j$ means `player $i$ knows $p_j$', are:
\[\begin{array}{llllll}
\gamma_1 = \gamma_2 & = &  p_1\et p_2 &\imp & K_1 p_2 \et K_2 p_1 & \et \\
 & & p_1\et \neg p_2 &\imp& K_1 \neg p_2 \et \neg K_2 p_1 & \et \\
 & & \neg p_1\et p_2 &\imp& \neg K_1 p_2 \et \neg K_2 \neg p_1 & \et \\
 & & \neg p_1\et \neg p_2 &\imp& \neg K_1 \neg p_2 \et K_2 \neg p_1
\end{array}\] 
They each have two strategies: declare their feelings (revealing the
value of $p_i$), or not. We succinctly explain that in this game, whatever the facts are, there is a strategy profile in which both players win by satisfying their goal formulas, but that they can never know that they win.
It is not so clear whether there (hopefully) is an equilibrium strategy profile allowing them to declare their love to each other. As $p_1$ and $p_2$ are true, it is an equilibrium when they both announce that, as $K_1 p_2 \et K_2 p_1$ is then true and they both win (the other three strategy profiles result in both losing, including another equilibrium namely when both don't declare). But Tony considers it possible that $\neg p_2$ in which case announcing $p_1$, and Maria's behaviour being equal, goal $K_1 \neg p_2 \et \neg K_2 p_1$ will fail. In that case he should have kept his mouth shut to have them win. Given the uncertainty over the game he has to reason about not two but {\bf four} strategies for Maria: depending on whether she is in love or not, whether she would show her feelings or not. What he will do given this information set of two indistinguishable outcomes, also depends on his risk aversity. If he's an optimist, he might still go for it. But if he's a pessimist, maybe better not. Maria's considerations are not dissimilar, but recall that she has a different personality (the goals are a different function from their value of $p_i$, in other words, permuting all occurrences of $1$ and $2$ in the goal results in a different goal). Example~\ref{ex.one} on page \pageref{ex.one} will reveal it all.
\end{example}

\begin{example}[A game of pennies that do not match] \label{example.pennies}
Consider two players Odd and Even both having a penny. They also both have a dice cup wherein they put their penny, shake the cup, and then put it on the table and watch privately whether their penny is heads or tails. Now they decide whether to inform the other player of the result, or not. If they both do or if they both don't, Even wins. So, Even wins if either $2$ players know or $0$ players know, so that one might therefore say that their state of knowledge is `even'. Otherwise, Odd wins. For the outcome it only matters whether they {\bf know} that the penny is heads or tails, it does not matter whether it {\bf is} heads or tails. What should they do?

We let Odd be player $1$ and Even be player $2$, and we let $p_1$ represent `Odd's penny is heads', whereas $p_2$ stands for `Even's penny is heads'. The goals are therefore (where $\Kw_i p_j$ abbreviates $K_i p_j \vel K_i \neg p_j$ and means `player $i$ knows whether $p_j$):
\[\begin{array}{rrrrr}
\gamma_1 & = & \Kw_1 p_2 & \eq & \neg\Kw_2 p_1 \\ 
\gamma_2 & = & \Kw_1 p_2 & \eq & \Kw_2 p_1 
\end{array}\] 
On first sight it seems quite straightforward what they should do, as the outcome does not depend on the valuation of $p_1$ and $p_2$. If Odd and Even both announce the result of their throw with the penny, Odd would then have done better not to make that announcement. But if that were to have happened, Even would have done better not to announce either. And so on. There is no equilibrium. Or is there? Yes, there is. And it is pure. Example~\ref{ex:noNE} on page \pageref{ex:noNE} will reveal it all.
\end{example}

Our framework of Boolean observation games clearly builds upon
\cite{AgotnesHHW13,HerzigLMS16} but a main difference is that these are complete information games whereas ours are incomplete information games. Thus we have very different strategies. Players do not control the \emph{values} of variables, but they control whether they \emph{reveal} the fixed values of variables that only they can observe. In that respect our framework also builds on the public announcement games
of \cite{agotnesetal:2011,agotnesetal.qa:2011}. They only allow strategies that are public announcements wherein the same information is revealed to all players. However, they permit announcing any epistemic formula, not merely propositional variables. A more detailed comparison with all these approaches is only possible after having given our framework in detail and is therefore in a later Section~\ref{section.related}.

Our games are strictly qualitative and thus abstract
from truly Bayesian approaches \cite{harsanyi:1967} with
probabilities. To determine equilibrium we compare information sets, called `expected outcomes'. As the expected outcome may not be a
value and the relation may not be a total order, our work is therefore in ordinal game theory
\cite{DurieuHQS08,cruzsimaan:2000,AmorFS17}.

\medskip

\noindent{\bf Our Contributions.}  We analyse structural and
computational properties of Boolean observation games. We define
Boolean observation games as incomplete and imperfect information
games, a novel perspective in Boolean games. We show that Boolean
observation games form a fragment of strategic games with compact
representation. We determine equilibria based on four different
profitable deviations from information sets, namely defined as: the
worst outcome is better, the best outcome is better, the expected
outcome is better, and the outcome without uncertainty is better (the
outcome is better even if all information sets are singleton, so that
the game is one of complete information). We also provide existence
results for such equilibria, which highly depends on what is
considered a profitable deviation. We identify various fragments of
Boolean observation games including one where your goal may be to keep
others ignorant but not to keep yourself ignorant, the self-positive
goals, and another one where the goals are `knowing whether formulas'
of which the realization does not depend on the valuation. The latter
we call knowing-whether games. We provide an embedding of the standard
Boolean games into a fragment of the knowing-whether games, and we
also provide an embedding of the knowing-whether games into the
Boolean games. Employing these embeddings we show that the
knowing-whether games correspond to Boolean games in terms of
existence of equilibrium outcomes. We also provide complexity results
for the natural questions of verification and checking of emptiness of
equilibrium outcomes in Boolean observation games, for most of the
profitable deviations considered, and stretching the results as much
as possible to also include fragments with ignorance goals. An
overview of these complexity results is found in the conclusions in
Table~\ref{tab:summary}.

\noindent{\bf Overview of Contents.}
Section~\ref{section.model} provides technical preliminaries needed to define Boolean observation games, that are then defined in the subsequent Section~\ref{section.model2}, of which the final Subsection~\ref{section.related} compares our proposal to other epistemic Boolean games. Section~\ref{section.kw} presents the correspondence between Boolean games and Boolean observation games. Section~\ref{section.exist} provides various results for the existence of Nash equilibria and  Section~\ref{sec:complexity} contains the results on the computational complexity of determining whether a strategy profile is an equilibrium, and whether equilibria exist.

\section{Preliminaries} \label{section.model}

In this section we introduce an auxiliary notion that is a complete information strategic game, which is played with strategies that are epistemic actions, that has epistemic formulas as goals and for which we propose a greatly simplified epistemic logic, and where outcomes are the truth values of those goals. Boolean observation games, that are incomplete information strategic games with more complex strategies and outcomes, will then be defined in the next section. The logic is simple in order to ensure a compact representation allowing to obtain complexity results comparable to those for Boolean games. Some logical details that are fairly elementary but that might distract from the game theoretical content that is our focus, are deferred to the Appendix.

\subsection{Strategies Consisting of Players Revealing Observations}

Let $N=\{1,\ldots,n\}$ be a finite set of \emph{players} $i$ and $P$ a finite
set of (propositional) \emph{variables} such that $(P_i)_{i \in N}$
defines a partition of $P$. The set $P_i$ is the set of variables
$p_i$ {\em observed} by player $i$ (that is, of which player $i$, and only player $i$, observes the value). A \emph{valuation} is a subset $v
\subseteq P$, where $p_i \in v$ means that $p_i$ is true and $p_i \notin v$
means that $p_i$ is false. The set $\powerset(P)$ of all valuations is denoted $V$.

A \emph{strategy} for player $i$ is a function $s_i : N \to
\powerset(P_i)$ that assigns to each player $j$ the set $s_i(j)
\subseteq P_i$ of variables that player $i$ reveals (announces) to
player $j$. We require that $s_i(i) = P_i$. Let $S_i$ denote the set
of all strategies of player $i$. A {\em strategy profile} is a member
$s$ of $S=S_1 \times \cdots \times S_n$. The set $P_i(s) = \{ p_j \in
P \mid p_j \in s_j(i) \}$ consists of the variables revealed to $i$ in
$s$. As $s_i(i) \subseteq P_i$, $P_i \subseteq P_i(s)$.  For $i \in
N$, we denote the $n$-tuple $s$ as $(s_i,s_{-i})$ where $s_{-i}$
represents the $(n-1)$-tuple of the strategies of other
players. Strategy $s^\emptyset_i$ is such that for all $j \in N$ with
$j \neq i$, $s_i(j) = \emptyset$. This means that player $i$ does not
reveal anything to anyone. Strategy $s^\all_i$ is such that for all
$i,j \in N$, $s_i(j) = P_i$. This means that player $i$ reveals
everything and to everyone.

Given $i \in N$ and strategy profile
$s$, the {\em observation relation} $\sim_i^s$ on $V$ is defined as, for $v,w \in V$: \begin{quote} $v \sim_i^s w$ \quad iff \quad $v \inter P_i(s) = w \inter P_i(s)$. \end{quote}
Observation relation $\sim_i^s$ encodes the informative effect of $s$. For $\sim^{s^\emptyset}_i$ we write $\sim_i$. This is the {\em initial observation relation}. We further note that  $P_i(s^\all) = V$ for any player $i$, so that $\sim^{s^\all}_i$ is the the identity relation $=$. A $\sim_i^s$ equivalence class, defined as $[v]_i^s := \{ w \in V \mid w \sim_i^s v \}$ (where $[v]_i^{s^\emptyset}$ is denoted $[v]_i$), is also called an {\em information set} (of player $i$ given valuation $v$ and observation relation $\sim_i^s$).

Inasfar as strategies consist of each player $i$ selecting a subset $P'_i$ of her variables $P_i$, these are like the strategies in Boolean games. However we interpret this differently: player $i$ does not make the variables in $P'_i$ true, but reveals the value of the variables in $P'_i$ according to a fixed valuation $v$. Another departure (or generalization) from Boolean games is that different variables are revealed to different agents. This is because we felt that more interesting game theoretical results could be obtained for such a generalization, and because more interesting communicative scenarios could then be treated with the game theoretical machinery.

\begin{example} \label{ex.zooop}
We assume a strategy profile to take place in some instantaneous, synchronous,
fashion, such as, when $s_1(2) = \{p_1,q_1\}$, $s_1(3) =
\{p_1,q_1\}$, and $s_1(4)=\emptyset$, player 1 informing player 2 and player 3 that $p_1$ and $q_1$ are both true, and such that player 4 observes this without being party to the message content (for example, 1 whispering to 2 and 3). In other words, player 4 knows that player 1 informs player 2 and player 3 {\bf whether} $p_1$ and $q_1$, but player 4 remains uncertain of the value of $p_1$ and $q_1$, so does not know that 1 informs 2 and 3 {\bf that} $p_1$ and $q_2$.\footnote{In a different semantics for strategies, less informative to the players, each player only learns what variables have been revealed by others to herself, and what variables she reveals to others. Applied to Example~\ref{ex.zooop}, this would also leave player 4 uncertain whether player 1 has informed player 2 and player 3. See Appendix~\ref{appendix.weaker}.}

Now consider $s'_1$ that is like $s_1$ except that $s_1(4) =
\{p_1,q_1\}$ as well. This is the public announcement of $p_1$ and $q_1$ by player $1$ to all players.

What if for example $s''_1(2) = \{p_1\}$ but $s''_1(3) =
\{p_1,p_2\}$? And what about $s_2$, $s_3$ and $s_4$? This cannot be done instantaneously. But we can ensure independence: all players commit to their $s_i$ before they execute it, and not after they see what variables are revealed to other players before it is their turn to reveal. Instead of whispering we can all have prepared closed envelopes adressed to all others on which is written for example, `from player 1 to player 2: contains the truth about $p_1$ and $p_2$'. All envelopes are collected blindly and then put on the table for all to see and are then handed out. 
\end{example}

Such forms of communication are known as semi-public announcement \cite{hvd.jolli:2002}, see Appendix~\ref{section.del} on dynamic epistemic logic for details. 

\subsection{Goals that are Epistemic Formulas}  \label{sec.logic}
The  \emph{language of epistemic logic} is defined as follows, where $i \in
N$ and $p_i \in P_i$.
\[L^K \ni \quad \quad \alpha := p_i \mid \neg \alpha
\mid \alpha \vee \alpha \mid K_i \alpha\]
Here, $\neg$ is negation, $\vee$ is disjunction, and $K_i \phi$ stands for `player $i$ knows $\phi$.' Other propositional connectives are defined by abbreviation, and
also $\M_i \alpha := \neg K_i \neg \alpha$ (player $i$
considers $\alpha$ possible), and $\Kw_i \alpha := K_i \alpha \vel
K_i \neg \alpha$ (player $i$ knows whether $\alpha$). The members of $L^K$ are {\em goals} and may as well be called, suiting our purposes {\em formulas}. 

The following fragments of $L^K$ also play a role, where $i, j \in N$ and $p_i \in P_i$. 
\[\begin{array}{ll} 
L^B \ni & \quad \quad \alpha := p_i \mid \neg \alpha
\mid \alpha \vee \alpha \\
L^K_\nnf \ni & \quad \quad \alpha := p_i \mid \neg p_i \mid \alpha \et \alpha \mid \alpha \vee \alpha \mid K_i \alpha \mid \M_i \alpha \\
L^+ \ni & \quad \quad \alpha := p_i \mid \neg p_i \mid \alpha \et \alpha \mid \alpha \vee \alpha \mid K_i \alpha \\
L^\Kw \ni & \quad \quad \alpha := \Kw_j p_i \mid \neg \alpha \mid \alpha \vee \alpha \\
L^\Kw_\nnf \ni & \quad \quad \alpha := \Kw_j p_i \mid \neg \Kw_j p_i 
\mid \alpha \vee \alpha \mid \alpha \et \alpha
\end{array}\]
The language $L^B$ of the \emph{Booleans} is the fragment of $L^K$ without $K_i$ modalities.  In the language $L^\Kw$ of {\em knowing whether formulas} (\emph{$\Kw$ formulas}) the constructs $\Kw_j p_i$ play the role of propositional variables. The fragments $L^K_\nnf$ and $L^\Kw_\nnf$ are those of the {\em negation normal form} ($\nnf$) of respectively $L^K$ and $L^\Kw$,  where the language $L^+$ of the {\em positive formulas} is the fragment of $L^K_\nnf$ without $\M_i$ modalities (corresponding to a universal fragment of first-order logic). Note that $L^\Kw$ and $L^\Kw_\nnf$ are really propositional languages, not modal languages. A goal is \emph{guarded} if it has shape $\gamma_i = K_i \alpha$.

Apart from the above fragments yet another fragment plays a role in
our contribution, namely that of the self-positive goals.  The {\em
  self-positive goal formulas} are defined as $L^{\mathsf{self}+} :=
\Union_{j \in N} L^{{j}+}$, where each $L^{{j}+}$ is given by the
following BNF, wherein $i,k\in N$ and $k \neq j$.
\[\begin{array}{lll}
L^{{j}+} \ni \ \ \alpha_j & ::= & p_i \mid \neg p_i \mid \alpha_j \et \alpha_j \mid \alpha_j \vel \alpha_j \mid K_j \alpha_j \mid K_k \alpha_j \mid \M_k \alpha_j
\end{array}\]
Here, $\alpha_j$ is the goal for player $j$. Note that $L^+$ is a fragment of $L^{{j}+}$, namely the fragment where all occurrences of $K_k$ are positive, and that $L^{{j}+}$ is a fragment of $L^K_{\mathsf{nnf}}$, namely the fragment wherein all occurrences of $K_j$ are positive. In a self-positive goal for agent $j$, $j$'s objective is to (get to) know others' variables {\bf and others' knowledge and ignorance}, although other players may either know or remain ignorant of $j$'s knowledge. This implies that $j$'s goal also cannot be for others to know $j$'s ignorance. A larger number of communicative scenarios seem to have self-positive goals than merely positive goals: it seems fairly typical that you wish others to remain ignorant even when you are only interested in obtaining (factual) knowledge.

The inductively defined semantics of $L^K$ formulas are relative to a
valuation $v$ and a strategy profile $s$, where $i \in N$ and $p_i \in
P_i$. 
\[\begin{array}{lcl}
v,{s} \models p_i & \text{iff} & p_i \in v \\
v,{s} \models \neg \alpha & \text{iff} & v,{s} \not\models \alpha \\
v,{s} \models \alpha_1 \vee \alpha_2 & \text{iff} & v,{s} \models
  \alpha_1 \text{ or } v,{s} \models \alpha_2 \\
v,{s} \models K_i\alpha & \text{iff} & w,{s} \models \alpha \text{ for all } w \text{ such that } v \sim^{s}_i w
\end{array}\]
For $v,s^\emptyset\models\alpha$ we write $v \models \alpha$. This is a bit sneaky: by definition this represents what players know after the strategy profile is executed wherein nobody reveals anything, but we can therefore just as well let it stand for what players initially know, before anything has been revealed. 

We let $s \models \alpha$ denote ``for all $v \in V$,
$v,s\models\alpha$,'' and $\models\alpha$ denote ``for all $s \in S$,
$s \models \alpha$''. In our semantics, $K_i p_i$, $K_i \neg p_i$, and
$\Kw_i p_i$ are always true (equivalent to the trivial assertion
$\top$). We therefore informally assume that they do not occur in goal
formulas.

We note that our epistemic semantics is not the usual one for the
epistemic language, interpreted on arbitrary Kripke models, but a
greatly simplified epistemic semantics dedicated to reason about
strategies that are joint revelations of observed variables. We do not
even use the word `model'. And we do not allow announcements
(revelations) of other information than variables. In
Appendix~\ref{section.del} we show how (valuation, strategy) pairs
induce multi-agent Kripke models. All these simplifications are in
order to obtain a smooth comparison with Boolean games and with
comparable complexities, unlike the higher complexities common in
multi-agent epistemic reasoning.

We continue with some elementary properties of this simple logical semantics, in the form of propositions.

\begin{restatable}{proposition}{restatableone}\label{propone}
Each formula in $L^K$ is equivalent to a formula in $L^K_\nnf$. Similarly, each formula in $L^\Kw$ is equivalent to a formula in $L^\Kw_\nnf$. 
\end{restatable}

%\restatableone*
\begin{proof}This well-known result in modal logic for $L^K$
  is shown by induction on formula structure, using the 
  equivalences $\neg\neg\alpha\eq\alpha$, $\neg (\alpha \vel \beta) \eq (\neg \alpha \et \neg \beta)$ and $\neg K_i \alpha \eq \M_i \neg \alpha$. For $L^\Kw$, as this is essentially a propositional and not a modal language, we only need to use the first equivalence.
\end{proof}

\begin{restatable}{proposition}{restatabletwo}\label{lemma.ii}
  \label{lm:KwAllVal}
For all $\phi\in L^\Kw$, valuations $v$, and strategy profiles $s$: $v,s\models\phi$ iff $s\models\phi$.
\end{restatable}
The basic but lengthy proof of this proposition is in Appendix~\ref{app.proofs}. 
Prop.~\ref{lemma.ii} says in other words, that if $v,s \models \phi$ for \emph{some} $v \in V$, then $v,s \models \phi$ for \emph{all} $v \in V$.

\begin{restatable}{proposition}{restatablethree}\label{cor.one}
For any $\alpha \in L^\Kw$, $\models \alpha \eq K_i \alpha$.
\end{restatable}

%\restatablethree*
\begin{proof}
Let valuation $v$ and strategy profile $s$ be given. 

Assume $v,s \models \alpha$. Then from  Prop.~\ref{lemma.ii} it follows that for all $w \in V$, $w, s \models \alpha$. Therefore, in particular, $w,s \models \alpha$ for all $w \sim_i^s v$, which is by definition $v,s \models K_i \alpha$. 

Now assume $v,s \models K_i \alpha$. From $v \sim_i^s v$ and the semantics of knowledge now directly follows $v,s \models \alpha$. 

As $v$ and $s$ were arbitrary, we have shown $\models \alpha \eq K_i \alpha$.
\end{proof}

As a consequence each formula in the fragment $\Kw_i p_j \mid \neg \alpha \mid \alpha \vee \alpha \mid K_i \alpha$ is equivalent to a formula in $L^\Kw$, in other words, knowledge can then be eliminated. This explains why we defined the fragment $L^\Kw$ without an inductive clause for knowledge. 

Knowledge cannot generally be eliminated from a language with knowing whether variables. For example, Anne (1) may know whether Bill (2) passed the exam ($p_2$), but Bill may be uncertain {\em whether} she knows. So we have $\Kw_1 p_2 \et \neg K_2 \Kw_1 p_2$. Props.~\ref{lemma.ii} and \ref{cor.one} (and the subsequent Prop.~\ref{five}) do not hold for knowing whether fragments on arbitrary Kripke models.

\begin{restatable}{proposition}{restatablefour} \label{five}
For all $i,j,k \in N$: $\models \Kw_i\Kw_j p_k$.
\end{restatable}

%\restatablefour*
\begin{proof}
Formula $\Kw_i\Kw_j p_k$ is by definition equivalent to $K_i \Kw_j p_k \vel K_i \neg\Kw_j p_k$. From Prop.~\ref{cor.one} it follows that this is equivalent to $\Kw_j p_k \vel \neg \Kw_j p_k$ which is a tautology. 
\end{proof}

Therefore, in our very simple epistemic logic it is common knowledge whether a player knows a variable. This reflects the dynamics of revealing variables. Suppose all players hold cards named $p_1$, $q_1$, $p_2$, \dots on the back side and the value $0$ or $1$ on the front (face) side. You may not know {\em that} your neighbour has shown to your other neighbour that the value of the card $p_1$ is 1 (true). But you know {\em whether} your neighbour has shown card $p_1$ to your other neighbour. You saw it happen. 

%\begin{example}
%\end{example}

\subsection{Pointed Boolean Observation Games}  

A \emph{pointed Boolean observation game} (pointed observation game) is a pair
$(G,v)$, denoted $G(v)$, where $v \in V$ and where $G$ is a triple $(N,(P_i)_{i \in N}, (\gamma_i)_{i \in N})$, where all $\gamma_i \in
L^K$. The players' strategies in the pointed observation game are the strategies $s_i \in S_i$. The players' goals in the pointed observation game are the $\gamma_i \in L^K$. Given $i \in N$, the \emph{outcome function} $u_i: V \times S \to \{0,1\}$ of a pointed observation game is defined as: \begin{quote} $u_{i}(v,s)=1$ if $v,{s} \models
\gamma_{i}$ and $u_{i}(v,s)=0$ if $v,{s} \not \models \gamma_{i}$. \end{quote} A
strategy profile $s$ is a \emph{Nash equilibrium} of $G(v)$ iff for
all $i \in N$ and $s_i' \in S_i$ we have $u_i(v,s) \geq
u_i(v,(s_i',s_{-i}))$.  That is, no player has a \textit{profitable deviation} from $s$ in $G(v)$, which would therefore be a $s_i' \in S_i$ such that $u_i(v,s) <
u_i(v,(s_i',s_{-i}))$. Observe that a player can only make a profitable deviation from $s$ if her goal is not satisfied in $s$. Let $\Ne(G(v))$ denote the set of Nash equilibria of $G(v)$. 

The pointed observation game is an auxiliary notion, matching the
intuition that after revealing variables a player wins when her goal
has become true.  The game is one of complete information because the
valuation is known to you, the reader. But the valuation is typically
not known to the players. It already uses the parameters of the Boolean observation game that we will now define in the next section.

\begin{example} \label{example.pointed}
We recall Example~\ref{example.tonymaria}. We summarily describe a pointed Boolean observation game and its equilibria, where a fuller development is only given in Example~\ref{ex.one}. Consider pointed game $G(v)$ with $G = (\{1,2\}, (\{p_1\},\{p_2\}), (\{\gamma_1,\gamma_2\})$ where $\gamma_1,\gamma_2$ are as in Example~\ref{example.tonymaria}, and where valuation $v = \{p_1,p_2\}$ (both are in love). The strategies are to reveal nothing or to reveal all, that is: $s_1^\emptyset$, $s_1^\forall$, $s_2^\emptyset$, and $s_2^\forall$.

The strategy profile $(s_1^\forall,s_2^\forall)$ is an equilibrium strategy profile of the pointed game $G(v)$, with outcome $1$ for both players. This is the only way to make $K_1 p_2 \et K_2 p_1$ true.  However, both players not announcing their variable is also an equilibrium with outcomes $0$. 

The pointed game $G(w)$ for valuation $w = \{p_1\}$ (only Tony is in love) has equilibrium $(s_1^\emptyset,s_2^\forall)$. We now need to make $K_1 \neg p_2 \et \neg K_2 p_1$ true. (Another equilibrium $(s_1^\forall,s_2^\emptyset)$ is when both get outcome $0$.)
\end{example}

\section{Defining Boolean Observation Games}
\label{section.model2}

%\subsection{Introduction}

We will now define the {\em Boolean observation game}. A Boolean
observation game is an incomplete information strategic form game with
{\bf uniform strategies} (uniform functions from valuations to
strategies) and {\bf expected outcomes} (information sets of
outcomes), whereas the auxiliary notion of a pointed observation game
is a complete information strategic form game with {\bf strategies}
and with (Boolean-valued) {\bf outcomes}.

\subsection{Boolean Observation Games}

This section contains the crucial game theoretical notions of our contribution. 

\medskip
\noindent {\bf Boolean Observation Game.}
A \emph{Boolean observation game} (or {\em observation game}) is a triple 
$G=(N,(P_i)_{i \in N}, (\gamma_i)_{i \in N})$, where all $\gamma_i \in
L^K$. Formula $\gamma_i$ is the \emph{goal} (\emph{objective}) of
player $i$. It is played with {\em uniform strategies} and the payoffs are {\em expected outcomes}. Both will now be defined.

\medskip

\noindent {\bf Uniform Strategy.}  A \emph{uniform strategy} for
player $i \in N$ is a function $\cs_i: V \to S_i$ such that for all
$v,w$ with $v \sim_i w$, $\cs_i(v) = \cs_i(w)$. It is \emph{globally
  uniform} iff for all $v,w \in V$, $\cs_i(v) = \cs_i(w)$. 
  
So, uniform
means the same for all indistinguishable valuations, which is
different from globally uniform, which means the same for all
valuations.  Let $\CS_i$ denote the set of uniform strategies of
player $i$, and $\CS = \CS_1 \times \dots \times \CS_n$ the set of
\emph{uniform strategy profiles}. Let $\CS_i^g$ and $\CS^g$ denote the set of
globally uniform strategies of player $i$ and the set of globally
uniform strategy profiles respectively.  Given a valuation $v$, a
uniform strategy profile $\cs$ determines a strategy profile
$\cs(v)=(\cs_1(v),\dots,\cs_n(v))$. Note that $(\cs(v)_i, \cs(v)_{-i})
= (\cs_i,\cs_{-i})(v)$. For $i \in N$ and $s_i \in S_i$, we define
$\dot{s_i}\in\CS_i^g$ as: for all $v \in V$,
$\dot{s_i}(v)=s_i$. Similarly for $s \in S$ we define
$\dot{s}\in\CS^g$ as the globally uniform strategy profile such that for all $v \in V$, $\dot{s}(v)=s$. It follows
from the definition that every globally uniform strategy profile $\cs \in
\CS^g$ is of the form $\dot{s}$ for some strategy profile $s \in S$.

\medskip
\noindent {\bf Expected Outcome.}  Given $i \in N$, the
\emph{expected outcome function} is a function $\uu_i: V \times \CS \to
\{0,1\}^*$ that is uniform in $V$, and defined as $\uu_i(v,\cs) =
(u_i(w,\cs(w)))_{w \sim_i v}$. So, {\em expected outcome}
$\uu_i(v,\cs)$ is a vector of outcomes $u_i(w,\cs(w))$ for each
valuation $w$ in the information set of player $i$. In our setting
where outcomes are $0$ (lose) or $1$ (win) this vector is a bitstring.

As far as nomenclature is concerned, we are putting the reader on the wrong foot, as a uniform strategy is not a kind of strategy (as defined in the previous section), nor is expected outcome a kind of (binary valued) outcome. However, we are in good company: an artificial brain is not a brain, and a cable car is not a car. So we hope the reader will allow us this slight abuse of language.

\medskip
\noindent {\bf Outcome Relation and Nash Equilibrium.}
To define the
notion of an equilibrium in observation games, we need to first define
a comparison relation between uniform strategy profiles. Note that
unlike in classical strategic games, the expected outcome function in
observation games generates a vector of outcomes. Therefore, there is
no canonical definition for the comparison relation. We define an
\textit{outcome relation} $>$ over vectors of outcomes and write
$\uu_i(v,\cs) > \uu_i(v,\cs')$ for ``player $i$ prefers $\cs_i$ over
$\cs'_i$ in the information set containing $v$''; we also say that $\cs_i$ is a profitable deviation from $\cs'_i$.

This outcome relation may not be a total order. We therefore prefer not to use notation $\leq$ to compare the bitstrings that are outcome sets, as it is ambiguous whether $x \leq y$ means ($x < y$ or $x = y$) or $x \ngtr y$ (and even when defined as either one or the other, it seems unkind to the reader).

Given an outcome relation $>$, a uniform strategy profile is a Nash equilibrium if no player has a profitable deviation.
\begin{quote}
A uniform strategy profile $\cs$ is a \emph{Nash equilibrium} of $G$ iff for all $i \in N$, $\cs_i' \in \CS_i$ and $v \in
V$, we have that $\uu_i(v,(\cs_i',\cs_{-i})) \ngtr \uu_i(v,\cs)$. 
\end{quote}
Given an observation game $G$, $\Ne(G)$ denotes its Nash equilibria, and among those $\Ne^g(G)$ denotes the globally uniform Nash equilibria.

Also, a uniform strategy
$\cs_i\in \CS_i$ is \textit{dominant} if for all $\cs\in\CS$ with $\cs = (\cs_i,\cs_{-i})$, for all $\cs'_i \in \CS_i$,
and for all $v$, $\uu_i(v, (\cs_i',\cs_{-i})) \ngtr \uu_i(v,\cs)$.\footnote{This is weak dominance of the kind `always at least as good' where we emphasize that we do not define it as `always at least as good and sometimes strictly better', which is also common in game theory.}

\medskip
\noindent {\bf Four Outcome Relations.}
It remains to define the outcome relation. We propose four. 
% Below, the $>$ on the right-hand side is simply the greater-than relation between natural numbers.
\[\begin{array}{llll}
\text{optimist}:  & \uu_i(v,\cs) >^\opt \uu_i(v,\cs') & \text{iff} & \max \uu_i(v,\cs) > \max \uu_i(v,\cs') \\
\text{pessimist}: & \uu_i(v,\cs) >^\pess \uu_i(v,\cs') & \text{iff} & \min \uu_i(v,\cs) > \min \uu_i(v,\cs') \\
\text{realist}:  & \uu_i(v,\cs) >^\real \uu_i(v,\cs') & \text{iff} & \Sigma \uu_i(v,\cs) > \Sigma \uu_i(v,\cs') \\
\text{maximal}: & \uu_i(v,\cs) >^\mmax \uu_i(v,\cs') & \text{iff} & u_i(w,\cs(w)) > u_i(w,\cs'(w)) \text{ for some } w \sim_i v
\end{array}\]
The optimist, pessimist and realist outcome relations are (strict)
total orders, as it suffices to assign a number to the information set
constituting an expected outcome.
The maximal outcome relation is not a total order. 

We let $\Ne_{\mathsf{pess}}(G)$,
$\Ne_{\mathsf{opt}}(G)$, $\Ne_{\mathsf{real}}(G)$, 
and $\Ne_{\mathsf{max}}(G)$ denote the Nash equilibria under the
pessimist, optimist, realist and maximal outcome relation, respectively. The optimist, pessimist and realist outcome relations are (strict)
total orders, as it suffices to assign a number to the information set
constituting an expected outcome. The maximal outcome relation is not a total order as illustrated in Example \ref{ex:maxrel}. However, defining this relation is useful since $\Ne_{\mathsf{max}}(G)$ has an interesting
interpretation which we discuss below.

\begin{example}
\label{ex:maxrel}
Let us consider an abstract example where a player has to choose between expected outcomes (bitstrings) 00, 10, 01, 11. We then get (where clustered bitstrings means equally preferred):
\[\begin{array}{lll}
\{01,10,11\} >^\opt 00 \\
11 >^\pess \{00,01,10\} \\
11 >^\real \{01,10\} >^\real 00 \\
ij >^\mmax kl \text{ \ iff  \ } i > k \text{ or } j > l
\end{array}\]
\end{example}

The $>^\mmax$ relation  is neither antisymmetric nor transitive. For instance, in Example \ref{ex:maxrel} we have that $10 >^\mmax 01$ but also $01 >^\mmax 10$, and it is not transitive because $01 >^\mmax 10 >^\mmax 01$ however $01 \ngtr^\mmax 01$. Thus the $>^\mmax$ relation is neither a total order nor a preorder.
However, it has a maximum
and a minimum: the expected outcome where the player always wins is preferred over all
other expected outcomes, and the expected outcome where the player always loses is less
preferred than all other expected outcomes. 

The maximal outcome relation
also satisfies the important property that all outcomes can be compared and therefore, the notion of a Nash equilibrium is well-defined. If $\uu_i(v,\cs) \neq \uu_i(v,\cs')$, then $\uu_i(v,\cs) >^\mmax \uu_i(v,\cs')$ or $\uu_i(v,\cs') >^\mmax \uu_i(v,\cs)$. The disjunction in the consequent is inclusive, both may hold (we recall that $10>^\mmax 01$ as well as $01 >^\mmax 10$, as in Example \ref{ex:maxrel}). To require this property is common in ordinal game theory \cite{DurieuHQS08}.

The outcome relations that we have proposed are qualitative versions of well-known criteria in decision theory and Bayesian reasoning. None assume a probability
distribution, however, all assume a strictly positive probability for each valuation.
\begin{itemize}
\item The optimist outcome relation is the max instantiation (as there is
only one maximal value) of the \emph{minimax regret} decision
criterion \cite{savage:1951}. With respect to the highest possible outcome in the information set, a lower possible outcome in the information set (which can only be 0 instead of 1) would cause regret if this were to happen. 

\item The pessimist outcome relation is the
min instantiation of the \emph{maximin} or \emph{Wald} decision
criterion \cite{wald:1945}. We then choose the information set with
the best worst outcome. This outcome relation has been used to model uncertainty in voting (with similar considerations involving Nash equilibria and dominance) in  \cite{conitzeretal.aaai:2011,hvdetal.TARKvote:2013,bakhtiarietal:2019}.

\item The realist
outcome relation is a qualititative version (lack of justification to
rule out any outcome) of a random decision in Bayesian terms, also
known as the insufficient reason or Laplace decision criterion, or as
the principle of indifference \cite[Chapter IV]{keynes:1921}. 

Instead of taking the sum of the outcomes in the information set we could of course have normalized this so it adds up to 1, suggesting an even distribution of probability mass. Such scaling is irrelevant for our purposes of determining Nash equilibria and dominance, wherein we only need to compare outcomes. That comparison relation remains the same.

This outcome relation was used in \cite{agotnesetal:2011,agotnesetal.qa:2011} to determine
equilibria of similar incomplete information games, but where more complex formulas than mere variables could be `revealed' (however, they could only be publicly announced). An issue for the realist outcome relation is whether bisimilar game states (that therefore satisfy the same goals for all players) should be counted once or twice.\footnote{Personal communication by Martin Otto.} On the one hand, if two game states are bisimilar this is justification / sufficient reason to rule out one of them, according to Laplace.  On the other hand these bisimilar game states might have originated from playing strategy profiles (executing epistemic actions) in initial game states that were non-bisimilar. It is relevant to observe this as we note that this phenomenon {\bf cannot} occur in our simpler setting involving observation relations.

\item 
The notion of Nash equilibrium for the maximal outcome relation has an interesting
interpretation. A  maximal Nash equilibrium is a uniform
strategy profile where no player has a profitable deviation even if the player has complete information about the game. There
is an equivalent formulation of maximal Nash equilibrium as a
qualitative version of {\em ex-post equilibrium} \cite{Apt11}, which
we show in Proposition \ref{prop:NEST}.

\end{itemize}
Various of the above outcome relations have also been considered in \cite{parikhetal:2013}.

\begin{restatable}{proposition}{propNEST}
\label{prop:NEST}
A uniform strategy profile $\cs$ is a maximal Nash equilibrium for $G$ iff for all $v
\in V$, $\cs(v)$ is a Nash equilibrium for $G(v)$. 
\end{restatable}
\begin{proof}
Suppose $\cs
  \not\in \Ne_{\mathsf{max}}(G)$.  Then there exist $v \in V$, $i \in N$, and $\cs'_i \in \CS_i$ such that $\uu_i(v,
  (\cs'_i,\cs_{-i})) >^\mmax \uu_i(v,\cs)$. It follows that there is $w
  \sim_i v$ such that $u_i(w, (\cs'_i,\cs_{-i})(w)) > u_i(w, \cs(w))$,
  so $u_i(w, (\cs'_i,\cs_{-i})(w)) = 1$ and $u_i(w, \cs(w)) =
  0$. Therefore $\cs(w)\notin \Ne(G(w))$.

Suppose $\cs(w)\notin \Ne(G(w))$ for some valuation $w$.  Then there exist $i \in N$, $s'_i \in
S_i$ such that $u_i(w, (s'_i,\cs_{-i}(w)) >
u_i(w,\cs(w))$. Let $\cs'_i \in \CS_i$ be the uniform strategy such
that for all $v \sim_i w$, $\cs'_i(v) = s'_i$ (so in particular,
$\cs'_i(w) = s'_i$), and for all $v \not\sim_i w$, $\cs'_i(v) =
\cs_i(v)$. By the maximal relation, from $u_i(w, (\cs'_i,\cs_{-i})(w)) =
u_i(w, (s'_i,\cs(w)_{-i}) > u_i(w,\cs(w))$ it follows that $\uu_i(w,
(\cs'_i,\cs_{-i})) >^\mmax \uu_i(w,\cs)$. Therefore $\cs \notin
\Ne_{\mathsf{max}}(G)$.
\end{proof}

In the remaining sections we focus on optimist, pessimist and maximal
Nash equilibrium and not on realist Nash equilibrium. We use the operational definition of maximal Nash
equilibrium given by the correspondence in Proposition
\ref{prop:NEST}. It is easy to see that a maximal Nash
equilibrium is also an optimist, pessimist, and realist Nash
equilibrium. In that sense the maximal outcome relation is the \emph{strongest} notion, resulting in the smallest number of equilibria for a game (if any).

\subsection{Various Classes of Observation Games}
\label{subsec:subclass}
With all the technical tools now at our disposal, very different observation games are of specific interest. We can distinguish them by which outcome relation they employ, and independently by the shape of the epistemic goals. Concerning goals it is useful to distinguish the following.
\begin{itemize}
\item In \emph{two-player zero-sum games}, $\gamma_i = \neg \gamma_j$ and in \emph{two-player symmetric games} $\gamma_i = \gamma_j$, where $|N| = 2$, $i \neq j$, and $i,j \in N$. In \emph{cooperative games} $\Et_{i \in N} \gamma_i$ is consistent. Example~\ref{ex.one} below is symmetric, and Example~\ref{ex:noNE} is zero-sum (and therefore not consistent). Communicative scenarios obeying the Gricean cooperative principle are clearly consistent observation games (and might still be considered games inasfar as people want to outdo each other in being informative). Whereas security protocol settings with eavesdroppers (consider observing an SMS code that you were sent to confirm a bank transfer) tend to be zero-sum; that is, a generalization of zero-sum: the objectives of the principals are the opposite of those of the eavesdroppers. 

We do not have theoretical results for zero-sum or symmetric games.

\item In \emph{knowing-whether observation games} ({\em knowing-whether games, $\Kw$ games}) all goals $\gamma_i$ are in $L^\Kw$. In knowing-whether games the outcome does not depend on the valuation. Whether some $\Kw_i p_j$ is true only depends on player $j$ revealing $p_j$ to player $i$, and does not depend on the valuation, because the truth of $p_j$ does not depend on the value of $p_j$. 

Section~\ref{sec:knowingwhether} is entirely devoted to knowing-whether games, and Section~\ref{sec:existsNE} contains results on existence of equilibria. They relate well to the usual Boolean game. Not surprisingly, as the outcome does not depend on the valuation, they also score better on the computational complexity of determining whether a uniform strategy profile is a Nash equilibrium, or whether Nash equilibria exist, than other classes of observation game. That will be investigated in Section~\ref{subsec.complkw}.

\item  In observation games with \emph{guarded} goals (of shape $\gamma_i = K_i \alpha$, see Subsection~\ref{sec.logic}) the players know whether they have achieved their objective after playing the game. Whereas in games where the goals are not guarded they may not and need an oracle to inform them of the outcome (such as, when standing in front of an ATM teller, the bank's interface informing them). If goals are guarded, Nash equilibria always exists for the optimist and the pessimist outcome relation, as formulated and shown in Theorem~\ref{thm:existsNEpess} in Section~\ref{sec:existsNE}.

\item In games where all $\gamma_i$ are \emph{positive formulas} (in the fragment $L^+$ where negations do not bind $K_j$ modalities, see Subsection~\ref{sec.logic}), a player's goal is never to remain ignorant of a fact, or even for other players to remain ignorant. Under such circumstances revealing all you know is a dominant strategy. This is therefore rather restricted.

\item More interesting than positive goals are the observation games with {\em self-positive goals} wherein your goal is to become less ignorant yourself although you may wish to keep other players ignorant (see again Subsection~\ref{sec.logic}).  We provide a result for self-positive goals in Corollary~\ref{cor.self} in Section~\ref{sec:existsNE}.

\end{itemize}
For all these, results on existence of equilibria and complexity also depend on which outcome relation is used, as already occasionally listed above. 

Last but not least one can consider \emph{iterated observation games} with temporal eventuality goals, where players successively reveal more and more of their observed variables. An example are (successive) \emph{question-answer games} wherein the strategic aspect is what variable(s) to ask another player(s) to reveal, which seems of particular interest for strategic negotiation (if you give me this, I'll give you that, and so on). All these come with specific questions on compact representation and existence of equilibria. 

We defer the investigation of iterated games and question-answer games to future research. In this work we focus on knowing-whether games and their relation to Boolean games, on the existence of equilibria for various outcome relations (where the realist outcome relation plays no role), and on complexity results for some of our variations.

We now continue with some detailed examples.

\begin{example} \label{ex.one}
Recall Example~\ref{example.tonymaria} (page \pageref{example.tonymaria}) and Example~\ref{example.pointed}. We now give full details.

Consider the observation game $\obsgame$ where $N=\{1,2\}$, $P_1=\{p_1\}$,
$P_2=\{p_2\}$ and the (symmetric) goals:
\[\begin{array}{llllll}
\gamma_1 = \gamma_2 & = &  p_1\et p_2 &\imp& K_1 p_2 \et K_2 p_1 & \et \\
 & & p_1\et \neg p_2 &\imp& K_1 \neg p_2 \et \neg K_2 p_1 & \et \\
 & & \neg p_1\et p_2 &\imp& \neg K_1 p_2 \et \neg K_2 \neg p_1 & \et \\
 & & \neg p_1\et \neg p_2 &\imp& \neg K_1 \neg p_2 \et K_2 \neg p_1 %\\
%\gamma_2 & = & \gamma_1
\end{array}\] 
As there are only two players and each player observes a single variable the strategies are to reveal nothing or to reveal all, that is: $s_1^\emptyset$, $s_1^\forall$, $s_2^\emptyset$, and $s_2^\forall$.

For each valuation $v$ the pointed observation game $G(v)$ has an equilibrium where both players get outcome 1. For example, if $p_1$ and $p_2$ are both true, then
both players revealing (announcing) that is an equilibrium with
outcome 1 for both players. However, both players not announcing their variable is also an equilibrium with outcome 0.

Let us now determine equilibria for $\obsgame$, with uniform strategies instead of strategies, and let us consider the different outcome relations.
\begin{itemize}
\item {\bf pessimist.} Player 1 cannot distinguish between the valuations $\{p_1,p_2\}$ and $\{p_1\}$.
Thus, for all $\cs \in \CS$ and for all $v \in V$, $\min \uu_1(v,\cs)
= 0$. The situation is symmetric for player $2$. Therefore, for all $\cs \in \CS$, $\cs \in
\Ne_{\mathsf{pess}}(G)$.
\item {\bf optimist.} Similarly, for all $\cs \in \CS$,
for all $v \in V$ and for all $i \in \{1,2\}$, $\max \uu_i(v,\cs) =
1$. Therefore, for all $\cs \in \CS$, $\cs \in \Ne_{\mathsf{opt}}(G)$.
\item {\bf realist.} In this example, whatever the valuation $v$, $\Sigma \uu_i(v,\cs) = \max \uu_i(v,\cs) = 1$, so that also, for all $\cs \in \CS$, $\cs \in \Ne_{\mathsf{real}}(G)$.
\item {\bf maximal.} $\Ne_{\mathsf{max}}(G) = \emptyset$. There are no maximal Nash equilibria, because every information set for both players always contains a win and a lose, so if they were to know the real valuation, one of those is not an equilibrium for the pointed game.
\end{itemize}

Possibly, the equilibria depend on what we called the `personalities of Tony and Maria', that is on the shape of the goals? We considered two different personalities that therefore allow four different goals, but (the reader can check that) none makes a difference for any of the four outcome relations, as the property that each information set contains a win and a lose persists throughout such transformations. The best is always win, and the worst is always lose. However for other `personalities' (for lack of a better term) this need not be, for example, change $\neg K_1 p_2 \et \neg K_2 \neg p_1$ in the third conjunct into $\neg K_1 p_2 \et K_2 \neg p_1$ (we removed one negation symbol). It is now dominant for player $1$ to announce the value of $p_1$ in the information set wherein $p_1$ is false.
\end{example}

\begin{example}
\label{ex:noNE}
Recall Example~\ref{example.pennies} on page \pageref{example.pennies} about the pennies that do not match. We can now model this as a knowing-whether Boolean observation game $\obsgame$ where
$N=\{1,2\}$, $P = P_1 \union P_2$ with $P_1=\{p_1\}$, $P_2=\{p_2\}$, and 
\[\begin{array}{rrrrr}
\gamma_1 & = & \Kw_1 p_2 & \eq & \neg\Kw_2 p_1 \\
\gamma_2 & = & \Kw_1 p_2 & \eq & \Kw_2 p_1
\end{array}\] 
For $i = 1,2$, player $i$ has strategy $s_i^\emptyset$ wherein she
reveals nothing (`hide $p_i$') and strategy $s_i^\forall$ wherein she reveals the
value of $p_i$. Irrespective of the valuation, in the strategy profiles
$(s_1^\emptyset,s_2^\emptyset)$ and $(s_1^\forall,s_2^\forall)$, player $1$ has a
profitable deviation in the corresponding pointed observation
game. Similarly, in $(s_1^\emptyset,s_2^\forall)$ and $(s_1^\forall,s_2^\emptyset)$,
player $2$ has a profitable deviation. Thus it can be verified that
$\Ne_{\mathsf{max}}(G) = \emptyset$. Also, within the set of all
globally uniform strategy profiles, $G$ does not have a Nash
equilibrium for the pessimist and optimist outcome relation.

However, this game has a Nash equilibrium with uniform
strategies that are not globally uniform, for the pessimist and for the optimist outcome
relation. Consider the uniform strategy profile $\cs=(\cs_1,\cs_2)$
where in $\cs_1$, player
$1$ reveals $p_1$ to $2$ when $p_1$ is true and hides $p_1$ from $2$
when $p_1$ is false, and in $\cs_2$, player $2$ reveals $p_2$ to $1$ when $p_2$ is true
and hides $p_2$ from $1$ when $p_2$ is false.
Thus $\min \uu_1(v,\cs) = \min \uu_2(v,\cs)=0$. It can then be verified that no player has profitable deviation from $\cs$ and therefore $\cs \in
\Ne_{\mathsf{pess}}(G)$. Similarly, it can be noted that $\max \uu_1(v,\cs) = \max \uu_2(v,\cs) = 1$. Therefore, $\cs \in \Ne_{\mathsf{opt}}(G)$.
\end{example}

\subsection{Comparison to Related Work} \label{section.related} \label{sec.epistemic}

In this section we compare in more detail our epistemic Boolean games to the two prior proposals in the literature known to us \cite{AgotnesHHW13,HerzigLMS16}, that were already mentioned in the introductory section. We recall that these are imperfect information games (they feature epistemic objectives), however they are not incomplete information games. We also succinctly compare our proposal to an incomplete information game, that is however not a Boolean game \cite{agotnesetal:2011}.

\paragraph*{Comparison to `Boolean Games with Epistemic Goals'.}
In `Boolean games with epistemic goals' \cite{AgotnesHHW13} the set of variables $P$ is partitioned into $|N|=n$ mutually disjoint subsets of variables $P_i$, for $i \in N$, such that the variables in $P_i$ can only be controlled by player $i$. This is as usual in Boolean games, and therefore the strategies played are also as usual, so that a strategy profile is a valuation of all variables. However, the goals are different from the usual in Boolean games, and like ours: these are not merely Boolean goals (formulas in the language $L^B$) whose satisfaction depends on this valuation but these are epistemic goals (the language $L^K$) whose satisfaction depends on what the players know about this valuation. This is where another parameter of their games comes into play: apart from a set $P_i$ of `controlled variables' each player $i$ also has a finite `visibility set' consisting of Boolean formulas, that is, some finite subset of the language $L^B$: those are the propositions whose value that player can observe of the outcome valuation. Such Booleans may involve variables not controlled by player $i$ but by other players $j$. Already, this seems to beg some questions on logical closure, for example if $p\et q$ is in the visibility set but neither variable $p$ nor variable $q$ (where we note that the epistemic goal formulas have the usual compositional semantics, so $K_i (p \et q)$ is true if and only if $K_ip$ and $K_i q$ are true). However, a special case is when the visibility set consists of variables only, which \cite{AgotnesHHW13} call \emph{atomic} games, and this suffices for a comparison with our results. The visibility set determines what is known by the players and thus which epistemic goals are satisfied in a valuation. Because the players altogether control the value of all variables the game is not one of incomplete information (strategies do not depend on an unknown initial valuation) although it is one of imperfect information (over the outcome valuation). The authors then determine that model checking goal formulas is PSPACE-complete and that the existence of Nash equilibria is in PSPACE, although they do not show a lower bound. They also provide an interesting embedding of their epistemic Boolean games into the standard Boolean games by observing that an epistemic goal corresponds to an exponentially larger Boolean goal that is the disjunction of all valuations over which the epistemic goal is uncertain. For example, in some given game, $K_i p$ may abbreviate $(p \et \neg q) \vel (p \et q)$. This is therefore a rather different embedding from our embedding of knowing-whether Boolean observation games into Boolean games wherein the goals remain the same but the set of variables (and thus valuations) is larger: we recall that a $\Kw$ game $G$ for variables $p_i$ is transformed into a Boolean game $B_G$ for variables $\Kw_j p_i$: the knowing-whether formulas are now considered atomic propositions. The goals remain the same in our approach, because knowing-whether goals are Booleans in the language wherein $\Kw_j p_i$ are atomic propositions. 

\paragraph*{Comparison to `Epistemic Boolean Games Based on a Logic of Visibility and Control'.} 
The authors of this work \cite{HerzigLMS16} propose a very expressive logical language and semantics for players controlling the value of propositional variables or observing the value of propositional variables. They also axiomatize this logic. They then use the logic to formalize game theoretical primitives, in particular the existence of equilibria, in an epistemic extension of Boolean games. This formalization allows them to determine the complexity of these games. The problems of determining whether a profile is Nash equilibrium as well as the existence of Nash equilibrium are both in PSPACE.

Their language extension includes knowledge, common knowledge, and for control or observation of propositional variables they propose additional propositional variables. We not only have, for example, a variable $p$, but also $S_i p$, for `player $i$ observes the value of $p$' and $C_i p$ for `player $i$ controls the value of $p$'. But also variables like $C_j S_i p$, for `player $j$ controls whether player $i$ observes $p$', and so on for any stack of $C_j$ or $S_i$ predicates. The interest of these complex propositional variables is that they induce relational Kripke models or can be used to formalize strategies in Boolean games. 

In the epistemic Boolean games of \cite{HerzigLMS16} the {\bf strategies} assign values to variables that are stacks of $S_i$ binding some atom $p$ (so without any $C_i$ or $K_i$), as in $S_i S_j p$, saying that $i$ can see whether $j$ can see the value of $p$, whereas the {\bf goals} are epistemic formulas in the language for such atoms $p$ (so without $S_i$ or $C_i$), as in $K_i p \et \neg K_j p$.\footnote{As $S_i$ stacks are arbitrarily long, there is an infinite set of such atoms to consider. However, the partition among players controlling variables is of a finite subset only of that infinite set. This permits $S_i p$ but not $p$ to be in that finite subset, which would rule out to determine the value of a goal $K_i p$ (as no player gives a value to $p$, that is, no player controls $p$). In their accompanying examples, the finite subset jointly controlled by all players is always subformula closed. This therefore seems an omitted requirement.} One might say that their epistemic Boolean games essentially remain \emph{Boolean} games, because the players still only control the value of variables, but this is only by a (quite smart) stretch of the modeling imagination, because their Boolean variables hard-code arbitrarily complex higher-order multi-agent observations. However, these are not games of incomplete information.

The focus of \cite{HerzigLMS16} is the axiomatization of their logic of visibility and control (it also contains program modalities with primitive operations assigning values to variables). The game-theorical contribution is mainly `proof of concept'.

\paragraph*{Comparison to `Public Announcement Games'.}
`Public announcement games' which are studied in \cite{agotnesetal:2011} and the related `question answer games'  \cite{agotnesetal.qa:2011} also present incomplete games of imperfect information. Expected outcomes are compared with the realist outcome relation. The value of variables is not controlled in any way in \cite{agotnesetal:2011,agotnesetal.qa:2011}, the valuations are fixed. Public announcement games are not Boolean games, because the players' strategies are revelations of any formula, not merely of Booleans. Of course one could consider a class of public announcement games wherein the strategies are restricted to announcing propositional variables only. However, we recall that public announcements are revelations of the same information to all players simultaneously, so this is not as general as our proposal.

\section{Knowing-Whether Boolean Observation Games}\label{section.kw}
\label{sec:knowingwhether}

In this section we show a correspondence between knowing-whether
Boolean observation games ($\Kw$ games) and Boolean games. We provide 
polynomial time reductions that convert a Boolean game to a $\Kw$
 game and vice-versa.

We first recall the definition of Boolean game. We then show that every Boolean game defines a $\Kw$ Boolean observation game, and that every $\Kw$ Boolean observation game defines a Boolean game. These embeddings are different, the first is not the converse of the second.

We further show a utility preserving
equivalence between strategies in Boolean games and equivalence
classes of globally uniform strategies in $\Kw$ games
(Lemmas \ref{lm:BoolObsBiject}, \ref{lm:ObsBoolBiject}). As a
consequence, we prove a correspondence between the existence of Nash
equilibria in Boolean games and the existence of maximal Nash
equilibria in $\Kw$ games, and for both reductions (Theorems
\ref{thm:BoolToKwObs}, \ref{thm:KwObsToBool}). 

We finally show that there always exists a pessimist equilibrium for
$2$-player $\Kw$ games, but not for $\Kw$ games in general: we give an $8$-player $\Kw$ game without a Nash equilibrium (where we do not know if such games exist for between $3$ and $7$ players).

Recall that for any $v \in V$ and $\alpha \in L^\Kw$, $v,s \models
\alpha$ iff $s \models \alpha$ (Prop.~\ref{lemma.ii}). This justifies writing $u_i(s)$ for
the outcome $u_i(v,s)$ of a pointed $\Kw$ game. Now consider a
globally uniform strategy profile $\dot{s}\in \CS^g$. As $\uu_i(v,\dot{s})= u_i(v,\dot{s}(v)) = u_i(v,s)$, this
justifies writing $\uu_i(\dot{s})$ for
the expected outcome of a such a $\Kw$ game. 

\subsection{Boolean Games}

Boolean games have the same parameters
as Boolean observation games but simpler strategies. A Boolean
game is denoted $\boolgame$ to distinguish it from a Boolean observation game
$\obsgame$.

A {\em Boolean game} is a tuple $\boolgame=(N,(P_i)_{i\in
  N},(\gamma_i)_{i \in N})$ where all $\gamma_i \in L^B$ (all goals are Boolean). For $i
\in N$, a \emph{strategy} $v_i$ for player $i$ is a \emph{(local)
valuation} $v_i \subseteq P_i$, where, slightly abusing notation, we
identify a \emph{strategy profile} $v = (v_1,\dots,v_n)$ with a
valuation $v = (v_1 \union \dots \union v_n) \in V$. For Boolean
games, the \emph{outcome function} is denoted $\boolutil$ to
distinguish it from the outcome function $u$ of pointed Boolean observation
games. We define $\boolutil_i(v)=1$ if $v \models \gamma_i$ and
$\boolutil_i(v)=0$ if $v \not\models \gamma_i$. Equilibrium is as for
pointed observation games: a strategy profile $v \in V$ is a
\emph{Nash equilibrium} in $B$ if for all $i \in N$ and $v_i'
\subseteq P_i$, $u_i^\boolgame(v) \geq
u_i^\boolgame(v_i',v_{-i})$. Given $B$, its Nash equilibria are
denoted $\Ne(B)$.

\medskip

Let us emphasize the difference between Boolean games and Boolean
observation games. In Boolean observation games, as in Boolean games,
a player $i$ selects a subset $v_i$ of her local variables
$P_i$. However, in Boolean observation games this subset may be a
different subset $s_i(j) \subseteq P_i$ for each other player
$j$. Also, in Boolean games, excuting strategy $v_i$ means that the
$p_i \in v_i$ \emph{become true} whereas the $p_i \in P_i\setminus
v_i$ become false. Whereas in Boolean observation games, executing
strategy with component $s_i(j)$ means that the $p_i \in s_i(j)$, that
already have an observed truth value, are \emph{revealed} (to $j$).

\subsection{Boolean Games to Knowing-Whether Games}
\label{subsec:BoolToKw}
We construct a $\Kw$ game denoted $G_B$ from a Boolean game $B$ as
follows. Let $\boolgame=(N,(P_i)_{i\in N},(\gamma_i)_{i \in N})$. Then
$G_B := (N,(P_i)_{i\in N},(\beta_i)_{i \in N})$ where each $\beta_i
:=\lambda(\gamma_i)$ is defined as follows. Let ${i^+} := i+1$ for $i
= 1,\dots,n-1$ and ${n^+} := 1$. Then $\lambda: L^B \imp L^\Kw$ is
inductively defined as: for all $i$, $p_i \in P_i$, $\lambda(p_i) :=
\Kw_{i^+}p_i$, and (trivially) $\lambda(\neg\alpha) := \neg
\lambda(\alpha)$ and $\lambda(\alpha_1\vel\alpha_2) :=
\lambda(\alpha_1) \vel \lambda(\alpha_2)$.
Note that $B$ and $G_B$ are defined for the same players and variables.

Given a strategy profile $v \in V$ for $B$, we define globally uniform strategy profile $\dot{s^v} \in \CS^g$ for $G_B$ such that for all $i \in N$ and $p_i \in P_i$: $p_i \in s^v_i(i^+)$ if $p_i \in v$; $s^v_i(i)= P_i$; and for all $j \in N$ with $j \neq i,i^+$, $s^v_i(j) = \emptyset$. Note that for all valuations $w$, including $v$, $\dot{s^v}(w) = s^v$. Notation $s^v$ is therefore not to be confused
with notation $\cs(v)$ for uniform profiles $\cs$. In this section we will show how the $v$ strategy for $B$ corresponds to the $\dot{s^v}$ strategy for $G_B$.

Note that for all $i \in N$,
  we have $|\beta_i|=\bigO(|\gamma_i|)$ where $|\beta_i|$ and
  $|\gamma_i|$ denote the size of (number of symbols in) $\beta_i$ and
  $\gamma_i$ respectively. Thus given $B$, the associated $\Kw$ game
  $G_B$ can be constructed in polynomial time.

\begin{example} \label{ex.btogb}
We illustrate how to construct a $\Kw$ game $G_B$ from a Boolean game $B$. (We will not analyze the equilibria of the game, if any.) Consider \[ B = (\{1,2,3\}, (\{p_1\},\{p_2\},\{p_3\}), (p_1 \eq p_3, p_3 \imp p_1, \neg p_1 \imp p_2)) \] Then $G_B$ has the same variables $p_1,q_1,p_2,p_3$ but different goals, namely $\Kw_2 p_1 \eq \Kw_1 p_3$ for player 1,  $\Kw_1 p_3 \imp \Kw_2 p_1$ for player 2, and $\neg \Kw_2 p_1 \imp \Kw_3 p_2$ for player 3. 

In the Boolean game, for player 1 to obtain her goal $\gamma_1 = p_1
\eq p_3$, player $1$ has to make $p_1$ true, it does not matter
whether player $2$ makes $p_2$ true or false, and player $3$ has to
make $p_3$ true. In the $\Kw$ game, in order to achieve the goal
$\beta_1 = \Kw_2 p_1 \eq \Kw_1 p_3$, player $1$ has to reveal $p_1$ to
player $2$, it does not matter whether player $2$ reveals $p_2$ to
player $3$, and player $3$ has to reveal $p_3$ to player $1$, and all
three do this independently from the valuation. Because in fact, for
example, player $1$ reveals the \emph{value} of $p_1$ to player $2$,
but what the value is does not matter as the outcome of a $\Kw$ game
is independent from the valuation. So they players execute globally
uniform strategies. More precisely, in order to ensure $\beta_1$
globally uniform strategy $\dot{s}$ is required such that
$s_1(2)=\{p_1\}$, $s_3(2)$ does not matter, and $s_3(1)=\{p_3\}$.
\end{example}

\begin{restatable}{lemma}{lmboolobsoutcome}
    \label{lm:BoolObsOutcome}
    Let $\obsgame_\boolgame$ be the $\Kw$ game associated with
    the Boolean game $\boolgame$.  For all $i \in N$, for all $w \in
    V$, $s^w \models \lambda(\gamma_i) \text{ iff } w \models \gamma_i$.
\end{restatable}
\begin{proof} 
This is shown by induction where only the base case is not-trivial. For that, we have that $s^w \models \Kw_{i^+}p_i$ iff $w \models p_i$ by definition of the embedding. \end{proof}
Therefore, for any $i \in N$ and $v \in V$: $\uu_i(v,\dot{s^w}) = u_i(v,s^w) = u_i(s^w) = u^B_i(w)$. This correspondence allows us to relate Nash equilibria in $B$ to Nash
equilibria in $G_B$. The result uses an interesting property of
$\Ne_{\mathsf{max}}$ in $\Kw$ games (this property does not hold for,
e.g., the pessimist outcome relation).

\begin{restatable}{lemma}{lmkwglobalne}
%\begin{lemma}
  \label{lm:kwglobalNE}
  Let $\Kw$ game $\obsgame$ be given. Let $\cs \in
  \Ne_{\mathsf{max}}(\obsgame)$ and $v \in V$. Let $s = \cs(v)$. Then
  $\dot{s}\in \Ne_{\mathsf{max}}(\obsgame)$.
%\end{lemma}
\end{restatable}
\begin{proof}
  Consider an arbitrary $v \in V$ and let $s=\cs(v)$. Suppose that
  $\cs \in \Ne_{\mathsf{max}}(\obsgame)$, we claim that the globally
  uniform strategy profile $\dot{s} \in \Ne_{\mathsf{max}}(\obsgame)$.
  Suppose not. Then there exists $i \in N$, $w \in V$ and $s_i' \in
  S_i$ such that $\util_i(w, (s_i',\dot{s}_{-i}(w))) >
  \util_i(w,\dot{s}(w))$. This implies that $w,
  (s_i',\dot{s}_{-i}(w))\models \gamma_i$ and $w, (\dot{s}(w))
  \not\models \gamma_i$. Since for all $j \in N$,
  $\dot{s}_j(w)=\cs_j(v)$, we have $w, (s_i',\cs_{-i}(v)) \models
  \gamma_i$. Since $\gamma_i \in L^\Kw$,
%%  from Lemma~\ref{lm:kwtruth}
  we have $v, (s_i',\cs_{-i}(v))\models \gamma_i$. Also, since
  $w,\dot{s}(w)\not\models \gamma_i$,
  %% and from Lemma~\ref{lm:kwtruth}
  we have $v, \dot{s}(w) \not\models \gamma_i$ and by definition of
  $\dot{s}$, we have $v, \cs(v) \not\models \gamma_i$.  Therefore,
  $\cs \not\in \Ne_{\mathsf{max}}(\obsgame)$ which gives the required
  contradiction.
\end{proof}

% As $\dot{s}$ is global, we therefore immediately obtain:
\begin{corollary}
\label{cor:kwglobalNE}
  Let $\Kw$ game $\obsgame$ be given. If $\Ne_{\mathsf{max}}(\obsgame) \neq \emptyset$ then $\Ne^g_{\mathsf{max}}(\obsgame) \neq \emptyset$.
\end{corollary}

\noindent {\bf An Equivalence Relation over Global 
Strategy Profiles.} Recall that every $\cs \in \CS^g$ is of the
form $\dot{s}$ where $s \in S(\obsgame_\boolgame)$. We define an
equivalence relation over $\CS^g$ in
  $\obsgame_\boolgame$ as follows. For
$i \in N$, $\dot{s}_i \equiv_i \dot{t}_i$ iff
$s_i(i^+)=t_i(i^+)$. For $\dot{s}, \dot{t} \in
\CS^g$, we define $\dot{s} \equiv \dot{t}$ iff for
all $i \in N$, $\dot{s}_i \equiv_i \dot{t}_i$. Let
$\CS^g\qslash \equiv$ denote the set of
equivalence classes and $[\cs]$ denote the equivalence class
containing $\cs \in \CS^g$.

%\begin{restatable}{lemma}{lmkwgolbalpayoff}
\begin{lemma}
  \label{lm:kwglobalpayoff}
Given $\cs \in \CS^g$, for all $\cstwo \in [\cs]$,
for all $i \in N$, for all $v \in V$,
$\util_i(v,\cs(v))=u_i(v,\cstwo(v))$.
%$u_i(\cs)(v)=u_i(\cs')(v)$.
\end{lemma}
%\end{restatable}
%\lmkwgolbalpayoff*
\begin{proof}
  Let $\cs = \dot{s}$ and $\cstwo=\dot{t}$. For all $i \in N$, since
  $\cstwo \in [\cs]$, we have $s_i(i^+)=t_i(i^+)$. By
  induction of the structure of $\gamma_i$, we can prove the
  following: for all $v \in V$, for all $i \in N$ and for all
  $\gamma_i \in L^\Kw$, we have $v,(\dot{s}(v)) \models \gamma_i$ iff
  $v, (\dot{t}(v)) \models \gamma_i$. This implies that
%%  for all $\cs  \in \CS^g(\obsgame_\boolgame)$, for all $\cs' \in [\cs]$,
  for all $i \in N$, for all $v \in V$,
  $\util_i(v,\dot{s}(v))=u_i(v,\dot{t}(v))$.
\end{proof}

\noindent {\bf An Outcome Preserving Bijection.} We now show that
there is an outcome preserving bijection $\chi$ between strategy
profiles in $\boolgame$ and equivalence classes in
$\CS^g\qslash \equiv$. For a Boolean game
$\boolgame$, and $v \in V$, $\chi(v) = [\dot{s^v}]$.

%\begin{restatable}{lemma}{bijectblob}
\begin{lemma}
  \label{lm:BoolObsBiject}
  Given a Boolean game $\boolgame$, the function $\chi: V \to
  \CS^g\qslash \equiv$ is a bijection.
\end{lemma}
%\end{restatable}

\begin{proof}
  Given $\dot{s} \in \CS^g$, consider $v \in V$ defined as follows:
  for all $i \in N$ and $p_i \in P_i$, $p_i \in v$ iff $p_i \in
  s_i(i^+)$. We then have $\chi(v)=[\dot{s}]$ and therefore
  $\chi$ is onto. For $v, w \in V$ such that $v \neq w$, there
  exists $i \in N$ , there exists $p_i \in P_i$ such that $p_i \in v$ and
  $p_i \not\in w$. Thus, for $\chi(v)=[\dot{s}]$ and
  $\chi(w)=[\dot{t}]$, we have $\dot{s} \not \equiv_i \dot{t}$,
  which implies that $\dot{s} \not\equiv \dot{t}$. Therefore, $\chi$
  is a bijection.
\end{proof}

Consequently, we can prove a correspondence between Nash equilibria existence.

\begin{restatable}{theorem}{thmbooltoobs}
  \label{thm:BoolToKwObs}
Let $B$ be a Boolean game. Then $\Ne_{\mathsf{max}}(\obsgame_\boolgame) \neq
  \emptyset$ iff $\Ne(B) \neq \emptyset$.
\end{restatable}
\begin{proof}
  ($\Leftarrow$) We argue that if $w \in \Ne(\boolgame)$ then
  $\dot{s^w} \in \Ne_{\mathsf{max}}(\obsgame_\boolgame)$. Suppose not,
  then there exists $i \in N$, $v \in V$ and $t_i \in S_i$ such that
  $\util_i(v,(t_i,\dot{s^w}_{-i}(v))) > \util_i(v,\dot{s^w}(v))$. Let
  $w'=\chi^{-1}([\dot{t}_i,\cs_{-i}])$. 
  From Lemmas \ref{lm:BoolObsOutcome}, \ref{lm:kwglobalpayoff} and
  \ref{lm:BoolObsBiject} it follows that
  $\boolutil_i(w')=\util_i(v,(\dot{t}_i,\cs_{-i})(v)) >
  \util_i(v,\dot{s^w}(v))=\boolutil_i(w)$ for all $v \in V$. Therefore
  $w \not\in \Ne(\boolgame)$ which is a contradiction.
  
  ($\Rightarrow$) Suppose $\Ne_{\mathsf{max}}(\obsgame_\boolgame) \neq
  \emptyset$. By Lemma~\ref{lm:kwglobalNE} there exists $\cs \in
  \CS^g$ such that $\cs \in \Ne_{\mathsf{max}}(\obsgame_\boolgame)$.
  Let $w=\chi^{-1}([\cs])$. We claim that $w \in
  \Ne(\boolgame)$. Suppose not, then there exists $i \in N$ and $w_i'$
  such that $\boolutil_i(w_i',w_{-i})>\boolutil_i(w)$. Let $w' =
  (w_i',w_{-i})$. Note that by definition $w \neq w'$. From Lemma
  \ref{lm:BoolObsOutcome} it follows that for all $v$ we have
  $\boolutil_i(w') = \util_i(v,\dot{s^{w'}}(v))$. From Lemmas
  \ref{lm:BoolObsOutcome}, \ref{lm:kwglobalpayoff} and
  \ref{lm:BoolObsBiject} it follows that $\boolutil_i(w) =
  \util_i(v,\cs(v))$ for all $v \in V$. Therefore for all $v \in V$, $
  \util_i(v,\dot{s^{w'}}(v)) = \boolutil_i(w') > \boolutil_i(w) =
  \util_i(v,\cs(v))$ which contradicts the fact that $\cs \in
  \Ne_{\mathsf{max}}(\obsgame_\boolgame)$.
\end{proof}

\medskip
\noindent {\bf Other ways to get a Kw Game from a Boolean Game.}  
Let us imagine our $n$ players sitting round a table numbered in clockwise fashion.
In the embedding  $\lambda: L^B \imp L^\Kw$ with basic clause \[ \lambda(p_i) :=
\Kw_{i^+}p_i,\] every player $i$ reveals the value of her observed variable $p_i$ to her left neighbour (while other players observe her doing that). There are many other embeddings that would serve equally well to obtain our results. For example, every player $i$ could reveal her variable to her {\bf right} neighbour. This would be a $\lambda'$ with basic clause \[ \lambda'(p_i) :=
\Kw_{i^-}p_i \] where $i^-$ is $i-1$ except for $1^- := n$. A more interesting embedding would be every player publicly announcing $p_i$ to all other players. We then have a $\lambda''$ for which \[\lambda''(p_i) :=
\Et_{j\in N} \Kw_j p_i.\]

\subsection{Knowing-Whether Games to Boolean Games}  
\label{subsec:KwToBool}
\noindent {\bf A Kw Game to a Boolean Game.} We now construct a Boolean
game denoted $B_G$ from a knowing-whether Boolean observation game $G$.
Let $\obsgame=(N,(P_i)_{i\in N},(\gamma_i)_{i \in N})$.
Assume that the goals $\gamma_i$ do not contain trivial constituents
$\Kw_i p_i$.\footnote{As such $\Kw_i p_i$ are always true, this would
otherwise cause a problem in the translation, because the players in
the constructed Boolean game would then be able to control the value
of propositional variables $\Kw_i p_i$ (which is undesirable), unlike
the players in the given $\Kw$ game. One can also address this
formally, without assumptions, with an inductively defined translation
mapping $\Kw_i p_i$ to $\top$.}
Then $B_G :=
(N,(Q_i)_{i\in N},(\gamma_i)_{i \in N})$ where for all $i \in N$,
$Q_i = \{ \Kw_j p_i \mid p_i \in P_i, i \neq j\}$.
We view $\Kw_j p_i$, for each $i$ and $j$ with $i \neq j$, as atomic
propositions in $B_G$.  Let $Q = \Union_{i \in N} Q_i$.

\medskip

\noindent{\bf Observation.} Both $G$ and $B_G$ are defined over
  the same set of players and goal formulas. The number of
  variables in $B_G$ for each $i \in N$ is $|Q_i| = (n-1)|P_i|$.
  Thus given $G$, the associated Boolean game $B_G$ can be constructed
  in polynomial time. Also, note that $B_{G_B} \neq B$ and $G_{B_G}
  \neq G$, the constructions are unrelated. Let us give an example of
  that.

\begin{example}
As an illustration to construct a Boolean game from a $\Kw$ game, let
us take the $\Kw$ game just constructed in Example~\ref{ex.btogb}. We
recall that  \[G_B = (\{1,2,3\}, (\{p_1,q_1\}, \{p_2\}, \{p_3\}),  (\gamma_1, \gamma_2, \gamma_3))\] where 
\begin{itemize}
\item $\gamma_1=\Kw_2
p_1 \eq \Kw_1 p_3$,
\item $\gamma_2=\Kw_1 p_3 \imp \Kw_2 p_1$,
\item $\gamma_3=\neg \Kw_2 p_1 \imp \Kw_3
p_2$.
\end{itemize}

The Boolean game $B_{G_B}$ constructed from that has the same
goals but has more variables, namely $\Kw_i p_j$ for all $i,j\in N$
with $i \neq j$ and for all $p_j \in P_i$, that is: $\Kw_1 p_2$,
$\Kw_1 p_3$, $\Kw_2 p_1$, $\Kw_2 q_1$, $\Kw_2 p_3$, $\Kw_3 p_1$,
$\Kw_3 q_1$, $\Kw_3 p_2$. Therefore $B_{G_B}$ has more variables than
$B$. The constructions are not each other's converse. However, in
order to realize the goals of $B_{G_B}$ the players only need to
assign a value to variables $\Kw_i p_j$ occurring in the goal
formulas, so with respect to playing {\bf this game} the extra
variables do not play a role. After replacing $\Kw_2 p_1$ by $p_1$,
etcetera for other variables ocurring in goal formulas, we recover the
original Boolean game for, however, far more variables that are not
used in goals.
\end{example}

Let $W = \mathcal{P}(Q)$ be the set of valuations over $Q$. We define
a function $\eta: \CS^g \to W$ and argue that it is a bijection which
is outcome equivalent.
Given $\dot{s} \in \CS^g$, define $w=\eta(\dot{s})$ as
follows: for $i \in N$, $\Kw_j p_i \in \eta(\dot{s})_i$ iff $p_i \in
s_i(j)$.

\begin{restatable}{lemma}{lmobsbooloutcome}
  \label{lm:ObsBoolOutcome}
Let $\boolgame_\obsgame$ be the Boolean game associated with the
$\Kw$ game $\obsgame$. For all $i \in N$, for all $s \in S$ and
for all $\gamma_i$, $s \models \gamma_i \text{ iff } \eta(\dot{s})
\models \gamma_i$.
%\gamma_i$.
\end{restatable}
\begin{proof}
This is shown by induction using as the base case that $s \models
\Kw_j p_i$, iff $\eta(\dot{s}) \models \Kw_j p_i$
The other cases are trivial.
\end{proof}
It therefore also follows, similarly to the above, that $\uu_i(\dot{s}) = u_i(s) = u^B_i(w^s)$.
\begin{lemma}
  \label{lm:ObsBoolBiject}
Given a $\Kw$ game $\obsgame$, let $\boolgame_\obsgame$ be the
associated Boolean game. The function $\eta: \CS^g \to W$ is a
bijection.
\end{lemma}
\begin{proof}
For an arbitrary $w \in W$, consider $\dot{s} \in \CS^g$ defined as
follows. For all $i \in N$, and for all $p_i \in P_i$, $p_i \in
s_i(j)$ iff $i=j$ or $\Kw_j p_i \in w_i$. 
By definition, $\eta(\dot{s})=w$ and thus $\eta$ is onto.

Consider $\dot{s}, \dot{t} \in \CS^g$ where $\dot{s} \neq
\dot{t}$. Then there exists $i,j \in N$ with $i \neq j$ and there
exists $p_i \in P_i$ such that $p_i \in s_i(j)$ and $p_i \not\in
t_i(j)$. This implies that $\Kw_j p_i \in \eta(\dot{s})_i$ and $\Kw_j
p_i \not\in \eta(\dot{t})_i$. Therefore $\eta$ is a bijection.
\end{proof}

\noindent 
{\bf Non-global Uniform Strategies as Mixed Strategies for Boolean Games.} 
We allow ourselves a little detour. We can straightforwardly adjust the function $\eta$ mapping globally uniform strategy profiles of the $\Kw$ game to valuations that are strategy profiles of the Boolean game, to a function mapping arbitrary uniform strategy profiles of the $\Kw$ game $G$ to {\em mixed strategy profiles} of the Boolean game $B_G$. We simply define the `revised $\eta$' on the level of strategy profiles $s \in S$. Given a uniform strategy profile $\cs \in \CS$, for each $s \in S$ such that $\cs(v)=s$ for some $v \in V$, we let $\pi(s) := |\{v \in V \mid
\cs(v) = s\}| / 2^{|P|}$.  Note that $2^{|P|} = |V|$. So $\pi(s)$ is
the probability that a valuation is mapped in $s$, given $\cs$. We can now define a mixed
strategy profile $w^\cs$ of the Boolean game $B_G$ as
the one executing each $s \in S$ with probability $\pi(s)$. We defer the investigation of embeddings into a mixed equilibrium to future research and for now restrict ourselves to a relevant example, finally closing the loop with matching pennies.

\begin{example} \label{ex:noNElast}
Once more we recall Example~\ref{example.pennies} on page \pageref{example.pennies} about the pennies that do not match, already further developed in Example~\ref{ex:noNE}, wherein it was shown that this game does not have a Nash
equilibrium with globally uniform strategies for the pessimist and optimist outcome relation, but has a Nash equilibrium with uniform
strategies that are not globally uniform: the uniform strategy profile $\cs=(\cs_1,\cs_2)$
where in $\cs_1$, player $1$ reveals $p_1$ to $2$ when $p_1$ is true
and hides $p_1$ from $2$ when $p_1$ is false, and in $\cs_2$ player
$2$ reveals $p_2$ to $1$ when $p_2$ is true and hides $p_2$ from $1$
when $p_2$ is false.

When translating this game into a Boolean game, we can now observe
that this uniform strategy $\cs$ becomes a mixed strategy $\eta(\cs)$
where player $1$ randomly chooses between revealing or hiding her
propositional variable $\Kw_2 p_1$ and where player $2$ randomly
chooses between revealing or hiding his propositional variable $\Kw_1
p_2$. To realize that this is indeed random it is important to observe
that in Example~\ref{example.pennies} the probability of observing
$p_1$ or $\neg p_1$ was determined by Odd flipping its penny before
privately watching the outcome under the dice cup, and similarly for
$p_2$ and Even. So after all, for those who may have wondered, there
was a reason for setting up the experiment just like that.
\end{example}

We continue with a relevant result for the maximal outcome
relation.

\begin{restatable}{theorem}{thmobstobool}
\label{thm:KwObsToBool}
%  \label{thm:obsboolNE}
Let $\obsgame$ be a $\Kw$ game. Then $\Ne(\boolgame_{\obsgame}) \neq \emptyset$ iff $\Ne_{\mathsf{max}}(\obsgame) \neq \emptyset$.
\end{restatable}
\begin{proof}
 ($\Leftarrow$) Suppose $\Ne_{\mathsf{max}}(\obsgame_\boolgame) \neq
  \emptyset$. By Lemma~\ref{lm:kwglobalNE} there exists $\cs \in
  \CS^g$ such that $\cs \in
  \Ne_{\mathsf{max}}(\obsgame_\boolgame)$. Let $w =\eta(\cs)$, we
  argue that $w$ is a Nash equilibrium. Suppose not, there exists $i
  \in N$, there exists $w_i'$ such that $\boolutil_i(w_i',w_{-i}) >
  \boolutil_i(w_i,w_{-i})$. Consider the globally uniform strategy
  profile $\cstwo=\eta^{-1}(w_i',w_{-i})$ (this is well defined by
  Lemma \ref{lm:ObsBoolBiject}). From Lemmas \ref{lm:ObsBoolOutcome}
  and \ref{lm:ObsBoolBiject} it follows that $\util_i(v, \cstwo(v)) =
  \boolutil_i(w_i',w_{-i}) > \boolutil_i(w) = \util_i(v,\cs(v))$. This
  implies that $\cs \not\in \Ne_{\mathsf{max}}(\obsgame_\boolgame)$
  which is a contradiction.

  ($\Rightarrow$) Suppose $w \in \Ne(\boolgame_\obsgame)$. Let
  $\cs=\eta^{-1}(w)$, we claim that $\cs \in
  \Ne_{\mathsf{max}}(\obsgame)$. Suppose not, then there exists $i \in
  N$, $v \in V$ and $t_i \in S_i$ such that
  $\util_i(v,(t_i,\cs_{-i}(v))) > \util_i(v,\cs(v))$. Let
  $w'=\eta(\dot{t_i},\cs_{-i})$.
  From Lemmas \ref{lm:ObsBoolOutcome} and \ref{lm:ObsBoolBiject} it
  follows that $\boolutil_i(w') = \util_i(v,(\dot{t_i},\cs_{-i})(v)) >
  \util_i(v,\cs(v)) = \boolutil_i(w)$. This implies that $w \not\in
  \Ne(\boolgame_\obsgame)$ which is a contradiction.
\end{proof}

\section{Existence of Nash Equilibrium}\label{section.exist}
\label{sec:existsNE}
In this section we focus on the question of existence of Nash
equilibria for observation games and identify various subclasses in which 
a Nash equilibrium is guaranteed to exist.

\subsection{Existence of Pessimist Nash Equilibrium in
  Knowing-Whether Games}
  \label{subsec:existsNEKW}

Example \ref{ex:noNE} shows that in the
$\Kw$ fragment a maximal Nash equilibrium is not guaranteed to exist
even for two-player games. It is natural to ask if a similar
observation holds for pessimist Nash equilibrium. We first show that
for two-player $\Kw$ games, a pessimist Nash equilibrium always exists
(Proposition~\ref{prop:KWNEexists}). However, for general $\Kw$ games,
existence is not guaranteed.  Example~\ref{ex:KWnoNEpess} gives an
$8$-player $\Kw$ game without a Nash equilibrium.

\begin{restatable}{proposition}{propKWNEexists}
  \label{prop:KWNEexists}
%\begin{proposition}
All two-player $\Kw$ games have a pessimist Nash
equilibrium.
%\end{proposition}
\end{restatable}
\begin{proof}
We construct a uniform strategy profile $(\cs_1^*,\cs_2^*)$ as
follows. For $i \in \{1,2\}$, let $\obar{\imath}$ denote the player
such that $\obar{\imath} \neq i$. If $i \in \{1,2\}$ has a uniform
strategy $s_i$ that is dominant, then set $\cs_i^*(v)=s_i$ for all $v
\in V$ and let $\cs_{\obar{\imath}}^*$ be the best response to
$\cs_i^*$. It can be verified that $\cs^*$ as defined above is a Nash
equilibrium.

If neither player has a uniform strategy that is dominant, then we have
the following
\begin{itemize}
\item For all $s_1 \in S_1$, there exists $s_2 \in S_2$ such that $v,
  (s_1,s_2) \not\models \gamma_1$ for all $v \in V$.
  \item For all $t_2\in S_2$, there exists $t_1 \in S_1$ such that $v,
    (t_1,t_2) \not\models \gamma_2$ for all $v \in V$.
\end{itemize}
For two player games there is a bijection between the set of strategies
$S_i$ and the set of (local) valuations $V_i$ (one can think of a
strategy as deciding for each proposition whether to reveal to the
other player). Therefore for each $t_1$ and $s_2$ as described above,
we can set $\cs_1^*(v^1)=t_1$ and $\cs_2^*(v^2)=s_2$ appropriately for
some $v^1$ and $v^2$.

To see that $(\cs_1^*,\cs_2^*)$ is a Nash equilibrium, note that for
all $i \in \{1,2\}$ and for all $v \in V$, $\min\uu_i(v, \cs^*) = 0$. Also, for
all $v \in V$, $\min
\uu_i(v,(\cs_i',\cs^*_{-i})) =0$ due to the above condition.
\end{proof}

However, for more than two players a $\Kw$ game need not have a
pessimist Nash equilibrium. We present a counterexample for eight
players.

%\begin{restatable}{example}{propKWnoNE}
\begin{example}
 \label{ex:KWnoNEpess}
      Consider the observation game $\obsgame$ where
    $N=\{1,2,\ldots,8\}$ and $P_i=\{p_i\}$ for $i \in N$. Player 8
    acts as an ``observer'' whose goal $\gamma_8=\top$.  To specify
    the goals of the other players we use the following formulas. \vspace{-.1cm}
\[\begin{array}{ll}
A = \Kw_3 p_1 \wedge \Kw_4 p_1 \quad & B = \Kw_3 p_1 \wedge \neg
      \Kw_4 p_1\\
C=\neg \Kw_3 p_1 \wedge \Kw_4 p_1 &D= \neg \Kw_3 p_1
      \wedge \neg \Kw_4 p_1
\end{array}\]

    The main idea is to exploit that player 1 controls a
    single variable. Therefore in any uniform strategy 
player 1 can
    choose to satisfy at most two of $A,B,C,D$. E.g.,
 ``when $p_1$ is true reveal $p_1$ to 3 and
    4 ($A$), when $p_1$ is false reveal $p_1$ to 4 but not to 3 ($C$).

    \medskip
    
%    \noindent {\bf Definition of the goal formulas.}
\paragraph*{Definition of the Goal Formulas.} For players
$\{2,\ldots,7\}$ the goal formulas are as follows.
%%

%    {\small
      \[\begin{array}{lll}
        \gamma_2 = &((C \vee A) \imp \Kw_8 p_2)
        \wedge ( \neg (C \vee A) \imp \neg \Kw_8 p_2),\\
        \gamma_3=  &(((B \vee D) \imp \Kw_8 p_3) \wedge (\neg (B \vee D) \imp \neg\Kw_8 p_3)) \vee\\
        &((A \vee D) \wedge \Kw_8 p_2 \wedge \neg \Kw_8 p_3),\\
        \gamma_4= &((D \vee C) \imp \Kw_8 p_4) \wedge (\neg (D
        \vee C) \imp \neg\Kw_8 p_4) \wedge\\
        & (((B \vee C)\wedge \Kw_8 p_5) \imp \neg \Kw_8 p_4),\\
        \gamma_5= &(((A \vee B) \imp \Kw_8 p_5) \wedge
        (\neg (A \vee B) \imp \neg\Kw_8 p_5)) \vee\\
        & ((A \vee C) \wedge \Kw_8 p_7 \wedge \neg \Kw_8 p_5), \\
        \gamma_6= &(((A \vee D) \imp \Kw_8 p_6) \wedge (\neg (A \vee D) \imp \neg\Kw_8 p_6)) \vee\\
        &((A \vee B) \wedge (\Kw_8 p_3 \vee \Kw_8 p_7) \wedge \neg \Kw_8 p_6),\\
        \gamma_7= &((B \vee C) \imp \Kw_8 p_7) \wedge (\neg (B \vee C) \imp \neg\Kw_8 p_7)\\
        %% \wedge\\
        %% &(\Kw_8 p_4 \imp \neg \Kw_8 p_7).\\
    \end{array}
      \]

    The goal of player $1$ is defined as $\gamma_1 := \Vel_{j=1}^6
    \alpha_j$ where
    % \alpha_1 \vee \alpha_2 \vee \alpha_3 \vee \alpha_4$ where
    $\alpha_1 = \Kw_8 p_2 \wedge D$, $\alpha_2 = \Kw_8 p_3 \wedge A$,
    $\alpha_3 = \Kw_8 p_4 \wedge B$, $\alpha_4 = \Kw_8 p_5 \wedge
    C$, $\alpha_5 = \Kw_8 p_6 \wedge B$ and $\alpha_6 = \Kw_8 p_7
    \wedge A$.
    We will now verify that $\Ne_{\mathsf{pess}}(G) =
    \emptyset$. 

The goals of the players (except 8) involve assertions about whether
players $2,
\ldots, 7$ reveal the proposition that they control to player 8 along
with whether 1 reveals $p_1$ to players 3 and 4.
For the purpose of this example, note
that for all $j \in \{2, \ldots,7\}$, for all $s_j \in S_j$ and
$k \neq 8$, the value of $s_j(k)$ is irrelevant.

We now argue that $\Ne_{\mathsf{pess}}(\obsgame) = \emptyset$. To
simplify the presentation, we split the reasoning into two
parts. First we argue that no globally uniform strategy of player 1
can be part of a pessimist Nash equilibrium in $G$. In the second part
we extend this to cover all uniform strategy profiles.

\medskip

%\noindent {\bf Globally uniform strategies of player 1.}

\paragraph*{Globally Uniform Strategies of Player 1.} We show
that for all uniform strategy profiles $\cs \in \CS$, if $\cs_1$ is
globally uniform then $\cs \not\in \Ne_{\mathsf{pess}}(G)$. In other
words, no uniform strategy profile in $\obsgame$ with a globally
uniform strategy for player $1$ can be a pessimist Nash equilibrium.

Consider an arbitrary uniform strategy profile $\cs \in \CS$ where $\cs_1= \dot{s}_1 \in
\CS_1^g$. The uniform strategy $\dot{s}_1$ satisfies exactly one of the
formulas $A, B, C, D$. From the goal formulas we can see that there
exists a non-empty subset of players $X \subseteq \{2, \ldots, 7\}$
such that for all $j \in X$, $\min\uu_j(v,
  (\dot{s}_1,\dot{s}^\forall_j,\cs_{N-\{1, j\}}))= 1$ for all $v$.
Thus if $\cs \in \Ne_{\mathsf{pess}}(\obsgame)$ then $\uu_j(v,\cs)= 1$
for all $j \in X$ and for all $v$. From the goal formulas it also
follows that there exists $\dot{s'}_1 \neq \dot{s}_1$ such that
$\min\uu_1(v, (\dot{s'}_1,\dot{s}^\forall_j,\cs_{N-\{1, j\}}))= 1$ for all
$j \in X$ and for all $v$.

In Table \ref{tab:one} we list all such possibilities.
The first column in Table \ref{tab:one} lists the formula in $A, B, C,
D$ that is satisfied by a globally uniform strategy $\dot{s}_1$. The
second column lists the players $j \in \{2, \ldots, 7\}$ who can
ensure an outcome 1 with $\dot{s}^\forall_j$ given $\dot{s}_1$. The
third column gives the corresponding formulas in $A, B, C, D$ that
player 1 should satisfy to achieve an outcome 1. For example suppose
$\dot{s}_1$ satisfies $A$ (first row), players 2, 5 and 6 can reveal
their proposition to player 8 to ensure an outcome of 1 given
$\dot{s}_1$. Player 1 can then choose to satisfy $D$ (corresponding to
$\alpha_1$), $C$ (corresponding to $\alpha_4$), $B$ (corresponding to
$\alpha_5$), respectively to achieve an outcome of 1. Using
Table \ref{tab:one} it can be verified that any $\cs \in \CS$ where
$\cs_1$ is a globally uniform strategy is not a pessimist Nash
equilibrium.

\begin{table}
  \centering
  \begin{tabular}{ | l | c | r| }
    \hline
    $A$ & $2, 5, 6$ & $D, C, B$\\ \hline
    $B$ & 3, 5, 7 & $A, C, A$\\ \hline
    $C$ & 2, 4, 7 & $D, B, A$ \\ \hline
    $D$ & 3, 4, 6 & $A, B, B$\\
    \hline
  \end{tabular}
        \caption{Uniform strategies for player 1. Explanations are given in the text. \label{tab:one}}
        \end{table}

%\noindent {\bf Arbitrary uniform strategies of player 1.}
\paragraph*{Arbitrary Uniform Strategies of Player 1.}
Next, note that since player 1 controls a single proposition $p_1$,
any uniform strategy (of player 1) can satisfy at most two of
$A,B,C,D$. For example, consider the uniform strategy $\cs_1$: when
$p_1$ is true reveal $p_1$ to 3 but not to 4 ($B$), when $p_1$ is
false do not reveal $p_1$ to 3 and do not reveal $p_1$ to 4 ($D$).
 %% E.g., ``when $p_1$ is true reveal $p_1$ to 3 and 4 ($A$),
 %% when $p_1$ is false reveal $p_1$ to 4 but not to 3 ($C$).
Using an argument similar to the one above we can show that if $\cs
\in \Ne_{\mathsf{pess}}(G)$ then $\min\uu_1(v,\cs)=1$ for all $v \in
V$.

Now consider the uniform strategy $\cs_1$ which is mentioned above for
player 1 that satisfies $B$ or $D$. We can argue that there is no
uniform strategy $\cs^*$ with $\cs^*_1 = \cs_1$ such that $\cs^* \in
\Ne_{\mathsf{pess}}(G)$. Given the uniform strategy $\cs_1$ of player 1,
we have the following.

\begin{itemize}
\item Player 2 can ensure an outcome of 1 (for all $v \in V$) by
  not revealing $p_2$ to player 8.
  \item Player 3 can ensure an outcome of 1 (for all
    $v \in V$) by revealing $p_3$ to player 8. But this would in
    turn imply that player 1 can satisfy $\alpha_2$ by deviating to a
    uniform strategy that satisfies $A$.
    \item If player 4 reveals $p_4$ to player 8 then player 1 can
      satisfy $\alpha_3$ by deviating to a uniform strategy that
      satisfies $B$.
    \item If player 5 reveals $p_5$ to player 8 then player 1 can
      satisfy $\alpha_4$ by deviating to a uniform strategy that
      satisfies $C$.
    \item If player 6 reveals $p_6$ to player 8 then player 1
      can satisfy $\alpha_5$ by deviating to a uniform strategy
      that satisfies $B$.
    \item If player 7 reveals $p_7$ to player 8 then player 1
      can satisfy $\alpha_6$ by deviating to a uniform strategy
      that satisfies $A$.
\end{itemize}

Recall that the goal of player 1 is $\gamma_1 := \Vel_{j=1}^6
\alpha_j$. Now consider any uniform strategy profile $\cs^*$ where
$\cs^*_1 = \cs_1$, players 2 and 3 are playing their best responses
and players 4,5,6 and 7 do not reveal the proposition that they
control to player 8. Then $\alpha_3$, $\alpha_4$, $\alpha_5$ and
$\alpha_6$ are not satisfied by $\cs^*$. From items 1 and 2 above, we
have that in $s_3^*$, player 3 reveals $p_3$ to player 8 and
subsequently player 1 has a profitable deviation to a uniform strategy
that satisfies $A$. Thus $s^* \not \in \Ne_{\mathsf{pess}}(G)$. If at
least one of the players $4, \ldots, 7$ reveal the proposition that
they control to player 8, then for player 1 to ensure the outcome 1
for all $v$, it need not necessarily have to satisfy $\alpha_2$ but can
choose to satisfy the corresponding formula $\alpha_3, \ldots,
\alpha_6$. But this would also imply deviating from $\cs_1$ as listed
in items 3-6 above. Thus we can conclude that for all uniform
strategy $\cs^*$ with $\cs^*_1 = \cs_1$ we have that $\cs^* \not\in
\Ne_{\mathsf{pess}}(G)$.

In Table \ref{tab:two} we enumerate all possibilities. In the first
column we list all the possible combinations of the formulas in
$A,B,C,D$ that player 1 can possibly satisfy in any uniform
strategy. Corresponding to each row which denotes a uniform strategy
$\cs_1$ of player 1, in the second column in Table \ref{tab:two}, we
list the minimal set of players $X$ which satisfy the following
conditions.  Given the uniform strategy $\cs_1$ of player 1,
\begin{itemize}
\item For all $j \in X$, player $j$ cannot ensure the outcome 1 (for
  all $v \in V$) by not revealing its proposition to player 8
  (assuming player 1 chooses $\cs_1$).
\item If for all $j \in X$, $\cs_j$ is a uniform strategy of player
  $j$ that reveals $p_j$ to player 8 then in the resulting uniform
  strategy profile $\cs$, we have $\uu_1(v, \cs)=1$ for all $v \in V$.
\end{itemize}
In other words, if all the players in $X$ reveal their proposition to
player 8 then the outcome for player 1 under the strategy $\cs_1$ is 1
for all $v$. For example, consider the uniform strategy $\cs_1'$ of
player 1 defined as follows: ``when $p_1$ is true reveal $p_1$ to 3
and 4 ($A$), when $p_1$ is false reveal $p_1$ to 4 but not to 3
($C$)''. Given $\cs_1'$, player 3 violates the first condition above
as player 3 can ensure the outcome 1 by not revealing $p_3$ to player
8. If both the players 5 and 7 reveal $p_5$ and $p_7$, respectively,
to player 8, then in the resulting uniform strategy profile $\cs'$ we
have $\uu_1(v,\cs')=1$ for all $v \in V$. However, note that then
$\cs' \not \in \Ne_{\mathsf{pess}}(G)$. In $\cs'$ since player 7
reveals $p_7$ to player 8 and $\cs_1'$ satisfies $C \vee A$, player 5
can deviate to not reveal $p_5$ and ensure an outcome of 1 (for all $v
\in V$). Thus player 5 has a profitable deviation from $\cs'$ and
therefore, $\cs' \not \in \Ne_{\mathsf{pess}}(G)$. A similar reasoning
applies to the other rows in Table \ref{tab:two}.  From the goal formulas
$\gamma_1, \ldots, \gamma_7$ we can verify that for every such set
$X$, there is a player $k \in X$ who can ensure an outcome of 1 by not
revealing $p_k$ to player 8. Therefore $\Ne_{\mathsf{pess}}(\obsgame)
= \emptyset$.
%\qed
\begin{table}
  \centering
  \begin{tabular}{ | l | c | r| }
    \hline
    $C \vee A$ & $\{5,7\}$\\ \hline
    $B \vee D$ & $\emptyset$ \\ \hline
    $D \vee C$ & $\emptyset$ \\ \hline
    $A \vee B$ & $\{3,6\}, \{6,7\}$\\ \hline
    $A \vee D$ & $\{2,3\}$\\ \hline
    $B \vee C$ & $\{4,5\}$\\
    \hline
  \end{tabular}
        \caption{Uniform strategies for player 1. Explanations are given in the text. \label{tab:two}}
        \end{table}
 
   \end{example}

\subsection{Existence of Nash Equilibrium in the General Case}
\label{subsec:existsNEGeneral}

Following up on Example 26 we now determine more generally for which
fragments of observation games the existence of a Nash equilibrium is
guaranteed.  An initial step would be to consider observation games
where the goal formulas for all players are restricted to the positive
fragment of $L^K$. For this fragment, the following result is
straightforward.
%  we have the following result.

\begin{proposition}
  \label{prop:existsNePositive}
Let $\obsgame=(N,(P_i)_{i \in N}, (\gamma_i)_{i \in N})$ be an
observation game where $\gamma_i \in L^+$ for all $i \in N$. Then
$\Ne(G) \neq \emptyset$ for all outcome relations.
\end{proposition}
\begin{proof}
Observe that when $\gamma_i \in L^+$ for all $i \in N$, the globally
  uniform strategy $\dot{s}^\forall_i$ (public announcement by player
  $i$ of $P_i$) is dominant for all $i \in N$. Thus $\dot{s}^\forall
  \in \Ne(G)$ for any outcome relation.
\end{proof}

In this section, we present a more general structural result that
identifies a class of observation games in which a Nash equilibrium is
guaranteed to exist. Our results show that the existence of
equilibrium crucially depends on the combination of positive/negative
epistemic assertions made by players in their goal
formulas. Observation games where the goal formulas are
  restricted to the positive fragment of $L^K$ can be
  viewed as a particular simple case in this setting.

We assume that the goal formulas are in negation normal form. For $i,j
\in N$ (where $j$ may be $i$) and $\gamma_i$ in $L^K_\nnf$, we first define
$x^j_{i}(\gamma_i)$ for $x \in \{+,-\}$. Intuitively,
$+_i^j(\gamma_i)$ and $-_i^j(\gamma_i)$ encode the fact that player
$i$ makes a positive and negative (respectively) epistemic assertion
about a variable assigned to player $j$ in the goal formula
$\gamma_i$. Formally, $x^j_{i}(\gamma_i)$ is defined as follows.

\begin{itemize}
\item For $\gamma_i = p_j$ and  $\gamma_i = \neg p_j$ we have $+^j_{i}(p_j)$ and $ +^j_{i}(\neg p_j)$.
\item $\gamma_i = K_k \phi$ (where $k \in N$): $+^j_{i}(K_k \phi) \text{ iff } +^j_{i}(\phi)$ and $-^j_{i}(K_k \phi) \text{ iff } -^j_{i}(\phi)$.
\item $\gamma_i = \M_k \phi$ (where $k \in N$): $+^j_{i}(\M_k \phi) \text{ iff} -^j_{i}(\phi)$ and $-^j_{i}(\M_k \phi) \text{ iff} +^j_{i}(\phi)$.
  \item $\gamma_i = \phi \et  \psi$:
    \begin{itemize}
    \item $+^j_{i}(\phi \et \psi) \text{ iff}
      +^j_{i}(\phi)\ \text{or}\ +^j_{i}(\psi)$ and $-^j_{i}(\phi \et
      \psi) \text{ iff} -^j_{i}(\phi)\ \text{or}\ -^j_{i}(\psi)$.
    \end{itemize}
  \item $\gamma_i = \phi \vel \psi$:
    \begin{itemize}
    \item $+^j_{i}(\phi \vel \psi) \text{ iff}
      +^j_{i}(\phi)\ \text{or}\ +^j_{i}(\psi)$ and $-^j_{i}(\phi \vel
      \psi) \text{ iff} -^j_{i}(\phi)\ \text{or}\ -^j_{i}(\psi)$.      
    \end{itemize}
\end{itemize}

Note that the definition of $+_i^j(\gamma_i)$ is intended to encode
the fact that player $i$ makes a positive epistemic assertion about a
variable assigned to player $j$ in the goal formula $\gamma_i$. So in
item 4, we have that $+_i^j(\phi \et \psi)$ holds iff the same holds
for at least one of the subformulas $\phi$ or $\psi$. A similar
comment applies to definition of $-_i^j(\gamma_i)$.

For every player $i$, we define $\mathit{type}(i) \subseteq
\{+,-,c+,c-\}$ as follows.
For $x \in \{+,-\}$,
\begin{itemize}
\item $x \in \mathit{type}(i)$ if there is a player
  %  $j \in N$    with
  $j \neq i$ such that $x^j_i(\gamma_i)$,
    \item $cx \in
      \mathit{type}(i)$ if $x^i_i(\gamma_i)$.
\end{itemize}

In other words, $+$ and $-$ are in $\mathit{type}(i)$ if there exists
some player $j$ with $j \neq i$ such that player $i$ makes a positive and
negative (respectively) epistemic assertion about a variable assigned
to player $j$ in $\gamma_i$. Likewise, $c+$ and $c-$ is in
$\mathit{type}(i)$ if player $i$ makes a positive and negative
(respectively) epistemic assertion about its own variable in
$\gamma_i$.

For example, for $i \in N$, consider the goal formula $\gamma_i$ given
in Example \ref{ex:noNE}, under its translation to negation normal
form (NNF). We have the following for player 1.
\begin{itemize}
\item $K_1 p_2$ occurs as a subformula in the NNF of $\gamma_1$ and therefore $+ \in \mathit{type}(1)$. $\M_1 p_2$  occurs as a subformula in the NNF of $\gamma_1$ and therefore $- \in \mathit{type}(1)$.
\item $K_2 p_1$ occurs as a subformula in the NNF of $\gamma_1$ and therefore $c+ \in \mathit{type}(1)$. $\M_2 p_1$  occurs as a subformula in the NNF of $\gamma_1$ and therefore $c- \in \mathit{type}(1)$.
\end{itemize}

For player 2, the reasoning is similar and therefore, we have that
$\mathit{type}(i)=\{+,-,c+,c-\}$ for all $i \in N$. In fact, Theorem
\ref{thm:existskwNE} given below shows that it is crucial that
$|\mathit{type}(i)|>3$ in this example.

Based on the notion of $\mathit{type}$, we define the following
subsets of $N$.  Let
\begin{itemize}
\item $X_+=\{i \in N \mid c+ \in \mathit{type}(i)\}$,
$X_-=\{i \in N \mid c- \in \mathit{type}(i)\}$,
\item $W_l=\{i \in N \mid \mathit{type}(i) = \{c+,c-\}\}$,
\item $W_+=\{i \in N \mid \mathit{type}(i) = \{+, c+, c-\}\}$,
  $W_-=\{i \in N \mid \mathit{type}(i) = \{-, c+, c-\}\}$.
\end{itemize}
%\noindent {\bf Ordering over strategy profiles.} 
For the proofs in this section, we also find it useful to define an
ordering $\succcurlyeq$ over the set of strategy profiles. Let $X
\subseteq N$ and $s_X, t_X \in S_X$. We say that $s_X \succcurlyeq
t_X$ if for all $i \in X$ and $j \in N$, $t_i(j) \subseteq s_i(j)$. We
can then show the following existence result.

%% Now consider Algorithm~\ref{alg:NEpess}.

\begin{algorithm}[ht]
%\footnotesize
    \caption{\label{alg:NEpess}}
\KwIn{$G=(N,(P_i)_{i \in N}, (\gamma_i)_{i \in N})$.}
\KwOut{A uniform strategy profile $\cs \in \Ne_{\mathsf{pess}}(G)$.}

Let $W_o: =\obar{X_+} \cup \obar{X_-} \cup W_l$; 

$\forall i \in \obar{X_+}$, $\forall v \in V$, set
$\cs_i(v):= s_i^\emptyset$; ~~~~~~~~~~ \tcc{a dominant uniform strategy}

$\forall i \in \obar{X_-}\setminus \obar{X_+}$, 
$\forall v \in V$, set $\cs_i(v):=s_i^\forall$; ~~ \tcc{a dominant uniform strategy}

$\forall i \in W_l$, $\forall v \in V$, if $\exists s \in S$ such
that $\forall w: w \sim_i v$, $u_i(w,s)=1$ then $\forall w: w\sim_i
v$, set $\cs_i(w):=s_i$ else $\cs_i(w):=s_i^\emptyset$;~~ \tcc{$u_i$ does not depend on others' choice}

$\forall i \in W_{+}$, $\forall j \in W_{-}$, $\forall v \in V$ set
$\cs_i(v):=s_i^\emptyset$; $\cs_j(v):=s_j^\forall$; ~~ \tcc{initialisation}

$\forall v \in V$, set $Y(v):=\emptyset$; $Z(v):=\emptyset$;
%Set $Y:=\emptyset$ and $Z:=\emptyset$

\Repeat{$\forall v \in V$, $Y(v)=Y'(v)$ and $Z(v)=Z'(v)$}{
  $\forall v \in V$, set $Y'(v):=Y(v)$; $Z'(v):=Z(v)$;

  \tcc{process players who make positive assertions about variables controlled by others}

\While{ $\exists v \in V$, $\exists i \in W_+\setminus Y(v)$, $\exists
  s_i$, such that $\forall w: w \sim_i v$, $\forall
  s_{W_{-}\setminus Z(w)}$, we have
  ${\scriptsize (s_i, \cs_{W_{+}\setminus\{i\}}(w), \cs_{Z(w)}(w), s_{W_{-}\setminus
    Z(w)}, \cs_{W_o}(w)), w \models \gamma_i}$}
{$\forall w: w \sim_i v$, set $\cs_i(w):=s_i$; 
$Y(w):=Y(w) \cup \{i\}$; }

\tcc{process players who make negative assertions about variables controlled by others}

\While{ $\exists v \in V$, $\exists i \in W_-\setminus Z(v)$, $\exists
  s_i$ such that $\forall w: w \sim_i v$, $\forall s_{W_{+}\setminus
    Y(w)}$, we have
  ${\scriptsize (s_i,\cs_{W_{-}\setminus\{i\}}(w),\cs_{Y(w)}(w),s_{W_{+}\setminus
      Y(w)},\cs_{W_o}(w)),w \models \gamma_i}$}
      {$\forall w: w \sim_i v$, set $\cs_i(w):=s_i$;
        $Z(w):=Z(w) \cup \{i\}$; } }

$\forall i \in W_+ \setminus Y(v)$, $\forall j \in W_{-}\setminus
Z(v)$ and $\forall v \in V$, set $\cs_j(v)(i):= \emptyset$;

$\forall i \in W_- \setminus Z(v)$, $\forall j \in W_{+}\setminus
Y(v)$ and $\forall v \in V$, set $\cs_j(v)(i):= P_j$;

{\bf return} $\cs$;
\end{algorithm}

%% Using the ``\textit{type}'' classification of goal formulas we have
%% the following existence results.

\begin{restatable}{theorem}{thmexistsNEpess}
  \label{thm:existsNEpess}
  Let $G=(N,(P_i)_{i \in N}, (\gamma_i)_{i \in N})$ be an observation
  game where all goals $\gamma_i$ are guarded.
  %%  and let $H[(\gamma_i)_{i \in N}]$ be the associated graph.
  If for all $i \in
  N$, $|\mathit{type}(i)| \leq 3$, then
  \begin{enumerate}
  \item $\Ne_{\mathsf{pess}}(G) \neq  \emptyset$.
  \item $\Ne_{\mathsf{opt}}(G) \neq \emptyset$.
  \end{enumerate}
\end{restatable}

\begin{proof}
  \noindent \textbf {Part 1.} $\Ne_{\mathsf{pess}}(G) \neq
  \emptyset$. Consider the procedure described as
  Algorithm~\ref{alg:NEpess}. We argue that
  Algorithm~\ref{alg:NEpess} always terminates and constructs a
  uniform strategy profile $\cs \in \Ne_{\mathsf{pess}}(G)$.
  First, we note that the sets $\obar{X_+}$, $\obar{X_-} \setminus
  \obar{X_+}$, $W_l$, $W_+$ and $W_-$ form a partition of $N$.

In each iteration of the outer loop in Algorithm~\ref{alg:NEpess}
(steps 7 - 13), the size of the set $Y(v)$ or $Z(v)$ strictly
increases for some $v \in V$. We also have that for all $v \in V$, $0
\leq |Y(v)| \leq |N|$ and $0 \leq |Z(v)| \leq |N|$. It follows that
Algorithm~\ref{alg:NEpess} always terminates. Let $\cs$ be the
strategy profile constructed by Algorithm~\ref{alg:NEpess}. From the
description of the procedure, it can also be verified that $\cs$ is a
uniform strategy profile. Thus to prove the claim, it suffices to show
that $\cs \in \Ne_{\mathsf{pess}}(G)$.

Note that for all $i \in \obar{X_+}$, $\cs_i^\emptyset$ is a dominant
uniform strategy and for all $i \in \obar{X_-}$, $\cs_i^\forall$ is a
dominant uniform strategy. Therefore, for all $v \in V$, for all $i \in
\obar{X_+} \cup \obar{X_-}$ and for all $s_i' \in S_i$, $u_i(v,\cs(v))
\geq u_i(v,(s_i',\cs_{-i}(v)))$.

For all $i \in W_l$, for all $v \in V$, we have $\uu_i(v,
(\cstwo_i,\cstwo_{-i}))=\uu_i(v,(\cstwo_i,\cstwo'_{-i}))$ for all $v
\in V$, $\cstwo_i \in \CS_i$ and for all $\cstwo_{-i}, \cstwo'_{-i}
\in \CS_{-i}$. Therefore, by the choice of $\cs_i$ made in line 4 of
Algorithm~\ref{alg:NEpess}, we have for all $i \in W_l$, for all $v
\in V$, for all $\cs_i' \in \CS_i$, $\uu_i(v,\cs) \geq
\uu_i(v,(\cs_i',\cs_{-i}))$.

Now consider a player $i \in W_{+}$. For $v \in V$, suppose $\cs_i(v)$
is assigned a value in the while loop (steps 9-10). Let $\cs^k$ denote
the resulting strategy profile after this assignment (step 10) and
$Z^k$ denote the value of $Z$ in the corresponding iteration. By
definition of the while loop, for all $w$ with $v \sim_i w$, for
all $s_{W_{-} \setminus Z^k(w)}$,
$(s_i,\cs^k_{W_{+}\setminus\{i\}}(w),\cs_{Z^k}(w),s_{W_{-}\setminus
  Z^k(w)},\cs_{W_o}(w)),w \models \gamma_i$. Since $i \in W_{+}$, this
implies that for all $t_{W_{+}\setminus\{i\}} \in
S_{W_{+}\setminus\{i\}}$ such that $t_{W_{+}\setminus\{i\}}
\succcurlyeq$ %_{{W_{+}\setminus\{i\}}}
$\cs^k_{W_{+}\setminus\{i\}}(w)$, for all $s_{W_{-} \setminus
  Z^k(w)}$,
$(s_i,t_{W_{+}\setminus\{i\}},\cs_{Z^k(w)}(w),s_{W_{-}\setminus
  Z^k(w)},\cs_{W_o}(w)),w \models \gamma_i$.
By definition, we have $\cs_{W_+\setminus \{i\}}(w) \succcurlyeq
\cs^k_{W_+\setminus \{i\}}(w)$ and for all $j \in Z^k(w)$,
$\cs_j^k(w)=\cs_j(w)$.  Therefore, it follows that 
$(s_i,\cs_{W_{+}\setminus\{i\}},\cs_{Z}(w),s_{W_{-}\setminus
  Z(w)},\cs_{W_o}(w)), w \models \gamma_i$ and $\uu_i(v,\cs)=1$.

Consider a player $i \in W_{-}$. For $v \in V$, suppose $\cs_i(v)$ is
assigned a value in the while loop (steps 11-12). Let $\cs^k$ denote
the resulting strategy profile after this assignment (step 12) and
$Y^k$ denote the value of $Y$ in the corresponding iteration. By
definition of the while loop, for all $w$ with $v \sim_i w$, for
all $s_{W_{+} \setminus Y^k(w)}$,
$(s_i,\cs^k_{W_{-}\setminus\{i\}}(w),\cs_{Y^k(w)}(w),s_{W_{+}\setminus
  Y^k(w)},\cs_{W_o}(w)),w \models \gamma_i$. Since $i \in W_{-}$, this
implies that for all $t_{W_{-}\setminus\{i\}} \in
S_{W_{-}\setminus\{i\}}$ such that $\cs^k_{W_{-}\setminus\{i\}}(w)
\succcurlyeq t_{W_{-}\setminus\{i\}}$, for all $s_{W_{+} \setminus
  Y^k(w)}$,
$(s_i,t_{W_{-}\setminus\{i\}},\cs_{Y^k(w)}(w),s_{W_{+}\setminus
  Y^k(w)},\cs_{W_o}(w)),w \models \gamma_i$.
By definition, $\cs^k_{W_-\setminus \{i\}}(w) \succcurlyeq
\cs_{W_-\setminus \{i\}}(w)$.  Thus
$(s_i,\cs_{W_{-}\setminus\{i\}},\cs_{Y}(w),s_{W_{+}\setminus
  Y(w)},\cs_{W_o}(w)),w \models \gamma_i$. Therefore, $\uu_i(v,\cs)=1$.
%% and
%% player $i$ does not have a profitable deviation.

Now suppose there exists $v \in V$ and $i \in W_{+}$ such that $i
\notin Y(v)$ (on termination of the repeat loop, steps 7-13). By
definition, for all $s_i$, there exists $w$ with $v \sim_i w$ and
there exists $t_{W_{-}\setminus Z(w)}$ such that
$(s_i,\cs_{W_{+}\setminus\{i\}}(w),\cs_{Z(w)}(w),t_{W_{-}\setminus
  Z(w)},\cs_{W_o}(w)),w \not\models \gamma_i$. Since $i \in W_+$ and
$s_j(v)(i) = \emptyset$ for all $j \in W_-\setminus Z(v)$, it follows
that for all $s_i$, there exists a $w$ with $v \sim_i w$ such that
$(s_i,\cs_{W_{+}\setminus\{i\}}(w),\cs_{Z}(w),s_{W_{-}\setminus
  Z(w)},\cs_{W_o}(w)),w \not\models \gamma_i$. Therefore, for all $\cs_i'
\in \CS_i$, $\uu_i(v,\cs) \geq \uu_i(v,(\cs_i',\cs_{-i}))$.

Suppose there exists $v \in V$ and $i \in W_{+}$ such that $i \notin
Z(v)$. Using a similar proof as above and using the fact that
$s_j(v)(i) = \emptyset$ for all $j \in W_-\setminus Z(v)$ we can argue
that for all $s_i$, there exists a $w$ with $v \sim_i w$ such that
$(s_i,\cs_{W_{-}\setminus\{i\}}(w),\cs_{Y(w)}(w),s_{W_{+}\setminus
  Y(w)},\cs_{W_o}(w)),w \not\models \gamma_i$. Therefore, for all $\cs_i'
\in \CS_i$, $\uu_i(v,\cs) \geq \uu_i(v,(\cs_i',\cs_{-i}))$.

\noindent {\bf Part 2.} To show that $\Ne_{\mathsf{opt}}(G) \neq
\emptyset$, we modify Algorithm~\ref{alg:NEpess} to reflect the
optimist decision rule.
This is achieved by changing the conditional in both the While loops
(line 9 and line 11) as described below. Note that the only change is
a switch to existential quantification over the valuations $w$ in
order to capture the definition of the optimist decision rule.

\noindent{\textit Line 9.}\\
{
\noindent{\bf While}
$\exists v \in V$, $\exists i \in W_+\setminus Y(v)$, $\exists
  s_i$, such that $\exists w: w \sim_i v$, $\forall
  s_{W_{-}\setminus Z(w)}$, we have
  ${\scriptsize (s_i, \cs_{W_{+}\setminus\{i\}}(w), \cs_{Z(w)}(w), s_{W_{-}\setminus
      Z(w)}, \cs_{W_o}(w)), w \models \gamma_i}$ {\bf do}.
}

\noindent{\textit Line 11.}\\
{
  \noindent{\bf While}
  $\exists v \in V$, $\exists i \in W_-\setminus Z(v)$, $\exists
  s_i$ such that $\exists w: w \sim_i v$, $\forall s_{W_{+}\setminus
    Y(w)}$, we have
  ${\scriptsize (s_i,\cs_{W_{-}\setminus\{i\}}(w),\cs_{Y(w)}(w),s_{W_{+}\setminus
      Y(w)},\cs_{W_o}(w)),w \models \gamma_i}$ {\bf do}.
}    
\end{proof}

The result in Theorem \ref{thm:existsNEpess} is tight in the sense
that there exist observation games where $|\mathit{type}(i)| =4$ for
$i \in N$ and $\Ne_{\mathsf{pess}}(G) = \emptyset$. This is
illustrated in Example \ref{ex:KWnoNEpess} where for players $i \in
\{2, \ldots, 7\}$, $\mathit{type}(i) =\{+, -, c+, c-\}$.
%Example \ref{ex:noNEpess}.

An interesting corollary of Theorem~\ref{thm:existsNEpess} is
for \emph{self-positive} goals: my objective is never to remain
  ignorant of others' variables even when it may be for others to
  remain ignorant. (We recall their definition in Section~\ref{sec.logic}.)

\begin{corollary} \label{cor.self}
Let $G=(N,(P_i)_{i \in N}, (\gamma_i)_{i \in N})$
be an observation game where all $(\gamma_i)_{i \in N}$
are guarded and self-positive, then $\Ne_{\mathsf{pess}}(G) \neq
\emptyset$ and $\Ne_{\mathsf{opt}}(G) \neq  \emptyset$.
%$G$ has a Nash equilibrium.
\end{corollary}
\begin{proof}
  Follows from Theorem~\ref{thm:existsNEpess}, since $\forall i$, $-
  \not\in \mathit{type}(i)$.
\end{proof}

For $\Ne_{\mathsf{max}}(G)$ we show a weaker result (Theorem
\ref{thm:existsNEmax}) which can be strengthened for $\Kw$ games
(Theorem \ref{thm:existskwNE}).

\begin{restatable}{theorem}{thmexistsNEmax}
%\begin{theorem}
\label{thm:existsNEmax}
Let $G=(N,(P_i)_{i \in N}, (\gamma_i)_{i \in N})$ be an observation
game where the goal formulas $(\gamma_i)_{i \in N}$ are guarded.
%% and $H[(\gamma_i)_{i \in N}]$ be the associated graph.
If for all $i \in N$, $|\mathit{type}(i)| \leq 2$ then
$\Ne_{\mathsf{max}}(G) \neq \emptyset$.
\end{restatable}

\begin{proof}
Let $G=(N,(P_i)_{i \in N}, (\gamma_i)_{i \in N})$ be an observation
game. Let $X_+=\{i \in N \mid c+ \in \mathit{type}(i)\}$ and $X_-=\{i
\in N \mid c- \in \mathit{type}(i)\}$. Consider the uniform strategy
profile $\cs$ defined as follows.
\begin{itemize}
  \item For $i \in X_+ \cap X_-$, we define $\cs_i$ using the
    iterative procedure: for $v \in V$ where $\cs_i(v)$ is
    not defined, if there exists $s \in S$ such that $u_i(v, \cs)=1$
    then set $\cs_i(w)=s_i$ for all $w: v \sim_i w$. Otherwise set
    $\cs_i(w)= s_i^\emptyset$ for all $w: v \sim_i w$.

  \item For all $i \in \obar{X_+}$, for all $v \in V$, let
    $\cs_i(v)=s_i^\emptyset$.
  \item For all $i \in N \setminus [(X_+ \cap X_-) \cup \obar{X_+}]$,
    for all $v \in V$, let $\cs_i(v)=s_i^\forall$.
\end{itemize}

Note that for all $i \in \obar{X_+}$, $\cs_i^\emptyset$ is a dominant
uniform strategy. For all $i \in \obar{X_-}$, $\cs_i^\forall$ is a
dominant uniform strategy and for all $i \in (\obar{X_+} \cup
\obar{X_-})$, both $\cs_i^\emptyset$ and $\cs_i^\forall$ are dominant
uniform strategies.

Since the goal formulas are guarded, we have for all $i \in N$ and
for all $v \in V$, $v, \cs(v) \models \gamma_i$ iff $w, \cs(v)
\models \gamma_i$ for all $w: v \sim_i w$. Also, for all $i \in X_+
\cap X_-$, with $\mathit{type}(i) \leq 2$, we have that $u_i(v,\cs(v))=u_i(v, (\cs_i(v),s'_{-i}))$ for all $s'_{-i} \in
S_{-i}$. It then follows that $\cs \in \Ne_{\mathsf{max}}(G)$.
\end{proof}

\subsection{Existence of Maximal Nash Equilibrium in Knowing-Whether Games}

For the subclass of $\Kw$ games, we show that Theorem \ref{thm:existsNEmax} can be
strengthened. 
We argue that if $G$ is a $\Kw$ game where for all $i \in N$, $|\mathit{type}(i)| \leq 3$ then
the output of Algorithm~\ref{alg:NEmaxkw} is a
globally uniform strategy profile $\cs$ such that $\cs \in
\Ne_{\mathsf{max}}(G)$.

% Consider Algorithm~\ref{alg:NEmaxkw} given below.

  \begin{algorithm}[ht]
%    \algsetup{linenosize=\tiny}
%  \footnotesize
\caption{\label{alg:NEmaxkw}}
\KwIn{A $\Kw$ game $G=(N,(P_i)_{i \in N}, (\gamma_i)_{i \in N})$.}
\KwOut{A uniform strategy profile $\cs \in \Ne_{\mathsf{max}}(G)$.}

%$M := (\sim_i)_{i \in N}$;
Let $W_o: =\obar{X_+} \cup \obar{X_-} \cup W_l$;

$\forall i \in \obar{X_+}$, $\forall v \in V$, set
$\cs_i(v):= s_i^\emptyset$; ~~~~~~~~~~ \tcc{a dominant uniform strategy}

$\forall i \in \obar{X_-}\setminus \obar{X_+}$, 
$\forall v \in V$, set $\cs_i(v):=s_i^\forall$; ~~ \tcc{a dominant uniform strategy}

$\forall i \in W_l$, if $\exists s \in S$, and $\exists v \in V$ such
that $u_i(v,s)=1$ then $\forall w \in V$ set $\cs_i(w):=s_i$ else set
$\cs_i(w):=s_i^\emptyset$; ~~ \tcc{$u_i$ does not depend on others' choice}

$\forall i \in W_{+}$, $\forall j \in W_{-}$, $\forall v \in V$ set
$\cs_i(v):=s_i^\emptyset$; $\cs_j(v):=s_j^\forall$; ~~ \tcc{initialisation}

%For all $v \in V$, set $Y(v):=\emptyset$ and $Z(v):=\emptyset$
Set $Y:=\emptyset$; $Z:=\emptyset$;

\Repeat{$Y=Y'$ and $Z=Z'$}{
set $Y':=Y$; $Z':=Z$;

\tcc{process players who make positive assertions about
    variables controlled by others}

\While{ $\exists i \in W_+\setminus Y(v)$, $\exists s_i$ and $\exists
  v \in V$ such that $\forall s_{W_{-}\setminus Z}$, we have
  $v (s_i,\cs_{W_{+}\setminus\{i\}}(v),\cs_{Z}(v),s_{W_{-}\setminus
    Z},\cs_{W_o}(v))\models \gamma_i$}
      {$\forall w \in V$, set
        $\cs_i(w):=s_i$;
        $ Y :=Y \cup \{i\}$; }
      
\tcc{process players who make negative assertions about
    variables controlled by others}
      
\While{ $\exists i \in W_-\setminus Z$, $\exists s_i$, $\exists v \in
  V$, such that $\forall s_{W_{+}\setminus Y}$, we have $v
  (s_i,\cs_{W_{-}\setminus\{i\}}(v),\cs_{Y}(v),s_{W_{+}\setminus
    Y},\cs_{W_o}(v)) \models \gamma_i$}
      {$\forall w \in V$, set $\cs_i(w):=s_i$;
        $Z:=Z \cup \{i\}$; } }

$\forall i \in W_+ \setminus Y$, $\forall j \in W_{-}\setminus Z$,
$\forall v \in V$, set $\cs_j(v)(i):= \emptyset$;

$\forall i \in W_- \setminus Z$, $\forall j \in
W_{+}\setminus Y$, $\forall v \in V$, set $\cs_j(v)(i):= P_j$;

{\bf return} $\cs$;
  \end{algorithm}

  \begin{lemma}
  \label{lm:NEforall}
  Algorithm~\ref{alg:NEmaxkw} always terminates and it satisfies the following properties.
  \begin{itemize}
  \item After each iteration of the while loops, steps 9-10 and steps
    11-12, the strategy profile $\cs$ constructed is a globally
    uniform strategy profile.
    \item The strategy profile $\cs$ which is the output of
      Algorithm~\ref{alg:NEmaxkw} is a globally uniform strategy
      profile.
  \end{itemize}    
\end{lemma}
\begin{proof}
  First, note that in each iteration of the outer loop in
  Algorithm~\ref{alg:NEmaxkw} (steps 7 - 13), the size of the set $Y$
  or $Z$ strictly increases.  Therefore Algorithm~\ref{alg:NEmaxkw}
  always terminates.

  At the end of the initialization steps (2 - 6), $\cs \in
  \CS^g$ by definition. So it suffices to argue that at the end
  of each iteration of the two While loops (steps 9 - 10 and 11 - 12),
  the following invariant is maintained: $\cs \in \CS^g$. We can
  argue by induction on the number of iterations of the while loops
  (steps 7 - 13). The claim follows from the definition of the
  assignment statements: steps 10 and 12.

  Thus on termination of the outer loop (steps 7 - 13) we have that
  $\cs \in \CS^g$.  It follows from the definition of lines 14
  and 15 that the output of Algorithm~\ref{alg:NEmaxkw}, $\cs \in
  \CS^g$.
\end{proof}

%% Theorem \ref{thm:existsNEmax} can be strengthened for $\Kw$ games:

\begin{restatable}{theorem}{thmexistskwne}
%\begin{theorem}
\label{thm:existskwNE}
Let $\obsgame=(N,(P_i)_{i \in N}, (\gamma_i)_{i \in N})$ be a $\Kw$
 game.
If for all $i \in N$, $|\mathit{type}(i)| \leq 3$ then
$\Ne_{\mathsf{max}}(G) \neq \emptyset$.
\end{restatable}

\begin{proof}
We argue that the output of Algorithm~\ref{alg:NEmaxkw} is a
globally uniform strategy profile $\cs$ such that $\cs \in
\Ne_{\mathsf{max}}(G)$. As in the case of
Theorem~\ref{thm:existsNEpess}, note that the sets $\obar{X_+}$,
$\obar{X_-} \setminus \obar{X_+}$, $W_l$, $W_+$ and $W_-$ form a
partition of $N$.

By Lemma~\ref{lm:NEforall}, Algorithm~\ref{alg:NEmaxkw} always
terminates. Let $\cs$ be the profile constructed by
Algorithm~\ref{alg:NEmaxkw}. To prove the claim, it suffices to argue
that $\cs \in \Ne_{\mathsf{max}}(G)$.

Note that for all $i \in \obar{X_+}$, $\cs_i^\emptyset$ is a dominant
uniform strategy and for all $i \in \obar{X_-}$, $\cs_i^\forall$ is a
uniform strategy that is dominant. Therefore, for all $v \in V$, for
all $i \in \obar{X_+} \cup \obar{X_-}$, $u_i(v,\cs(v)) \geq
u_i(v,(s_i',\cs_{-i}(v)))$.

For all $i \in W_l$ we have that $u_i(v,\cs(v))=u_i(v,
(\cs_i(v),s'_{-i}))$ for all $s'_{-i} \in S_{-i}$. Since the goals are
knowing whether formulas, we have if there exists $v \in V$ and there
exists $s \in S$ such that $u_i(v,s)=1$ then for all $w \in V$, for
all $s_{-i}' \in S_{-i}$, $u_i(w, (s_i, s_{-i}'))=1$. Therefore, for
all $v \in V$, for all $i \in W_l$, $u_i(v,\cs(v)) \geq
u_i(v,(s_i',\cs_{-i}(v)))$.

Now consider a player $i \in W_{+}$. For $v \in V$, suppose $\cs_i(v)$
is assigned a value in the while loop (steps 9-10). Let $\cs^k$ denote
the resulting strategy profile after this assignment (step 10). Let
$Z^k$ denote the value of $Z$ in the corresponding iteration. By
Lemma~\ref{lm:NEforall}, we have $\cs^k \in \CS^g$ and by definition
of the while loop, there exists $v$ such that for all $s_{W_{-}
  \setminus Z^k}$,
$v,(s_i,\cs^k_{W_{+}\setminus\{i\}}(v),\cs_{Z^k},s_{W_{-}\setminus
  Z^k},\cs_{W_o}(v))\models \gamma_i$. By
%% Lemma~\ref{lm:kwtruth} and
Lemma~\ref{lm:NEforall} and the fact that $\gamma_i \in L^\Kw$
it follows that for all $w$, for all $s_{W_{-} \setminus Z^k}$,
$v,(s_i,\cs^k_{W_{+}\setminus\{i\}}(v),\cs_{Z^k},s_{W_{-}\setminus
  Z^k},\cs_{W_o}(v))\models \gamma_i$. Since $i \in W_{+}$, this
implies that for all $w \in V$, for all $t_{W_{+}\setminus\{i\}} \in
S_{W_{+}\setminus\{i\}}$ such that $t_{W_{+}\setminus\{i\}}
\succcurlyeq$
%_{{W_{+}\setminus\{i\}}} $\cs^k_{W_{+}\setminus\{i\}}(w)$, for all
$s_{W_{-} \setminus Z^k}$,
$w,(s_i,t_{W_{+}\setminus\{i\}},\cs_{Z^k}(w),s_{W_{-}\setminus
  Z^k(w)},\cs_{W_o}(w)) \models \gamma_i$.
By definition, $\cs_{W_+\setminus \{i\}}(w) \succcurlyeq
\cs^k_{W_+\setminus \{i\}}(w)$ and for all $j \in Z^k$,
$\cs_j^k(w)=\cs_j(w)$.  
Thus we have\\
$w, (s_i, \cs_{W_{+}\setminus\{i\}}, \cs_{Z}(w), s_{W_{-}\setminus
  Z}, \cs_{W_o}(w)) \models \gamma_i$. Therefore, $u_i(w,\cs)=1$ for all
$w \in V$.
% and player $i$ does not have a profitable deviation.

Consider a player $i \in W_{-}$. For $v \in V$, suppose $\cs_i(v)$ is
assigned a value in the while loop (steps 11-12). Let $\cs^k$ denote
the resulting uniform strategy profile after this assignment (step 12)
and $Y^k$ denote the value of $Y$ in the corresponding iteration. By
Lemma~\ref{lm:NEforall}, $\cs^k \in \CS^g$. By definition of the while
loop, there exists $v$ such that for all $s_{W_{+} \setminus Y^k}$,
$v,(s_i,\cs^k_{W_{-}\setminus\{i\}}(v),\cs_{Y^k}(v),s_{W_{+}\setminus
  Y^k},\cs_{W_o}(w)) \models \gamma_i$. By
% Lemma~\ref{lm:kwtruth} and
Lemma~\ref{lm:NEforall} and the fact that $\gamma_i \in L^\Kw$,
%it follows
we have for all $w$ such that for all
$s_{W_{+} \setminus Y^k}$,
$w,(s_i,\cs^k_{W_{-}\setminus\{i\}}(w),\cs_{Y^k}(w),s_{W_{+}\setminus
  Y^k},\cs_{W_o}(w)) \models \gamma_i$.  Since $i \in W_{-}$, this
implies that for all $t_{W_{-}\setminus\{i\}} \in
S_{W_{-}\setminus\{i\}}$ such that $\cs^k_{W_{-}\setminus\{i\}}(w)
\succcurlyeq t_{W_{-}\setminus\{i\}}$, for all $s_{W_{+} \setminus
  Y^k(w)}$,
$w,(s_i,t_{W_{-}\setminus\{i\}},\cs_{Y^k(w)}(w),s_{W_{+}\setminus
  Y^k(w)},\cs_{W_o}(w))\models \gamma_i$.
By definition, $\cs^k_{W_-\setminus \{i\}}(w) \succcurlyeq
\cs_{W_-\setminus \{i\}}(w)$.  Thus
$w,(s_i,\cs_{W_{-}\setminus\{i\}},\cs_{Y}(w),s_{W_{+}\setminus
  Y(w)},\cs_{W_o}(w)) \models \gamma_i$. Therefore, $u_i(w,\cs)=1$ for all
$w \in W$.
%and player $i$ does not have a profitable deviation.

Now suppose there exists $i \in W_{+}$ such that $i \notin Y$ (on
termination of the repeat loop, steps 7-13). By definition, for all
$s_i$, for all $v \in V$, there exists $t_{W_{-}\setminus Z}$ such
that
$v,(s_i,\cs_{W_{+}\setminus\{i\}}(v),\cs_{Z}(v),t_{W_{-}\setminus
  Z},\cs_{W_o}(v)) \not\models \gamma_i$. Since $i \in W_+$ and
$s_j(v)(i) = \emptyset$ for all $j \in W_-\setminus Z$, it follows
that for all $s_i$, for all $v \in V$ 
$v,(s_i,\cs_{W_{+}\setminus\{i\}}(v),\cs_{Z}(v),s_{W_{-}\setminus
  Z},\cs_{W_o}(v)) \not\models \gamma_i$.
% Thus player $i$ does not have a profitable deviation.

Suppose there exists $v \in V$ and $i \in W_{+}$ such that $i \notin
Z$. Using a similar proof as above and using the fact that
$s_j(v)(i) = \emptyset$ for all $j \in W_-\setminus Z(v)$ we can argue
that for all $s_i$, for all $v \in V$, 
$v,(s_i,\cs_{W_{-}\setminus\{i\}}(w),\cs_{Y}(w),s_{W_{+}\setminus
  Y},\cs_{W_o}(w))\not\models \gamma_i$.
% Thus player $i$ does not have a profitable deviation.
It follows that $\cs \in \Ne_{\mathsf{max}}(G)$.
\end{proof}

\noindent Examples~\ref{prop:obs3nostable} and \ref{prop:kw4nostable}
show that Theorems \ref{thm:existsNEmax} and \ref{thm:existskwNE} are
tight.

  \begin{example}
    \label{prop:obs3nostable}
Consider the two-player game where $N=\{1,2\}$, $P_1=\{p\}$ and
$P_2=\{q^1,q^2,q^3\}$. Let $\gamma_1 = (\Kw_1 q^2 \wedge \Kw_2 p) \vee (\Kw_1 q^3 \wedge
  \neg \Kw_2 p)$ and  $\gamma_2 = (q^1 \to \Kw_1 q^2) \wedge (\neg q^1 \to \Kw_1 q^3) \wedge (\neg \Kw_1 q^2 \vee \neg \Kw_1 q^3)$.
Note that in this game, $|\mathit{type}(1)| =3$ and
$|\mathit{type}(2)| = 2$. The goal of player 1 is a $\Kw$
formula.
It can be verified that $\Ne_{\mathsf{max}}(G) = \emptyset$.
  \end{example}

  \begin{example}
    \label{prop:kw4nostable}
Consider the two-player game where $P_1=\{p_1,q_1\}$ and
$P_2=\{p_2\}$. Let the goal formulas be: $\gamma_1 = (\neg \Kw_1
p_2 \imp (\Kw_2 p_1 \wedge \neg \Kw_2 q_1)) \wedge (\Kw_1 p_2 \imp
(\neg \Kw_2 p_1 \wedge \Kw_2 q_1))$ and $\gamma_2 = (\Kw_2
p_1 \wedge \Kw_1 p_2) \vee (\Kw_2 q_1 \wedge \neg \Kw_1 p_2)$.
In this game, both goals are $\Kw$ formulas. We have
$|\mathit{type}(1)| =4$ and $|\mathit{type}(2)| = 3$. It can be
verified that $\Ne_{\mathsf{max}}(G) = \emptyset$.
  \end{example}

\section{Representation and Complexity}
\label{sec:complexity}

%\subsection{Introduction}

For an observation game $\obsgame = (N,(P_i)_{i \in N}, (\gamma_i)_{i
  \in N})$, let $|N|=n$, $|P|=k$ and $\max_{i \in N}|\gamma_i| = m$
(where $|\gamma_i|$ denotes the number of symbols in $\gamma_i$). For
$i \in N$, every strategy $s_i: N \to \powerset(P_i)$, can be
represented in size $\bigO(nk)$. That is, both observation games and
strategies have compact representations --- linear in $n$, $k$ and $m$.

On the other hand, each uniform strategy $\cs_i: V \to S_i$ can be
encoded as a tuple of Boolean functions $(\cs_i^j(p_i))_{j \in N, p_i
  \in P_i}$ where each $\cs_i^j(p_i): \powerset(P) \to \{\top,
\bot\}$. Here $\cs_i^j(p_i)(v) = \top$ is viewed as player $i$
revealing the variable $p_i$ to player $j$ under the valuation
$v$. We assume that the Boolean function $\cs_i^j(p_i)$ is represented
as a propositional formula $\beta_i^j(p_i)$ over the propositions
$P$. It is well known that every such Boolean function can be
represented as a propositional formula, in the worst case the size of
$\cs_i^j(p_i)$ can be exponential in $k$.

In this section we address the computational complexity of the
following two basic algorithmic questions.

\begin{itemize}
  \item {\bf Verification.} Given an observation game $\obsgame$ and a
    uniform strategy profile $\cs \in \CS$, is $\cs \in
   \Ne_{\mathsf{x}}(G)$ for an outcome relation $\mathsf{x} \in
    \{\mathsf{pess},\mathsf{opt}, \mathsf{max}\}$?
    
\item {\bf Emptiness.}  Given an observation game $\obsgame$ is
  $\Ne_{\mathsf{x}}(G) = \emptyset$ for an outcome relation $\mathsf{x} \in
    \{\mathsf{pess},\mathsf{opt}, \mathsf{max}\}$?
\end{itemize}

We show that the verification and emptiness questions are
PSPACE-complete and NEXPTIME-complete respectively
for the
  maximal outcome relation. We also show that for the pessimist and
  optimist outcome relations, the verification and emptiness questions
  are in PSPACE and NEXPTIME respectively.
To obtain these results it
is crucial to establish the complexity of the model checking
problem of the logic $L^K$. The following result shows that the model
checking problem is PSPACE-complete. It follows directly from
Proposition 2 in \cite{AgotnesHHW13}\footnote{We thank Paul
Harrenstein for providing us an unpublished full version of
\cite{AgotnesHHW13} which includes a proof of Proposition 2. For the
sake of completeness, we give a full proof of Theorem \ref{thm:mc} in
the Appendix.}.

\begin{restatable}{theorem}{thmmc}
\label{thm:mc}
  Given $\alpha \in L^K$ along with a strategy profile $s \in S$ and a
  valuation $v \in V$, checking if $v, s \models \alpha$ is PSPACE-complete.
  %\end{theorem}
\end{restatable}

It is well known that the model checking problem for epistemic logic
formulas over Kripke structures (for example, formulas of multi-agent
S5) can be solved in polynomial time \cite{FHMV95,HV91}.  Note that in
our setting, a Kripke structure is not explicitly part of the input,
rather the underlying relational structure is compactly presented in
terms of the valuation $v$ and strategy $s$. This is the reason for
PSPACE-hardness of the model checking problem.

%\noindent{\bf Verification.} 

\subsection{Verification}

In the rest of this section, we refer to
valuations over various sets of variables and therefore find it
convenient to use the following notations. Let $A$ be a finite set of
variables. We use $V(A)$ to denote the set of all valuations over
$A$.

\begin{theorem}
  \label{thm:compVerifNeMax}
  Given an observation game $\obsgame = (N,(P_i)_{i \in N},
  (\gamma_i)_{i \in N})$ and a uniform strategy profile $\cs \in \CS$,
  checking if $\cs \in \Ne_{\mathsf{max}}(G)$ is PSPACE-complete.
\end{theorem}
\begin{proof}
We can argue that the complement of the problem is in PSPACE. That is,
given $\obsgame$ and $\cs \in \CS$, the problem is to verify if $\cs
\not\in \Ne_{\mathsf{max}}(G)$.  This can be solved with the following
two steps:
\begin{enumerate}
\item Guess $i \in N$, $v \in V$ and $s_i \in S_i$.
\item Verify that $u_i(v,(s_i,\cs(v)_{-i})) > u_i(v,\cs(v))$.
\end{enumerate}
For step 1 note that the size of a strategy $|s_i| = \bigO(nk)$. So
the triple $(i, v, s)$ which forms a possible witness to the fact that
$\cs \not\in \Ne_{\mathsf{max}}(G)$ has a polynomial representation. By Theorem
\ref{thm:mc} it follows that step 2 can be solved in PSPACE. Since
PSPACE is closed under complementation and NPSPACE = PSPACE due to
Savich's Theorem, the membership in PSPACE follows.

To show PSPACE-hardness, we give a reduction from the model checking
problem for $L^K$. That is, given $\alpha \in L^K$, a strategy profile
$s \in S$ and a valuation $v \in V$ we construct an observation game
$\obsgame$ and a uniform strategy $\cs$ as follows. Let
$P(\alpha)$ denote the set of variables occurring in
$\alpha$ and $p_1 \in P(\alpha)$ be an arbitrary fixed
variable. Let $q$ be a variable such that $q \not\in
P(\alpha)$.

The set of players $N=\{1,2\}$. $P_1=P(\alpha)$ and $P_2=\{q\}$. To
define the goal formulas we make use of the following notations. Let
$\delta_v$ denote the Boolean formula over $P_1$ which uniquely
characterises the valuation $v$. That is, $\delta_v := \Et_{p \in v} p
\et \Et_{p \notin v} \neg p$. For the (fixed) variable $p_1 \in
P_1$, we define the formula $\mathit{flip}(p_1)$ as follows.

$\mathit{flip}(p_1) = \begin{cases}
  \Kw_2 q &\text{ if } p_1 \not\in s_1(2),\\
  \neg \Kw_2 q & \text{ if }p_1 \in s_1(2).
\end{cases}
$

The goal formulas are then defined as:
\begin{itemize}
\item $\gamma_1 = \delta_v \wedge (\alpha \vee \mathit{flip}(p_1))$
\item $\gamma_2 = \top$.
\end{itemize}

Let $\cs$ be any uniform strategy profile such that for all $w \in
V(P_1 \cup \{q\})$ with $w \cap P_1 =v$ we have $\cs(w)=s$. Now
consider a $w \in V(P_1 \cup \{q\})$ such that $w \cap P_1 =v$.

Suppose $w, s \not\models \alpha$. By the definition of
$\mathit{flip}(p_1)$, we have that $w, \cs(w) \not \models
\mathit{flip}(p_1)$ and thus $w, \cs(w) \not\models \gamma_1$. Again, by
the definition of $\mathit{flip}(p_1)$, there exists $s'_1$ such that
$w, (s'_1, \cs_{-1}(w)) \models \mathit{flip}(p_1)$ and therefore
$u_1(w, (s'_1, \cs_{-1}(w))) > u_1(w, \cs(w))$. Thus $\cs \not\in
\Ne_{\mathsf{max}}(G)$.

Conversely, suppose $w, s \models \alpha$. Then for player 1, $w,
\cs(w) \models \gamma_1$. For all $w'\in V(P_1 \cup \{q\})$ such that
$w' \cap P_1 \neq v$, for all $s'_1 \in S_1$, $w', (s'_1, \cs_{-1}(w)
\not \models \delta_v$ and therefore, $w', (s'_1, \cs_{-1}(w')) \not
\models \gamma_1$. For player 2, for all valuations $u \in V(P_1
\cup \{q\})$, we have $u, \cs(u) \models \gamma_2$. Therefore $\cs \in
\Ne_{\mathsf{max}}(G)$.

By Theorem \ref{thm:mc} PSPACE-hardness follows, which gives the desired
result.
\end{proof}

In the case of pessimist and optimist outcome relations, the following
computational upper bounds for the verification question are
relatively straightforward. Whether matching lower bounds can be
shown is an interesting question.

\begin{theorem}
  \label{thm:compVerifNePess}
Given an observation game $\obsgame = (N,(P_i)_{i \in N},
  (\gamma_i)_{i \in N})$ and a uniform strategy profile $\cs \in \CS$,
  checking if $\cs \in \Ne_{\mathsf{x}}(G)$ is in PSPACE where
  $\mathsf{x} \in \{\mathsf{pess}, \mathsf{opt}\}$. 
\end{theorem}
\begin{proof}
 Observe that by Theorem \ref{thm:mc}, for $i \in N$,
  $\cs\in \CS$ and $v \in V$, checking if $u_i(v,\cs(v))=1$ can be
  done in PSPACE. It follows that checking if $\max
  \uu_i(v,\cs)=1$ (respectively, if $\min \uu_i(v,\cs)=1$) can
  be checked in PSPACE.
Therefore, to check if $\cs \not\in \Ne_{\mathsf{pess}}(G)$
  (respectively, if $\cs \not\in \Ne_{\mathsf{opt}}(G)$), it
  suffices to perform the following two steps.
  \begin{enumerate}
  \item Guess a player $i$, a valuation $v$ and a strategy $s'_i \in S_i$.
  \item Verify if $\min \uu_i(v,\cs) < \min \uu_i(v,(\dot{s}'_i,\cs_{-i}))$, \\ (respectively, if $\max \uu_i(v,\cs) < \max \uu_i(v,(\dot{s}'_i,\cs_{-i}))$).
  \end{enumerate}
  
  For step 1 note that the size of the strategy $|s_i| =
  \bigO(nk)$. Thus the triple $(i,v,s_i)$ that forms a possible
  witness to the fact that $\cs \not\in \Ne_{\mathsf{pess}}(G)$
  (respectively, $\cs \not\in \Ne_{\mathsf{opt}}(G)$), has a polynomial
  representation. By the observation above, step 2 can be solved in
  PSPACE. Since PSPACE is closed under complementation and NPSPACE =
  PSPACE, the membership in PSPACE follows.
\end{proof}

\subsection{Emptiness}
Next we address the complexity of checking
for emptiness of maximal Nash equilibria in observation games. We find
it useful to introduce the following definitions.
Let $A=\{a_1, \ldots, a_l\}$ and $B=\{b_1, \ldots, b_l\}$ be two
finite sets of variables where $|A|=|B|$ and let $\zeta: A \to B$
be a bijection. For valuations $v^1 \in V(A)$ and $v^2
\in V(B)$, we say that $\cons_\zeta(v^1,v^2)$ holds if for all $j: 1 \leq j
\leq l$, $a_j \in v^1$ iff $\zeta(a_j) \in v^2$. We also define the formula
$\fcons_\zeta(A,B) := \wedge_{j=1}^l(a_j \leftrightarrow \zeta(a_j))$.

Given a uniform strategy $\cs_i$ and a set $Z \subseteq P_i$, we say
that $\cs_i$ is \textit{globally $Z$-uniform} if for all $v, v' \in
V$, if $\restr{v}{Z} = \restr{v'}{Z}$ then $\cs_i(v) = \cs_i(v')$. For
$i \in N$, let $\CS_i^Z = \{\cs_i \in \CS_i \mid \cs_i \text{ is
globally } Z
\text{-uniform}\}$. Note that $\CS_i^Z$ can be viewed as a natural
generalisation of $\CS_i^g$ by parameterising the uniform
strategies on the set $Z$.

\medskip

\noindent {\bf An NEXPTIME-complete Problem.}  We now show that given
an observation game $\obsgame$, checking if
$\Ne_{\mathsf{max}}(\obsgame)$ is empty is NEXPTIME-complete. To prove
the hardness, we give a reduction from the {\sc Dependency quantifier
  Boolean formula game (Dqbfg)} \cite[p.87]{HD09}. {\sc Dqbfg}
involves a three player game with players 1, 2 and 3. There are four
finite sets of variables which are mutually disjoint, $X_2, X_3,
A_2$ and $A_3$ along with a Boolean formula $\varphi$ over the
variables $X_2 \cup X_3 \cup A_2 \cup A_3$. Let $X=X_2 \cup X_3$
and $A= A_2 \cup A_3$. For the rest of this section we use $L^B$ to
denote the set of Boolean formulas over the variables $X \cup
A$. Players' strategies are defined as follows.

\begin{itemize}
\item Player 1: a strategy $t_1 \in V(X)$.
\item Player 2: a strategy $t_2: V(X_2) \to V(A_2)$.
\item Player 3: a strategy $t_3: V(X_3) \to V(A_3)$.
\end{itemize}

In other words, a strategy for player 1 is to select a valuation for
variables in $X$. Player 2 chooses a valuation for variables in
$A_2$ and his strategy can depend on the valuation for variables in
$X_2$. Similarly, a strategy for player 3 is to choose a valuation for
variables in $A_3$ which can depend on the valuation of
variables in $X_3$.

For player $i \in \{1,2,3\}$ let $T_i$ denote the set of strategies of
player $i$ and $T$ the set of strategy profiles. It is easy to observe
that a strategy profile $t=(t_1,t_2,t_3)$ defines a valuation over the
set of variables $X \cup A$. For a formula $\alpha \in L^B$ we then
have the natural interpretation for $t \models \alpha$. Given
strategies $t_2 \in T_2$ and $t_3 \in T_3$, we say that the pair
$(t_2, t_3)$ is a \textit{winning strategy} for the coalition of
players 2 and 3 if for all $t_1 \in T_1$, $(t_1, t_2, t_3) \models
\neg \phi$.

An instance of {\sc Dqbfg} is then given by the tuple $H=((X_i)_{i \in
  \{2,3\}}, (A_i)_{i \in \{2,3\}}, \varphi)$ and the associated
decision problem is to check if the coalition of players 2 and 3 have
a winning strategy in $H$.

\begin{theorem}[\cite{HD09}]
{\sc Dqbfg} is NEXPTIME-complete.
\end{theorem}

\noindent {\bf The reduction.} Given an instance of {\sc Dqbfg}
$H=((X_i)_{i \in \{2,3\}}, (A_i)_{i \in \{2,3\}}, \varphi)$, we
construct an observation game $G_H=(N, (P_i)_{i \in N}, (\gamma_i)_{i
  \in N})$ as follows. The set of players $N = \{1,2,3\}$. For $i \in
\{2,3\}$, let $Y_i$ be a copy of the variables in $X_i$, so $|Y_i|
= |X_i|$ and let $Y=Y_2 \cup Y_3$. Let $P_1 = X$, $P_2= A_2 \cup Y_2
\cup \{q\}$ and $P_3 = A_3 \cup Y_3 \cup \{r\}$. For the rest of this
section, we use $V$ and $L^K$ to denote the set of all valuations and
the set of all formulas over the variables in the observation game
$G_H$ respectively (so $V= V(X \cup Y \cup A \cup \{q,r\})$).

We also define the bijection $\zeta: X \to Y$ as the function that
maps each variable in $X_i$ to its corresponding copy in
$Y_i$. Formally, let $X_1 = \{x_1^1, \ldots, x_1^l\}$, $Y_1
= \{y_1^1, \ldots, y_1^l\}$, $X_2 = \{x_2^1, \ldots, x_2^h\}$ and $Y_2
= \{y_2^1, \ldots, y_2^h\}$. Then $\zeta(x_1^j) = y_1^j$ for all
$j \in \{1, \ldots, l\}$  and $\zeta(x_2^j) = y_2^j$ for all
$j \in \{1, \ldots, h\}$. To simplify notation, we denote
$\cons_\zeta$ by $\cons$ and $\fcons_\zeta$ by $\fcons$ for this fixed
bijection $\zeta$.

In order to define the goal formulas, we first inductively define a
function $\lambda: L^B \to L^K$ that transforms $\varphi$ to a formula
in $L_K$ as follows. 
\begin{itemize}
\item For $p \in X$, $\lambda(p):=p$.
\item For $p \in A_2$, $\lambda(p):= \Kw_3 p$.
\item For $p \in A_3$, $\lambda(p):= \Kw_2 p$.
\item $\lambda( \neg \alpha) := \neg \lambda(\alpha)$.
\item $\lambda(\alpha_1 \vee \alpha_2):= \lambda(\alpha_1) \vee
  \lambda(\alpha_2)$.
\end{itemize}

Let $\psi_2 = (\Kw_2 r \eq \neg \Kw_3 q)$ and $\psi_3 = (\Kw_2 r \eq
\Kw_3 q)$. Recall that Example \ref{ex:noNE} shows that already for
the $\Kw$ fragment of observation games, $\Ne_{\mathsf{max}}$ need not
always exist. Observe that the formulas $\psi_2$ and $\psi_3$
precisely correspond to $\gamma_1$ and $\gamma_2$ respectively as used
in Example \ref{ex:noNE}. We define the players' goal formulas as
follows.

\begin{itemize}
\item $\gamma_1 = \top$.
  \item For $i \in \{2,3\}$, $\gamma_i = (\lambda(\neg \varphi) \vee
    \psi_i) \wedge \fcons(X_2,Y_2) \wedge \fcons(X_3,Y_3)$.
\end{itemize}

\noindent{\bf Properties of $G_H$.} It is easy to see that the resulting
observation game $G_H$ is polynomial in the size of $H$. We first make
the following observations about $G_H$.

\begin{lemma}
  \label{lm:noNeMax}
  Let $G_H$ be the observation game corresponding to $H$ and let $\cs
  \in \CS$. If there exists $v \in V$ such that
  $\cons(\restr{v}{X_1}, \restr{v}{Y_1})$,
  $\cons(\restr{v}{X_2},\restr{v}{Y_2})$ and $v, \cs(v)
  \models \lambda(\varphi)$ then $\cs \notin \Ne_{\mathsf{max}}(G_H)$.
\end{lemma}
\begin{proof} Suppose there exists $v \in V$ such that
  $\cons(\restr{v}{X_1}, \restr{v}{Y_1})$,
  $\cons(\restr{v}{X_2},\restr{v}{Y_2})$ and $v, \cs(v) \models
  \lambda(\varphi)$. Then $v,\cs(v) \models \fcons(X_2,Y_2)
  \wedge \fcons(X_3,Y_3)$. By Example \ref{ex:noNE}, we have that
  there exists $i \in \{2,3\}$ such that $v, \cs(v) \not\models
  \psi_i$ and there exists $s_i \in S_i$ such that $v, (s_i,
  \cs_{-i}(v)) \models \psi_i$. Therefore we have $u_i(v, \cs(v)) <
  u_i(v, (s_i,\cs_{-i}(v)))$. Thus $\cs \not\in
  \Ne_{\mathsf{max}}(G_H)$.
\end{proof}

\begin{lemma}
  \label{lm:compTruth}
  For $i \in \{2,3\}$, for all $s \in S$, for all $v,v' \in V$ such
  that $\restr{v}{(X \cup Y)} = \restr{v'}{(X \cup Y)}$ we have $v, s
  \models \gamma_i$ iff $v', s \models \gamma_i$.
\end{lemma}
\begin{proof}
For $i \in \{2,3\}$, the claim clearly holds for formulas $\psi_i$ and
$\fcons(X_i,Y_i)$. Thus for $\gamma_i$, the claim follows by a simple
induction on $\varphi$.
\end{proof}

Next, we show that if the set of maximal Nash equilibria in $G_H$ is
non-empty then this set contains certain restricted types of maximal Nash
equilibria.

Let $\CStwo$ denote the set of uniform strategy profiles $\cs \in \CS$
that satisfy the following conditions:
  \begin{itemize}
  \item $\cs_1 \in \CS_1^g$,
  \item for $i \in \{2,3\}$, $\cs_i \in \CS_i^{Y_i}$.
  \end{itemize}

In other words, $\CStwo$ consists of the set of all uniform strategy profiles
$\cs$ such that $\cs_1$ is globally uniform and for $i \in \{2,3\}$,
$s_i$ is globally $Y_i$-uniform.

\begin{lemma}
\label{lm:setUniform}
If $\Ne_{\mathsf{max}}(G_H) \neq \emptyset$ then there exists $\cs^* \in
\Ne_{\mathsf{max}}(G_H)$ such that $\cs^* \in \CStwo$.
\end{lemma}
\begin{proof}
For players $i \in \{2,3\}$ we define an equivalence relation $\cong_i
\subseteq V \times V$ as follows. For $v, v' \in V$, $v \cong_i v'$ if
$\restr{v}{Y_i} = \restr{v'}{Y_i}$. For $v \in V$, let $[v]_i$ denote
the equivalence class containing the valuation $v$ and $c_v^i \in
[v]_i$ denote a fixed valuation which is interpreted as the canonical
representative element in the equivalence class $[v]_i$.

Suppose $\cs \in \Ne_{\mathsf{max}}(G_H)$. Consider the uniform strategy
profile $\cs^* \in \CStwo$ defined as follow.

\begin{itemize}
  \item For player 1, fix a valuation $w \in V$ and let $\cs^*_1(v) =
    \cs_1(w)$ for all $v \in V$.
    \item For players $i \in \{2,3\}$, for all $v \in V$,
      $\cs_i^*(v)=\cs_i(c_v^i)$.
\end{itemize}

We claim that $\cs^* \in \Ne_{\mathsf{max}}(G_H)$. Suppose not, then
there exists $i \in \{2,3\}$, there exists $w \in V$, there exists
$s_i \in S_i$ such that $u_i(w, (s_i, \cs^*_{-i}(w))) > u_i(w,
\cs^*(w))$. Then $w, (s_i, \cs^*_{-i}(w)) \models \gamma_i$ and $w,
\cs^*(w) \not\models \gamma_i$.

Now consider the valuation $u$ defined as follows:
$\restr{u}{P_1}=\restr{w}{P_1}$ and for $i \in \{2,3\}$,
$\restr{u}{P_i} = \restr{c^i_w}{P_i}$. By definition of $u$ we have
that $\restr{u}{(X \cup Y)} = \restr{w}{(X \cup Y)}$ and therefore, $w
\cong_i u$. From the definition of $\cs^*$ it follows that
$\cs^*(w)=\cs(u)$.

Since $w, (a_i, \cs^*_{-i}(w)) \models \gamma_i$ we have that $w,
(a_i, \cs^*_{-i}(u)) \models \gamma_i$. By Lemma \ref{lm:compTruth} we
have that $u, (a_i, \cs^*_{-i}(u)) \models \gamma_i$. Since $w,
\cs^*(w) \not\models \gamma_i$ we have that $w, \cs^*(u) \not\models
\gamma_i$. By Lemma \ref{lm:compTruth} we have that $u, \cs^*(u)
\not\models \gamma_i$. However, this implies that $\cs^* \notin
\Ne_{\mathsf{max}}(G_H)$ which is a contradiction.
\end{proof}

\noindent {\bf Strategy Translation.} Note that by the construction of
$G_H$, the strategies of player 1 are irrelevant in terms of existence of
maximal Nash equilibria. Player 1 can ensure a utility of 1 by
choosing any strategy. We now define two functions which translate
strategies of players 2 and 3 between $H$ and $G_H$. For the rest of
the section we make use of the following concise notation. For $i =
2$, let $i^+ = 3$ and for $i = 3$, let $i^+ =2$.

For $i \in \{2,3\}$, let $\chi_i: T_i \to \CS_i^{Y_i}$ be the function
that translates every strategy $t_i$ of player $i$ in $H$ to a globally
$Y_i$-uniform strategy $\cs_i=\chi_i(t_i)$ in $G_H$ as defined below.
\begin{itemize}
\item For all $v \in V$, if $\cons(\restr{v}{X_i}, \restr{v}{Y_i})$
    then $\cs_i(v)(i^+)=t_i(\restr{v}{X_i})$ and $\cs_i(v)(i^+)=
    \emptyset$ otherwise. For all $v \in V$, $\cs_i(v)(1)=\emptyset$
    and $\cs_i(v)(i)=P_i$.    
  \end{itemize}

For $i \in \{2,3\}$, let $\mu_i:  \CS_i^{Y_i} \to T_i$ be the function
that translates every globally $Y_i$-uniform strategy $\cs_i$ of player $i$ in
$H$ to a strategy $t_i=\mu_i(\cs_i)$ in $G_H$ as defined below.
\begin{itemize}
  \item For all $v \in V$, such that $\cons(\restr{v}{X_i},
  \restr{v}{Y_i})$, $t_i(\restr{v}{X_i})=\cs_i(v)(i^+)$.
\end{itemize}

Note that since $\cs_i \in \CS_i^{Y_i}$, $\mu_i$ is well defined.

\begin{lemma}
  \label{lm:chiMu}
  For all $i \in \{2,3\}$ and for all $\cs_i \in \CS^{Y_i}$, let
  $\cs_i'=\chi_i(\mu_i(\cs_i))$. For all $\cs_1, \cs_1' \in \CS_1$,
  for all $i \in \{2,3\}$ and for all $v \in V$ such that
  $\cons(\restr{v}{X_i}, \restr{v}{Y_i})$ 
  %% $\cons(\restr{v}{X_2}, \restr{v}{Y_2})$ and
  %% $\cons(\restr{v}{X_3},\restr{v}{Y_3})$,
  we have $v, (\cs_1,
  \cs_2,\cs_3)(v)) \models \gamma_i$ iff $v, (\cs'_1,
  \cs'_2,\cs'_3)(v)) \models \gamma_i$.
\end{lemma}
\begin{proof}
For $i \in \{2,3\}$, the claim clearly holds for the formulas
$\psi_i$ and $\fcons(X_i,Y_i)$. Thus for
$\gamma_i$, the claim follows by induction on $\varphi$.
\end{proof}

\begin{lemma}
\label{lm:truthPres}
For all $\alpha \in L^B$,
% (over the variables $X_2 \cup X_3 \cup A_2 \cup A_3$),
for all $t \in T$, for all $i \in \{2,3\}$ and for all $v \in V$ such
that $\restr{t}{X} = \restr{v}{X}$ and $\cons(\restr{v}{X_i},
\restr{v}{Y_i})$
we have
$t \models \alpha$ iff $v, (\cs_1, \chi_2(t_2),\chi_3(t_3))(v)) \models
\lambda(\alpha)$ for all $\cs_1 \in \CS_1$.
\end{lemma}
\begin{proof}
For $i \in \{2,3\}$, let $\cs_i = \chi_i(t_i)$.  The proof is by
induction on the structure of $\alpha$ where the interesting cases
involve the three base cases.
  \begin{itemize}
    \item $\alpha = p \in X$. Then we have $\lambda(p)=p$
      and the following sequence of equivalences. $t \models p$ iff $p
      \in t_1$ iff $p \in v$ (since $\restr{t_1}{X} = \restr{v}{X}$)
      iff $v, (\cs_1, \cs_2,\cs_3)(v)) \models p$ for all $\cs_1
      \in \CS_1$.
    \item $\alpha = p \in A_2$. Then we have $\lambda(p)= \Kw_3 p$ and
      the following sequence of equivalences. $t \models p$ iff $p \in
      t_2(\restr{t_1}{X_2})$ iff $p \in \cs_2(v)(3)$ (since
      $\cons(\restr{v}{X_2}, \restr{v}{Y_2})$) iff $v, (\cs_1,
      \cs_2,\cs_3)(v)) \models \Kw_3 p$ for all $\cs_1 \in
      \CS_1$.
    \item $\alpha = p \in A_3$. Then we have $\lambda(p)= \Kw_2 p$ and
      the following sequence of equivalences. $t \models p$ iff $p \in
      t_3(\restr{t_1}{X_3})$ iff $p \in \cs_3(v)(2)$ (since
      $\cons(\restr{v}{X_3}, \restr{v}{Y_3})$) iff $v, (\cs_1,
      \cs_2,\cs_3)(v)) \models \Kw_2 p$ for all $\cs_1 \in \CS_1$.
    \item For $\alpha = \neg \alpha_1$ and $\alpha = \alpha_1 \vee
      \alpha_2$ the claim follows by a direct application of the
      induction hypothesis.
  \end{itemize}
\end{proof}

\begin{lemma}
  \label{lm:nexpC}
Let $H=((X_i)_{i \in \{2,3\}}, (A_i)_{i \in \{2,3\}}, \varphi)$ be an
instance of {\sc Dqbfg} and $G_H$ the associated observation
game. The coalition of players 2 and 3 have a winning strategy in $H$ iff
$\Ne_{\mathsf{max}}(G_H) \neq \emptyset$.
\end{lemma}
\begin{proof}
Let $(t_2,t_3)$ be a winning strategy for the coalition of players 2
and 3 in $H$. By definition of a winning strategy, for all $t_1 \in
T_1$, we have $(t_1, t_2, t_3) \models \neg \varphi$. Let $t=(t_1,
t_2, t_3)$ and consider the observation game $G_H$.

Note that in $G_H$, by the definition of player 1's goal $\gamma_1$,
we have for all $v \in V$ and for all $\cs \in \CS$, $u_1(v,
\cs(v))=1$. Now consider an arbitrary valuation $v \in V$. There are
two cases to consider.

\begin{itemize}
 \item[Case 1.] Suppose there exists $i \in \{2,3\}$ such that
   $\cons(\restr{v}{X_i},\restr{v}{Y_i})$ does not hold. By semantics,
   for all $\cs \in \CS$ and for all $i \in \{2,3\}$ we have $v,\cs(v)
   \not\models \fcons(X_2,Y_2) \wedge \fcons(X_3,Y_3)$ and thus
   $v,\cs(v) \not\models \gamma_i$. Therefore, $u_i(v,\cs(v)) = 0$.
 \item[Case 2.] Suppose for all $i \in \{2,3\}$,
   $\cons(\restr{v}{X_i},\restr{v}{Y_i})$ holds. By semantics we have
   for all $\cs \in \CS$, for all $i \in \{2,3\}$, $v,\cs(v) \models
   \fcons(X_2,Y_2) \wedge \fcons(X_3,Y_3)$. Let $t_1'=\restr{v}{X}$
   and $t'=(t_1',t_2,t_3)$. By definition of $t'$, we have
   $\restr{t'}{X}=\restr{v}{X}$. Since $(t_2,t_3)$ is a winning
   strategy for players 2 and 3 in $H$, we have $(t_1', t_2, t_3)
   \models \neg \varphi$. By Lemma \ref{lm:truthPres} $v,
   (\cs_1,\chi_2(t_2),\chi_3(t_3))(v) \models \lambda(\neg \varphi)$.
\end{itemize}
Since the choice of $v$ was arbitrary, we can conclude that
$(\cs_1,\chi_2(t_2),\chi_3(t_3)) \in \Ne_{\mathsf{max}}(G_H)$. In
fact, note that the argument shows a stronger claim - for all $\cs'_1
\in \CS_1$, the uniform strategy profile $(\cs'_1,
\chi_2(t_2),\chi_3(t_3)) \in \Ne_{\mathsf{max}}(G_H)$.

\noindent $(\Leftarrow)$ Suppose $\Ne_{\mathsf{max}}(G_H) \neq
\emptyset$. By Lemma \ref{lm:setUniform}, there exists a $\cs \in
\CStwo$ such that $\cs \in \Ne_{\mathsf{max}}(G_H)$. Let $(t_2,t_3) =
(\mu_2(\cs_2), \mu_3(\cs_3))$. We argue that $(t_2,t_3)$ is a winning
strategy for the coalition of players 2 and 3 in $H$.

Suppose not, then there exists $t'_1 \in T_1$ such that for the 
strategy profile $t'=(t'_1, t_2, t_3)$, we have $t' \models \varphi$. Consider
the pair of strategies $(\cs'_2, \cs'_3) = (\chi_2(t_2),\chi_3(t_3))$
and a valuation $v \in V$ such that $\restr{v}{X} = \restr{t'}{X}$ and
for all $i \in \{2,3\}$, $\cons(\restr{v}{X_i}, \restr{v}{Y_i})$. By
Lemma \ref{lm:truthPres} we have that $v, (\cs'_1, \cs'_2, \cs'_3)
\models \lambda(\varphi)$ for all $\cs'_1 \in \CS_1$. In particular,
$v, (\cs_1, \cs'_2, \cs'_3) \models \lambda(\varphi)$. Since
$\cs_i'=\chi_i(\mu_i(\cs_i))$ for $i \in \{2,3\}$, by Lemma
\ref{lm:chiMu} we have that $v, (\cs_1, \cs_2, \cs_3) \models
\lambda(\varphi)$.
By Lemma \ref{lm:noNeMax}, $\cs \not\in \Ne_{\mathsf{max}}(G_H)$ which
is a contradiction.
\end{proof}

\begin{theorem}
  \label{thm:emptyNeMax}
Given an observation game $G$, checking if $\Ne_{\mathsf{max}}(G) \neq
\emptyset$ is NEXPTIME-complete.
\end{theorem}
\begin{proof}
Recall that for each player, a uniform strategy $\cs_i$ can be encoded
as a tuple of Boolean functions $(\cs_i^j(p_i))_{j \in N, p_i \in
  P_i}$ each of which can be represented by a propositional formula
$\beta^j_i(p_i)$ whose size is at most exponential in $k$. To show
that the problem is in NEXPTIME, we first guess a uniform strategy
profile $\cs$. This involves guessing $n^2k$ formulas each of which
can be exponential in $k$.
Membership in NEXPTIME then follows from
Theorem~\ref{thm:compVerifNeMax}.

By Lemma \ref{lm:nexpC}, it follows that checking if
$\Ne_{\mathsf{max}}(G) \neq \emptyset$ is NEXPTIME-hard. Thus the 
claim follows.
\end{proof}

In the case of pessimist and optimist outcome relations, an
  argument similar to that given in the proof of Theorem \ref{thm:emptyNeMax}
  along with Theorem \ref{thm:compVerifNePess} immediately gives us an
  upper bound on the complexity of emptiness problem.

\begin{theorem}
  \label{thm:KwemptyNePess}
Given an observation game $G$, checking if $\Ne_{\mathsf{x}}(G)
  \neq \emptyset$ is in NEXPTIME where $\mathsf{x} \in
  \{\mathsf{pess}, \mathsf{opt}\}$.
\end{theorem}

\subsection{Knowing-Whether Observation Games} \label{subsec.complkw}
\label{subsec:KwComp}
In the case of $\Kw$  games, we show that both the
verification problem and the emptiness problem have ``better''
complexity bounds which match the known complexity results for the
corresponding questions in Boolean games. We first recall the relevant
results for Boolean games.

\begin{theorem}[\cite{HTW17}]
  \label{thm:compVerifBool}
  \textit{(Verification)} Given a Boolean game $\boolgame$ along with
  a strategy profile $v$ checking if $v \in \Ne(\boolgame)$ is
  co-NP-complete.
\end{theorem}

\begin{theorem}[\cite{BonzonLLZ06}]
  \label{thm:emptyNeBool}
  \textit{(Emptiness)} Given a Boolean game $\boolgame$, checking if
  $\Ne(\boolgame) \neq \emptyset$ is $\Sigma^p_2$-complete.
\end{theorem}

In the context of $\Kw$ games, as an immediate consequence
of Proposition~\ref{lm:KwAllVal} we get that the model checking question for
the fragment $L^{\Kw}$ is in polynomial time.

\begin{corollary}
  \label{cor:Kwmc}
  Given $\alpha \in L^{\Kw}$ along with a strategy profile $s \in S$ and a
  valuation $v \in V$, checking if $v, s \models \alpha$ is in PTIME.
\end{corollary}

We then have the following results for the complexity of verification
and emptiness in $\Kw$ games.

\begin{theorem}
  \label{thm:KwcompVerifNeMax}
  Given a $\Kw$ game $\obsgame = (N,(P_i)_{i \in N}, (\gamma_i)_{i \in
    N})$ and a uniform strategy profile $\cs \in \CS$, checking if
  $\cs \in \Ne_{\mathsf{max}}(G)$ is co-NP-complete.
\end{theorem}
\begin{proof}
  Membership in co-NP follows immediately from Corollary
  \ref{cor:Kwmc}. For hardness, we show a reduction from the
  corresponding verification problem in Boolean games which is: given
  a Boolean game $\boolgame$ and a strategy profile $v$ in
  $\boolgame$, to check if $v \in \Ne(\boolgame)$. By Theorem
  \ref{thm:compVerifBool} this problem is known to be co-NP-complete.

  Given a Boolean game $\boolgame$ and a strategy profile $w$ in
  $\boolgame$, let $\obsgame_\boolgame$ and $\dot{s^w}$ be the
  corresponding observation game and the globally uniform strategy
  profile in $\obsgame_\boolgame$ as defined in Section
  \ref{subsec:BoolToKw}. We argue that $w \in \Ne(\boolgame)$ iff
  $\dot{s^w} \in \Ne_{\mathsf{max}}(\obsgame_\boolgame)$.

  \medskip
  
  \noindent $(\Rightarrow)$ This direction is exactly the same as the
  first part of the proof of Theorem \ref{thm:BoolToKwObs}. Suppose $w
  \in \Ne(\boolgame)$ and $\dot{s^w} \not\in
  \Ne_{\mathsf{max}}(\obsgame_\boolgame)$. Then there exists $i \in
  N$, $v \in V$ and $t_i \in S_i$ such that
  $\util_i(v,(t_i,\dot{s^w}_{-i}(v))) > \util_i(v,\dot{s^w}(v))$. Let
  $w'=\chi^{-1}([\dot{t}_i,\cs_{-i}])$ From Lemmas
  \ref{lm:BoolObsOutcome}, \ref{lm:kwglobalpayoff} and
  \ref{lm:BoolObsBiject} it follows that
  $\boolutil_i(w')=\util_i(v,(\dot{t}_i,\cs_{-i})(v)) >
  \util_i(v,\dot{s^w}(v))=\boolutil_i(w)$ for all $v \in V$. Therefore
  $w \not\in \Ne(\boolgame)$ which is a contradiction.

  \medskip
  
  \noindent $(\Leftarrow)$ Suppose $\dot{s^w} \in
  \Ne_{\mathsf{max}}(\obsgame_\boolgame)$ and $w \not\in
  \Ne(\boolgame)$. Then there exists $i \in N$ and $w_i'$ such that
  $u_i^\boolgame((w_i',w_{-i}) > u_i^\boolgame(w)$. Let
  $w'=(w_i',w_{-i})$. From Lemma \ref{lm:BoolObsOutcome} we have that
  $u_i(v,\dot{s^{w'}}(v)) = u_i^\boolgame(w') > u_i^\boolgame(w) =
  u_i(v,\dot{s^{w}}(v))$. This implies that $\dot{s^w} \not\in
  \Ne_{\mathsf{max}}(\obsgame_\boolgame)$ which is a contradiction.
\end{proof}

\begin{theorem}
  \label{thm:KwemptyNeMax}
Given a $\Kw$ game $G$, checking if $\Ne_{\mathsf{max}}(G) \neq
\emptyset$ is $\Sigma^P_2$-complete.
\end{theorem}
\begin{proof}
  Membership in $\Sigma^P_2$ follows immediately from Corollary
  \ref{cor:kwglobalNE} and Theorem \ref{thm:KwcompVerifNeMax}. For
  $\Sigma^P_2$-hardness, notice that the translation from observation
  games to Boolean games that we provide in Section
  \ref{subsec:KwToBool} is polynomial time computable. Thus given an
  instance of an observation game $\obsgame$, we can construct a
  Boolean game $\boolgame_\obsgame$ in polynomial time. By Theorem
  \ref{thm:KwObsToBool}, $\Ne_{\mathsf{max}}(\obsgame) \neq \emptyset$
  iff $\Ne(\boolgame_\obsgame) \neq \emptyset$. From Theorem
  \ref{thm:emptyNeBool} it follows that checking if
  $\Ne_{\mathsf{max}}(G) \neq \emptyset$ is $\Sigma^P_2$-complete.
\end{proof}

\section{Discussion and Conclusion}
\label{sec:concl}

\noindent {\bf Summary.}
We introduced Boolean observation games as a qualitative model which
combines aspects of imperfect and incomplete information games. For these games we
studied Nash equilibria based on different ways to compare sets of outcomes, that result in different expectations of outcomes. Our main technical
contributions are for the existence of Nash equilibria, for
the computational analysis of Nash equilibria, as well as for identifying
knowing-whether games, a fragment of observation games that
precisely corresponds to Boolean games in terms of existence of Nash
equilibria.  A summary of our results are listed in Table
\ref{tab:summary}.
\begin{table}[ht]
  \resizebox{\textwidth}{!}{%
  \begin{tabular}{|ll|l|l|l|p{2.3cm}|p{2.5cm}|}
    \hline
    \multirow{5}{*}{} & &
    \multicolumn{3}{|c|}{Existence} &
    \multicolumn{2}{c|}{Complexity} \\
    & & $|\mathit{type}(i)| \leq 2$ & $|\mathit{type}(i)| \leq 3$ & $|\mathit{type}(i)| > 3$ & Verification & Emptiness\\
        \hline
    \multirow{3}{*}{Observation games} %& & & & & &\\
    &$\Ne_{\mathsf{pess}}$ & Yes {\small (Theorem~\ref{thm:existsNEpess})}& Yes {\small (Theorem~\ref{thm:existsNEpess})}& No {\small (Example \ref{ex:KWnoNEpess})}& PSPACE {\small (Theorem~\ref{thm:compVerifNePess})}& NEXPTIME {\small (Theorem~\ref{thm:KwemptyNePess})}\\
    \cline{2-7}
    &$\Ne_{\mathsf{opt}}$ & Yes {\small (Theorem~\ref{thm:existsNEpess})} & Yes {\small (Theorem~\ref{thm:existsNEpess})}& --- & PSPACE {\small (Theorem~\ref{thm:compVerifNePess})}& NEXPTIME {\small (Theorem~\ref{thm:KwemptyNePess})}\\
    \cline{2-7}
    &$\Ne_{\mathsf{max}}$ & Yes {\small (Theorem~\ref{thm:existsNEmax})} & No {\small (Example \ref{prop:obs3nostable})}& No {\small (Example \ref{prop:obs3nostable})} & PSPACE-complete {\small (Theorem~\ref{thm:compVerifNeMax})}& NEXPTIME-complete {\small (Theorem~\ref{thm:emptyNeMax})}\\
            \hline
    \multirow{3}{*}{$\Kw$ games}%& & & & & &\\
    &$\Ne_{\mathsf{pess}}$ & Yes {\small (Theorem~\ref{thm:existsNEpess})}& Yes {\small (Theorem~\ref{thm:existsNEpess})} & --- & --- & ---\\
    \cline{2-7}
    &$\Ne_{\mathsf{opt}}$ & Yes {\small (Theorem~\ref{thm:existsNEpess})}& Yes {\small (Theorem~\ref{thm:existsNEpess})}& --- & --- & ---\\
    \cline{2-7}
    &$\Ne_{\mathsf{max}}$ & Yes {\small (Theorem~\ref{thm:existskwNE})}& Yes {\small (Theorem~\ref{thm:existskwNE})}& No {\small (Example \ref{prop:kw4nostable})}& Co-NP-complete {\small (Theorem~\ref{thm:KwcompVerifNeMax})} & $\Sigma^p_2$-complete {\small (Theorem~\ref{thm:KwemptyNeMax})}\\
    \hline
  \end{tabular}}
\caption{Summary of results. \label{tab:summary}}
\end{table}

\noindent {\bf Complexity and Existence.} Note that in Boolean
observation games, the underlying relational structure (the Kripke model) is not explicit. It is implicitly presented in terms of a valuation $v$ and
strategy profile $s$. Therefore, even the basic model checking problem is
PSPACE-complete given the compact presentation. This in turn is one of
the main reasons for the ``high'' complexity bounds that we obtain for
the computational analysis of this model.

An alternative would be to explicitly have a Kripke
model as part of the input. Suppose the Kripke model is
defined over a set of worlds $W$. Then a uniform strategy can be
thought of as a uniform function from $W$ to the set $S_i$ of strategies for player $i$, which would have a
polynomial representation in terms of the number of worlds $|W|$, the number of agents $n$, and the number of variables (atoms) $|P|$. Computing
$u_i(v,s(v))$ can then also be done in polynomial time. As a consequence it can
be shown that the verification problem is in co-NP and the emptiness
problem is in $\Sigma^P_2$. However, the size of the Kripke structure
can in the worst case be exponential in $|P|$.

Clearly, the computational complexity of the model requires further
analysis. There are two approaches which are interesting. The first is
to try and identify fragments of the model which provide better
complexity bounds. In the subclass of knowing-whether games we obtain
bounds which match the known bounds for the corresponding questions in
Boolean games. It is also known that in two player Boolean games where
the goal formulas are restricted to Horn-renamable DNF, 2CNF or
monotone CNF, the emptiness of Nash equilibrium can be checked in
polynomial time \cite{BonzonLLZ06}. By modifying the arguments
appropriately, we can identify subclasses of knowing-whether games in
which the corresponding emptiness question can be solved in polynomial
time.
In general, any natural restriction of the logical specification
language which results in the corresponding model checking question to
have ``better'' complexity is a promising fragment.

The second approach would be to identify the specific parameters
within the model which contribute to the exponential complexity
bounds. Some of the natural candidates are the number of players and
the number of variables used in the goal formulas. Since the hardness
results given in Theorem \ref{thm:compVerifNeMax} and Theorem
\ref{lm:nexpC} are for two and three players respectively, bounding
the number of players alone is not sufficient. Analysis of fragments
where the number of variables in the goal formulas are bounded appears
to be a promising research direction which require more careful study.

Analysing the lower bounds in the case of pessimist and optimist
outcome relations is another question which is relevant.

In Section \ref{subsec:existsNEKW}, we analyse existence of Nash
equilibria in knowing-whether games, and in Section
\ref{subsec:existsNEGeneral} we identify conditions based on
positive/negative epistemic assertions which ensure existence of Nash
equilibria. Identifying other fragments where Nash equilibria are
guaranteed to exist is an obvious direction of future research. It
would be particularly interesting if the existence result can be
related to structural properties of the underlying game.

\noindent {\bf Extensions of the Model.}
There are many extensions of the model which are interesting for
further research.  One could imagine a whole and ever widening range
of qualitative incomplete information games of imperfect
information. For the strategies, instead of merely revealing the value
of propositional variables, we could consider revealing the value of
any epistemic proposition, as already considered in
\cite{agotnesetal:2011,agotnesetal.qa:2011} for more complex, arbitrary, 
Kripke models. Instead of having merely partitions (exhaustive and
exclusive) of all variables, one could consider overlapping sets of
variables (exhaustive but not exclusive, so more than one player may
observe the same variables)\footnote{Kindly suggested by Paul
Harrenstein.} {as for example employed in \cite{belardinellietal:2017}}. Doing the same for Boolean games would create the
possibility of conflict, as not more than one player can control the
value of a variable. But as many agents as you wish can make the same
observation.

Another interesting extension to explore would be to consider iterated Boolean 
observation games, wherein players can gradually reveal more and more
of their variables. This would be a generalization similar to that
already studied for Boolean games in
\cite{GutierrezHW15,GutierrezHPW16}. It would involve epistemic
temporal goals or dynamic epistemic goals. Different from iterated
Boolean games, in iterated Boolean observation games one can only
reveal more and more variables in every round, until all have been
revealed. This should therefore considerably reduce the complexity of iterated Boolean observation games with respect to otherwise comparable iterated Boolean games.

Yet another relevant direction is (epistemic) incentive engineering in
Boolean observation games, similar to what is studied in Boolean games
\cite{WEK13,Tur13,HTW17}.

\subsection*{Acknowledgements} We thank the reviewers for their valuable comments which helped
  improve the presentation of the paper. Hans van Ditmarsch gratefully acknowledges a stay at IIT Kanpur initiating this research collaboration, during his temporary CNRS affiliation in 2018 at UMI ReLaX, Chennai, India. Sunil Simon
was partially supported by the grant CRG/2022/006140.

\bibliographystyle{plain}
\bibliography{biblio2021,extrabib}

\appendix
\section{Dynamic Epistemic Logic} \label{appendix} %\label{section.del}

\subsection{Proof in Section~\ref{sec.logic}} \label{app.proofs}

\restatabletwo*
\begin{proof}
The proof is by induction on the structure of $\Kw$ formulas in negation normal form ($L^\Kw_\nnf$). The direction from right to left is by definition. For the direction from left to right we proceed as follows.

Case atom: $v,s \models \Kw_i p_j$, iff (for all $w \sim^s_i v$, $w,s \models p_j$ iff $v,s \models p_j$), iff (for all $w$ with $w \inter P_i(s) = v \inter P_i(s)$, $w,s \models p_j$ iff $v,s \models p_j$), iff $p_j \in P_i(s)$. As $v$ no longer appears in the final statement, $v$ is arbitrary. Therefore, the initial statement $v,s\models\Kw_i p_j$ is equivalent to ``for all $w \in V$, $w,s\models\Kw_i p_j$,'' in other words, to $s\models\Kw_i p_j$.

Case negated atom: $v,s \models \neg\Kw_i p_j$, iff (there are $w,x \in V$ with $w \sim^s v$ and $x \sim^s v$ and such that $w,s \models p_j$ and $x,s \models \neg p_j$), iff (there are $w,x \in V$ with $w \inter P_i(s) = x \inter P_i(s) = v \inter P_i(s)$ and such that $w,s \models p_j$ and $x,s \models \neg p_j$), iff $p_j \notin P_i(s)$. As in the previous case, the final statement is independent from $v$ and therefore the initial statement is equivalent to $s \models \neg \Kw_i p_j$. 

Case conjunction: $v,s \models \alpha \et \beta$, iff $v,s \models \alpha$ and $v,s \models \beta$, iff (IH) $s \models \alpha$ and $s \models\beta$, iff $s \models \alpha\et\beta$.

Case disjunction: $v,s \models \alpha \vel \beta$, iff $v,s \models \alpha$ or $v,s \models \beta$, iff (IH) $s \models \alpha$ or $s \models \beta$, iff $s \models \alpha\vel \beta$.
\end{proof}

\subsection{Strategies as Epistemic Actions} \label{section.del}

In this section we compare our modelling and our results with related work in epistemic logic. We model strategy profiles as epistemic actions in a dynamic epistemic logic, where we also discuss an alternative semantics of strategies resulting in far larger models. The alternatives can be compared on their game theoretical implications, which may help to motivate our preference.

\bigskip

The situation wherein each player only observes the value of its own variables, corresponds to a Kripke model where the accessibility relation is the initial observation relation, and a strategy profile corresponds to an action model that, when executed in this Kripke model, results in an updated model wherein the accessibility relation is the observation relation (for that strategy profile). In this section we make precise how. It may serve to illustrate that our setting is very simple. This was why we were able to obtain modelling and computational results for Boolean observation games that are close or analogous to those for Boolean games.

An {\em epistemic model} (Kripke model) $M$ is a triple $(W,\sim,\pi)$ where $W$ is an (abstract) domain of \emph{worlds} or \emph{states}, where $\sim$ is a collection of equivalence relations on $W$, one for each agent, denoted $\sim_a$ (also known as {\em indistinguishability relations}), and where $\pi$ is a \emph{valuation} (function) mapping each state $w \in W$ to the subset of the propositional variables $P$ that are true in that state. A pointed epistemic model $(M,w)$ is a pair consisting of an epistemic model and a state $w \in W$. %A pointed epistemic model is often called an epistemic model as well, without danger of ambiguity.

Now consider the situation in our observation games where each of $n$ players $1,\dots,n$ only observes the value of its own variables $P_i$, but before they enact/play a strategy $s_i$. We have implicitly modelled this as the strategy profile $s^\emptyset$ wherein no player reveals any variable. We can identify this situation with the following epistemic model. 

\begin{quote} The {\em initial observation model} $(\mathit{IM},v)$, where $\mathit{IM} = (V,\sim,\pi)$, is such that: \begin{itemize} \item domain $V$ is the set of valuations of $P$ ($V = \powerset(P)$); \item for each player $i \in N$ and valuations $v,w \in V$, $v \sim_i w$ iff $v \inter P_i = w \inter P_i$; \item for each $v \in V$, $\pi(v) = v$. \end{itemize}\end{quote}
Note that the relations are exactly as in interpreted systems \cite{faginetal:1995}.

Similarly, the result of playing strategy profile $s \in S$ given valuation $v \in V$ of observed variables, corresponds to an updated epistemic model.
\begin{quote} The {\em observation model} $(\mathit{IM}^s,v)$, where $\mathit{IM}^s = (V,\sim^s,\pi)$, is such that $V$ and $\pi$ are as for $\mathit{IM}$, whereas in this case $v \sim^s_i w$ iff $v \inter P_i(s) = w \inter P_i(s)$. \end{quote} We recall that $P_i(s) = \{ p \in P \mid
\text{there is a} \ j \in N \ \text{with}\ p \in s_j(i) \}$, the variables revealed to $i$ in $s$, where by definition $P_i(i)=P_i$ so that always $P_i \subseteq P_i(s)$.

Surely more interestingly, we can model a strategy profile as an independent semantic primitive namely as an {\em action model} $U$ such that \begin{quote} $v,s \models \phi$ iff $\mathit{IM} \otimes U, (v,s) \models \phi$ \end{quote} where the former is the satisfaction relation in our logical semantics for $L^K$ and the latter is the satisfaction relation in action model semantics. In order to establish that we first need to define action models and their execution (following details as in \cite{baltagetal:1998,hvdetal.del:2007,moss.handbook:2015}).

An {\em action model} $U$ is a triple $(E,\approx,\pre)$ where $E$ is a domain of \emph{actions}, for each player $i = 1,\dots,n$, $\approx_i$ is an equivalence relation on $E$, and $\pre$ is a {\em precondition} function mapping each action $e \in E$ to an executability precondition $\pre(e)$ that is a formula in some logical language $L$. The execution of an action model in an epistemic model $M = (W,\sim,\pi)$ is then defined as the restricted modal product $M \otimes U = (W',\sim',\pi')$ where $W' = \{ (w,e) \mid w \in W, e \in E, M,w \models \pre(e) \}$, where $(w,e) \sim'_i (w',e')$ iff $w \sim_i w'$ and $e \approx_i e'$, and where $\pi'(w,e) = \pi(w)$.

In the case of strategy profiles for observation games, the logical language of action model preconditions can be restricted to $L^B$, the Booleans (the language required to describe preconditions is therefore simpler than the language $L^K$ to describe epistemic goals), and a rather simple action model corresponds to a strategy profile $s$. A strategy profile can be identified with the following action model. In the definition, $\delta_v \in L^B$ is the description of the valuation $v$, defined as $\delta_v := \Et_{p\in v} p \et \Et_{p \notin v} \neg p$.
\begin{quote} A {\em strategy profile action model} $U^s$ is a triple $(V,\sim^s,\pre)$ where the set of actions is the set of valuations $V$, where for each $i = 1,\dots,n$, $v \sim^s_i w$ iff $v \inter P_i(s) = w \inter P_i(s)$, and where for each action $v \in V$, $\pre(v) = \delta_v$.
\end{quote} The domain of the strategy profile action model is therefore the same as the domain of an observation model, namely the set of all valuations.

In can be verified that \begin{quote} $\mathit{IM} \otimes U^s$ is isomorphic to $\mathit{IM}^s$. \end{quote} This is fairly elementary. We note that each action can only be executed in a single world --- $\mathit{IM}, v \models \delta_v$, so that the size of $\mathit{IM}^s$ is the same as the size of $\mathit{IM}$. Then, $(v,v) \sim_i (w,w)$ iff, by definition of action model execution, $v \sim_i v$ (in $\mathit{IM}$) and $v \sim^s_i w$ (in $U^s$), iff, by definition of these relations, $v \inter P_i = w \inter P_i$ and $v \inter P_i(s) = w \inter P_i(s)$. As the latter is a refinement of the former, the desired result that $v \inter P_i(s) = w \inter P_i(s)$ follows. Finally, $\pi'(v,v) = \pi(v) =  v$. And the valuations $\pi$ do not change.

In fact, already $U^s$ is isomorphic to $\mathit{IM}^s$ (slightly abusing the notion, but when we identify valuations with their description). It should be noted that it is common that action models are isomorphic to updated models when executed in initial models consisting of all valuations (and representing some sort of initial maximal ignorance over those valuations). 

As a word of warning: the `actions' that are the points in our action model $U^s$ do {\bf not} correspond to the strategies, that are sometimes also called actions. The action model `action' combines the strategies of all players simultaneously, so they rather correspond to strategy profiles. 

\medskip

\noindent {\bf More Succinct Action Models.}
A slightly more succinct modelling of strategy profiles as action models is conceivable, that is a quotient of the action model $U^s$ defined above with respect to variables that are not revealed by any player. Let us call this set $\overline{P^s}$, that is therefore defined as the complement of the set $P^s := \{ p \in P \mid \is i,j \in [1..n], i \neq j, p \in s_i(j) \}$. We can now redefine $U^s_{\mathsf{small}}$ as $(\powerset(P^s),\sim^s,\pre)$ where in this case for any $v,w \subseteq P^s$ (so for \emph{partial} valuations of atoms revealed by some agent only), $v \sim_i w$ iff $v \inter P_i(s) = w \inter P_i(s)$. This looks the same as before, but note that $P_i(s)$ may involve far more variables, namely in $\overline{P^s}$, than $v$ and $w$, that are both restricted to $P^s$. Also, still $\pre(v) = v$ for all $v \in P^s$ (and where $\pre(\emptyset)=\top$ in case $P^s =  \emptyset$). 

Again, it is elementary to show that $\mathit{IM} \otimes U^s_{\mathsf{small}}$ is isomorphic to $\mathit{IM}^s$. We now have that $\mathit{IM},w \models \pre(v)$ iff $v \subseteq w$. But in this case $U^s_{\mathsf{small}}$ is typically smaller than the resulting updated model $\mathit{IM}^s$. The resulting $\mathit{IM}^s$, as before, has the same domain as the initial model $\mathit{IM}$. 

We now have, for example, that the action model corresponding to the `reveal nothing' strategy profile $s^\emptyset$ is the trivial singleton action model $U^{s^\emptyset}_{\mathsf{small}}$ with precondition $\top$ (as $P^{s^\emptyset} = \emptyset$), and in this case $\mathit{IM} \otimes U^{s^\emptyset}_{\mathsf{small}}$ is isomorphic to the initial observation model $\mathit{IM}$ again: the relations $\sim_i$ have not changed.

\subsection{Strategies for Weaker Observations give Bigger Models} \label{appendix.weaker}

In our modelling, it is common knowledge to all players what variables have been revealed by who and to whom: the strategy profile $s$ is common knowledge `after the fact'. But, although I therefore know what variables are revealed by other players to yet other players, I still have not learnt the \emph{values} of these variables.

\begin{quote}
For example: After player 1 reveals atom $p_1$ to player $2$ and atom $q_1$ to player $3$, player $2$ knows whether $p_1$ and player $3$ knows whether $q_1$. Also, player $2$ knows that player $3$ knows whether $q_1$, and player $3$ knows that player $2$ knows whether $p_1$.
\end{quote}

In a different modelling, each player {\em only} learns what variables have been revealed by other players to herself, and what variables she reveals to others. 

\begin{quote}
For example: After player 1 reveals atom $p_1$ to player $2$ and atom $q_1$ to player $3$, player $2$ knows whether $p_1$ and player $3$ knows whether $q_1$. However, player $2$ does not know that player $3$ knows whether $q_1$, and player $3$ does not know that player $2$ knows whether $p_1$. Player $2$ also considers it possible that no variable has been revealed to $3$, in which case $3$ does not know whether $q_1$. And similarly for player $3$. 
\end{quote}

So, clearly, depending on which modelling one prefers, different goal formulas $\gamma$ of observation games may be satisfied, and it will therefore affect the existence of Nash equilibria and what the optimal strategies are.

Let us first formalize this as an action model, and let us be explicit about the (rather different) updated model as well. The strategies $s_i$ and profiles $s = (s_1,\dots,s_n)$ remain the same, and thus also the $P_i(s)$, the set of atoms revealed to agent $i$. However, we can no longer define an updated observation model as one wherein only the indistinguishability relations have been changed, namely as $v \sim_i^s w$ iff $v \cap P_i(s) = w \cap P_i(s)$, while keeping the domain (and the valuation). %The current modelling implicitly makes it common knowledge what atoms are revealed to what agents. 

Instead of models consisting of valuations (domain $V$) we now need much larger models consisting of pairs $(v,t)$ for valuations $v$ and profiles $t$ (domain $V \times S$) and define: \begin{quote} For all $v,v' \in V$ and for all $s,t,t' \in S$ and for all players $i \in N$: $(v,t) \sim_i^s (v',t')$ if $v \cap P_i(s) = v' \cap P_i(s)$ [same valuation inasfar observed], $t_i = t'_i = s_i$ [same variables revealed to others], and $P_i(s) = P_i(t) = P_i(t')$ [same variables revealed by others to you].\end{quote} As a consequence, we cannot describe the initial observation model as the one wherein $s^\emptyset$ is executed, because that would still blow up the model and introduce maximal uncertainty about what is revealed by who. So the initial observation model $\mathit{IM}$ needs to be given separately (namely as the model already defined in Appendix~\ref{section.del}). However, once this is done, that is all. An action model can also be given for this modelling.

In this alternative modelling the players would remain far more ignorant about other players: optimist expected outcome would be more optimist, pessimist expected outcome would be more pessimist, realist expected outcome would quantify over a far larger set of possible outcomes. Basically, any epistemic feature is diluted. It therefore appeared to us that our preferred modelling provides more interesting results and variations. 

Beyond that, the envisaged iterated Boolean observation games would become less meaningful for such strategies encoding weaker observations, as a player remains unaware of other players' increasing knowledge over such iterations, unless as a consequence of that player informing those other players.

\section{Representation and Complexity}

\thmmc*
\begin{proof}
  The membership in PSPACE is straightforward. For PSPACE-hardness, we
  give a reduction from {\sc Quantified Boolean Formula} (QBF) which
  is a canonical PSPACE-complete problem \cite{Papa94}.

  A QBF instance
  consists of a formula of the form $Q_1x_1 Q_2 x_2 \ldots Q_n x_n\ 
  \psi(x_1, x_2, \ldots, x_n)$ where every $Q_i$ is either a $\exists$
  or $\forall$ quantifier, every $x_i$ is a propositional variable and
  $\psi(x_1, x_2, \ldots, x_n)$ is a Boolean formula over the
  variables $x_1, \ldots, x_n$. From the definition, it follows that
  every QBF instance is either true or false (irrespective of the
  valuation under which it is evaluated).
  %% Also, note that if all $Q_j$'s in a
  %% QBF instance are $\exists$ quantifier, then the QBF is True iff the
  %% Boolean formula $\psi(x_1, x_2, \ldots, x_k)$ is satisfiable.

  Given an instance $\phi = Q_1x_1 Q_2 x_2 \ldots Q_n x_n\ \psi(x_1,
  x_2, \ldots, x_n)$ of QBF, we associate with each variable $x_i$, a
  player $i$ (thus $N = \{1, \ldots, n\}$) and let $P =\{x_1, \ldots,
  x_n\}$. We use the following notation introduced in Section
  \ref{subsec:BoolToKw}: for $i = 1,\dots,n-1$ let ${i^+} := i+1$ and
  ${n^+} := 1$. For all $i \in N$, let $P_i=\{x_{i^+}\}$ and let
  $s_i^*$ denote the strategy where player $i$ reveals $x_{i^+}$ to
  all players except player $i^+$. That is, $s^*_i(i^+)= \emptyset$
  and $s^*_i(j) = P_i$ for all $j \neq i^+$.

  Let $\alpha_\phi \in L^K$ be the formula obtained from $\phi$ by
  replacing all occurrence of $\forall x_i$ by $K_i$ and all
  occurrence of $\exists x_i$ by $\neg K_i \neg$. Let $v_{\bot} =
  \emptyset$ denote the valuation that assigns all variables the value
  false. We show that the QBF instance $\phi$ is true iff $v_\bot,
  s^* \models \alpha_\phi$.

  We first argue that for all QBF instances $\phi$ and for all
  valuations $v$ over $P$, $v \models \phi$ iff $v, s^* \models
  \alpha_\phi$. The proof is by induction on the structure of $\phi$
  and the non-trivial cases involve quantifiers. Suppose $\phi =
  \forall x_i \psi$ so that $\alpha_\phi = K_i \alpha_\psi$, then

 \medskip

  \begin{tabular}{lll}
    $v \models \forall x_i \psi$ & iff &for all valuations $u$ where
    $u \cap (P \setminus \{x_i\}) = v \cap (P \setminus \{x_i\})$, $u
    \models \psi$\\
    & iff &for all valuations $u$ where
    $u \cap (P \setminus \{x_i\}) = v \cap (P \setminus \{x_i\})$, $u, s^*
    \models \alpha_\psi$\\
    & iff &for all $u$ where $u \sim_i^{s^*} v$ we have $u, s^*
    \models \alpha_\psi$\\
    & iff &$v, s^* \models K_i \alpha_\psi$.
%    \alpha_\phi$.
  \end{tabular}

  \medskip

  Since all variables in the QBF instance $\phi$ are bound, we have
  the following. $\phi$ is true iff $v_\bot \models \phi$ iff $v_\bot,
  s^* \models \alpha_\phi$. The claim then follows from the
  PSPACE-completeness of QBF.
  %% \medskip  
  %% \begin{tabular}{lll}
  %%   $\phi$ is true & iff &$v_\bot \models \phi$\\
  %%   &iff &$v_\bot, s^* \models \alpha_\phi$.
  %% \end{tabular}
\end{proof}

\end{document}